%% file: main.tex
\documentclass{vldb}

\input{macros}
\input{metadata}

\begin{document}

\input{title_and_authors}

\definecolor{titlecolor}{HTML}{666666}
\definecolor{bulletcolor}{HTML}{666666}

\maketitle

\input{abstract}

\section{Introduction}
\input{introduction-new}

\section{Events in Action}\label{sec:in_action}
\input{cep_in_action}

\section{A query language for CEP}\label{sec:prelim}
\input{preliminaries}

\section{Selection strategies}\label{sec:selectors}
\input{selectors}

\section{Syntactic analysis of CEL}\label{sec:language-prop}
\input{language-prop}

\section{A computational model for CEL}\label{sec:computational-models}
\input{automata}


\section{Algorithms for evaluating CEA}\label{sec:evaluation}
\input{evaluation}

\section{Experimental evaluation}\label{sec:experiments}
\input{experiments}

\section{Future work}\label{sec:conclusions}
\input{conclusions}


\bibliographystyle{abbrv}
\bibliography{references}

%
\onecolumn
\appendix

\section{Proofs of Section~\ref{sec:selectors}}
\input{appendix-selectors}

\section{Proofs of Section~\ref{sec:language-prop}}
\input{appendix-language-prop}

\section{Proofs of Section~\ref{sec:computational-models}}
\input{appendix-automata}


\section{Proofs of Section~\ref{sec:evaluation}}
\input{appendix-evaluation}

\section{Queries for SASE and EsperTech}
\input{appendix-queries}

\end{document}

%% file: macros.tex
\usepackage[utf8]{inputenc}
\usepackage{amsmath}
\usepackage{wrapfig}
\usepackage{amssymb}
\usepackage{stmaryrd}
\usepackage{thmtools}
\usepackage{epsfig}
\usepackage{bbold}
\usepackage{multicol}
\usepackage{multirow}
\usepackage{MnSymbol}
\usepackage{tikz}
\usepackage{pgfplots}

\pgfplotsset{compat=1.12}
\usetikzlibrary{arrows,decorations.pathmorphing,backgrounds,positioning,fit,calc,automata}
\usetikzlibrary{trees}
\usetikzlibrary{shapes}
\usetikzlibrary{chains}
\usetikzlibrary{patterns}
\usepackage{calc}
\usepackage[textwidth=2cm,textsize=small]{todonotes}
\usepackage{algorithm}
\usepackage[noend]{algpseudocode}
\usepackage{paralist}
\usepackage{enumitem}
\usepackage{balance}
\usepackage{footnote}
\makesavenoteenv{tabular}

\setitemize{noitemsep, topsep=6pt, parsep=1pt, partopsep=6pt, leftmargin=12pt, itemindent=0pt}

\usepackage{xpatch}
\usepackage{textcase}
\makeatletter
\xpatchcmd{\@sect}{\uppercase}{\MakeTextUppercase}{}{}
\xpatchcmd{\@sect}{\uppercase}{\MakeTextUppercase}{}{}
\makeatother

\newcommand{\dom}{\operatorname{dom}}
\newcommand{\cO}{\mathcal{O}}

\newcommand{\dset}{\mathbf{D}}
\newcommand{\aset}{\mathbf{A}}
\newcommand{\pset}{\mathbf{P}}
\newcommand{\uset}{\mathbf{U}}

\newcommand{\vset}{\mathbf{X}}

\newcommand{\tuples}{\operatorname{tuples}}

\newcommand{\att}{\operatorname{att}}

\newcommand{\vdef}{\operatorname{vdef}}
\newcommand{\vdefplus}{\operatorname{vdef}_{+}}
\newcommand{\var}{\operatorname{var}}
\newcommand{\type}{\operatorname{type}}

\newcommand{\vi}{\operatorname{bound}}

\newcommand{\leqnext}{\leq_{\operatorname{next}}}
\newcommand{\leqlast}{\leq_{\operatorname{last}}}
\newcommand{\lnext}{<_{\operatorname{next}}}
\newcommand{\llast}{<_{\operatorname{last}}}

\newcommand{\TRUE}{\mathtt{TRUE}}
\newcommand{\FALSE}{\mathtt{FALSE}}

\newcommand{\SEL}{\mathtt{SEL}}

\newcommand{\as}{~\mathtt{AS}~}
\newcommand{\FILTER}{~\mathtt{FILTER}~}
\newcommand{\cor}{~\mathtt{OR}~}

\newcommand{\tcor}{\mathtt{OR}}
\newcommand{\sq}{\,;\,}
\newcommand{\ks}{+}

\newcommand{\STRICT}{\mathtt{STRICT}}
\newcommand{\NEXT}{\mathtt{NXT}}
\newcommand{\LAST}{\mathtt{LAST}}
\newcommand{\MAX}{\mathtt{MAX}}

\newcommand{\sFILTER}{\mathtt{FILTER}}
\newcommand{\scor}{\mathtt{OR}}

\newcommand{\sSTRICT}{\mathtt{STRICT}}
\newcommand{\sNEXT}{\mathtt{NXT}}
\newcommand{\sLAST}{\mathtt{LAST}}
\newcommand{\sMAX}{\mathtt{MAX}}

\newcommand{\cdotlor}{\mathrel{\ooalign{$\lor$\cr\hidewidth\raise.225ex\hbox{$\cdot\mkern3.3mu$}\cr}}}

\newcommand{\cA}{\mathcal{A}}

\newcommand{\cP}{\mathcal{P}}

\newcommand{\cT}{\mathcal{T}}
\newcommand{\cU}{\mathcal{U}}
\newcommand{\cR}{\mathcal{R}}

\newcommand{\bbN}{\mathbb{N}}

\newcommand{\sem}[1]{{\lsem{}{#1}\rsem}}

\newcommand{\trans}[2][]{\raisebox{-1pt}[10pt][0pt]{$\overset{#2}{\underset{^{#1}}{\raisebox{0pt}[3pt][0pt]{$\relbar\mspace{-8mu}\longrightarrow$}}}$}}
\newcommand{\amark}{\bullet}
\newcommand{\umark}{\circ}
\newcommand{\rmatch}{\operatorname{events}}
\newcommand{\run}{\operatorname{Run}}
\newcommand{\TPO}{\operatorname{TPO}}
\newcommand{\types}[1]{{{#1}\operatorname{-types}}}

\newcommand{\alist}{\operatorname{list}}
\newcommand{\aactive}{\operatorname{active}}
\newcommand{\qinit}{\operatorname{q_{init}}}

\newcommand{\aold}{{\tt old}}

\newcommand{\mNode}{\operatorname{Node}}
\newcommand{\myield}{{\tt yield}}
\newcommand{\mlazycopy}{{\tt lazycopy}}

\newcommand{\madd}{{\tt add}}
\newcommand{\mappend}{{\tt append}}
\newcommand{\mposition}{{\tt position}}
\newcommand{\mlist}{{\tt list}}
\newcommand{\mbegin}{{\tt begin}}
\newcommand{\mnext}{{\tt next}}
\newcommand{\matend}{{\tt atEnd}}

\newcommand{\mpop}{{\tt pop}}
\newcommand{\mpushblack}{{\tt push\text{-}black}}
\newcommand{\mpushwhite}{{\tt push\text{-}white}}
\newcommand{\mpopwhites}{{\tt pop\text{-}whites}}
\newcommand{\mempty}{{\tt empty}}
\newcommand{\menum}{{\tt Output}}

\newtheorem{theorem}{Theorem}
\newtheorem{lemma}{Lemma}

\newtheorem{proposition}{Proposition}

\newcommand*\circled[1]{\tikz[baseline=(char.base)]{
		\node[shape=circle,draw,inner sep=1pt] (char) {#1};}}

\newcommand*\circledblack[1]{\tikz[baseline=(char.base)]{
		\node[shape=circle,draw,inner sep=1pt, fill=black, text=white] (char) {\textbf{#1}};}}



%% file: metadata.tex
\vldbTitle{Foundations of Complex Event Processing}
\vldbAuthors{Marco Bucchi, Alejandro Grez, Cristian Riveros and Martín Ugarte}
\vldbVolume{11}
\vldbNumber{2}
\vldbYear{2018}
\vldbDOI{https://doi.org/TBD}

%% file: title_and_authors.tex

\title{Foundations of Complex Event Processing}



\numberofauthors{3}

\author{
       \alignauthor
       Marco Bucchi\\
       \affaddr{PUC Chile}
       \email{mabucchi@uc.cl}
       \and
       \alignauthor
       Alejandro Grez\\
       \affaddr{PUC Chile}
       \email{ajgrez@uc.cl}
       \and
       \alignauthor
       Cristian Riveros\\
       \affaddr{PUC Chile}
       \email{cristian.riveros@uc.cl}
       \and
       \alignauthor
       Martín Ugarte\\
       \affaddr{Université Libre de Bruxelles}
       \email{mugartec@ulb.ac.be}
}


%% file: abstract.tex

\begin{abstract}
Complex Event Processing (CEP) has emerged as the unifying field for technologies that require processing and correlating distributed data sources in real-time. CEP finds applications in diverse domains, which has resulted in a large number of proposals for expressing and processing complex events. However, existing CEP languages lack from a clear semantics, making them hard to understand and generalize. Moreover, there are no general techniques for evaluating CEP query languages with clear performance guarantees.

In this paper we embark on the task of giving a rigorous and efficient framework to CEP. We propose a formal language for specifying complex events, called CEL, that contains the main features used in the literature and has a denotational and compositional semantics. We also formalize the so-called selection strategies, which had only been presented as by-design extensions to existing frameworks. With a well-defined semantics at hand, we study how to efficiently evaluate CEL for processing complex events in the case of unary filters. We start by studying the syntactical properties of CEL and propose rewriting optimization techniques for simplifying the evaluation of formulas. Then, we introduce a formal computational model for CEP, called complex event automata (CEA), and study how to compile CEL formulas into CEA. Furthermore, we provide efficient algorithms for evaluating CEA over event streams using constant time per event followed by constant-delay enumeration of the results. By gathering these results together, we propose a framework for efficiently evaluating CEL with unary filters. Finally, we show experimentally that this framework consistently outperforms the competition, and even over trivial queries can be orders of magnitude more efficient.
\end{abstract}
\pagebreak


%% file: introduction-new.tex

Complex Event Processing (CEP) has emerged as the unifying field of technologies for detecting situations of interest under high-throughput data streams. 
In scenarios like Network Intrusion Detection~\cite{mukherjee1994}, Industrial Control Systems~\cite{groover2007} or Real-Time Analytics~\cite{sahay2008}, CEP systems aim to efficiently process arriving data, giving timely insights for implementing reactive responses to complex events.

Prominent examples of CEP systems from academia and industry include SASE \cite{SASE}, EsperTech~\cite{EsperTech}, Cayuga \cite{cayuga}, TESLA/T-Rex~\cite{TESLA,Cugola:2012}, among others (see \cite{cugola2012} for a survey).
The main focus of these systems has been in practical issues like scalability, fault tolerance, and distribution, with the objective of making CEP systems applicable to real-life scenarios. Other design decisions, like query languages, are generally adapted to match computational models that can efficiently process data (see for example \cite{SASEcomplexity}). This has produced new data management and optimization techniques, generating promising results in the area \cite{SASE,EsperTech}.

Unfortunately, as has been claimed several times~\cite{galton2002two,zimmer1999semantics,TESLA,DEBStutorial} CEP query languages lack from a simple and denotational semantics, which makes them difficult to understand, extend, or generalize. 
The semantics of several languages are defined either by examples~\cite{RAPIDE,amit,RACED}, or by intermediate computational models~\cite{SASE,nextCEP,distCED}. Although there are frameworks that introduce formal semantics (e.g. \cite{cayuga,CEDR,PBCED,TESLA,anicic2010rule}), they do not meet the expectations to pave the foundations of CEP languages. For instance, some of them are too complicated (e.g. sequencing is combined with filters), have unintuitive behavior (e.g. sequencing is non-associative), or are severely restricted (e.g. nesting operators is not supported). One symptom of this problem is that iteration, which is a fundamental operator in CEP, has not yet been defined successfully as a compositional operator. Since iteration is difficult to define and evaluate, it is usually restricted by not allowing nesting or reuse of variables~\cite{SASE,cayuga}. Thus, without a formal and natural semantics, the languages for CEP are in general cumbersome.

The lack of simple denotational semantics makes query languages also difficult to evaluate. A common factor in CEP system is to find sophisticated heuristics~\cite{SASEcomplexity, TESLA} that cannot be replicated in other frameworks.
Further, optimization techniques are usually proposed at the architecture level~\cite{GEM,cayuga,distCED}, preventing from a unifying optimization theory.
In this direction, many CEP frameworks use automata-based models~\cite{cayuga,CEDR,PBCED} for query evaluation. 
However, these models are usually complicated~\cite{distCED,nextCEP}, informally defined~\cite{cayuga} or non-standard~\cite{TESLA,SASEautomata}. In practice this implies that, although finite state automata is a recurring approach in CEP,  
there is no general evaluation strategy with clear performance guarantees.

Given this scenario, the goal of this paper is to give solid foundations to CEP systems in terms of query language and query evaluation. Towards these goals, we first provide a formal language that allows for expressing the most common features of CEP systems, namely sequencing, filtering, disjunction, and iteration. We introduce complex event logic (CEL for short), a logic with well-defined compositional and denotational semantics. We also formalize the so-called \emph{selection strategies}, an important notion of CEP that is usually discussed directly~\cite{SASEcomplexity,cayuga} or indirectly~\cite{CEDR} in the literature but has not been formalized at the language level.

Then, we embark on the design of a formal framework for CEL evaluation. This framework must consider three main building blocks for the efficient evaluation of CEL: (1) syntactic techniques for rewriting CEL queries, (2) a well-defined intermediate evaluation model, and (3) efficient translation and algorithms to evaluate this model. 
Regarding the rewriting techniques, we study the structure of CEL by introducing the notions of well-formed and safe formulas, and show that these restrictions are relevant for query evaluation. Further, we give a general result on rewriting CEL formulas into the so-called LP-normal form, a normal form for dealing with unary filters. 
For the intermediate evaluation model, we introduce a formal computational model for the regular fragment of CEL, called \emph{complex event automata} (CEA). We show that this model is closed under I/O-determinization and provide translations for any CEL formula into CEA.
More important, we show an efficient algorithm for evaluating CEA with clear performance guarantees: constant time per tuple followed by constant-delay enumeration of the output. 
We bring together our results to present a formal framework for evaluating CEL.
Towards the end of the paper, we show an experimental evaluation of our framework with the leading CEP systems in the area.
Our experiments shows that our framework outperforms previous systems by order of magnitudes in terms of processing time and memory consumption.

\smallskip

\noindent \textbf{Related work.}
Active Database Systems (ADSMS) and Data Stream Management Systems (DSMS) are solutions for processing data streams and they are usually associated with CEP systems. Both technologies, and specially DSMS, are designed for executing relational queries over dynamic data~\cite{Chen:2000,Abadi:2003,Arasu:2003}.
In contrast, CEP systems see data streams as a sequence of data events where the arrival order is the main guide for finding patterns inside streams (see~\cite{cugola2012} for a comparison between ADSMS, DSMS, and CEP). 
In particular, DSMS query languages (e.g. CQL~\cite{Arasu:2006}) are incomparable with our framework since they do not focus on CEP operators like sequencing and iteration.

Query languages for CEP are usually divided into three approaches~\cite{cugola2012, DEBStutorial}: logic-based, tree-based and automata-based models. Logic-based models have their roots in temporal logic or event calculus, and usually have a formal, declarative semantics~\cite{anicic2010rule,artikis2015event,chesani2010logic} (see \cite{artikis2012logic} for a survey).
However, this approach does not include iteration as an operator or they do not model the output explicitly. 
Furthermore, their evaluation techniques rely on logic inference mechanisms which are radically different from our approach.
Tree-based models~\cite{mei2009zstream,liu2011cube,EsperTech} have also been used for CEP but their language semantics is usually non-declarative and their evaluation techniques are based on cost-models, similar to relational database systems. 

Automata-based models are the closest approach to the techniques used in this paper. Most proposals (e.g. SASE\cite{SASEautomata}, NextCEP\cite{nextCEP}, DistCED\cite{distCED}) do not rely in a denotational semantics; their output is defined by intermediate automata models. This implies that either iteration cannot be nested \cite{SASEautomata} or its semantics is confusing~\cite{nextCEP}.
Other proposals~(e.g. CEDR\cite{CEDR}, TESLA\cite{TESLA}, PBCED\cite{PBCED}) are defined with a formal semantics but they do not include iteration. 
An exception to this is Cayuga\cite{cayuga2} but its language does not allow the reuse of variables and the sequencing operator is non-associative, which derives in a cumbersome semantics.
Our framework is comparable to these systems, but provides a well-defined language that is compositional, allowing arbitrary nesting of operators. Moreover, we present the first evaluation of CEP queries that guarantees constant time per event and constant-delay enumeration of the output. We show experimentally that this vastly improves performance.

Finally, there has been some research in theoretical aspects of CEP like, for instance, in axiomatization of temporal models~\cite{white2007next}, privacy~\cite{he2011complexity}, and load shedding~\cite{he2013load}. This literature does not study the semantics and evaluation of CEP and, therefore, is orthogonal to our work.
\smallskip

\noindent{\bf Organization.} We give an intuitive introduction to CEP and our framework in Section~\ref{sec:in_action}. In Section~\ref{sec:prelim} and~\ref{sec:selectors} we formally present our logic and selection strategies. The syntactic structure of the logic is studied in Section~\ref{sec:language-prop}. The computational model is studied in Section~\ref{sec:computational-models} where we also show how to compile formulas into automata. Section~\ref{sec:evaluation} presents our algorithms for  efficient evaluation of automata. Section~\ref{sec:experiments} puts all the results in perspective and shows our experimental evaluation of the framework.
Future work is finally discussed in Section~\ref{sec:conclusions}. Due to space limitations all proofs are deferred to the appendix.

%% file: cep_in_action.tex

We start by presenting the main features and challenges of CEP. The examples used in this section will also serve throughout the paper as running examples.\par

In a CEP setting, events arrive in a streaming fashion to a system that must detect certain \emph{patterns}~\cite{cugola2012}. For the purpose of illustration assume there is a stream produced by wireless sensors positioned in a farm, whose main objective is to detect fires. As a first scenario, assume that there are three sensors, and each of them can measure both temperature (in Celsius degrees) and relative humidity (as the percentage of vapor in the air). Each sensor is assigned an id in $\{0,1,2\}$. The \emph{events} produced by the sensors consist of the id of the sensor and a measurement of temperature or humidity. In favor of brevity, we write $T(id, tmp)$ for an event reporting temperature $tmp$ from sensor with id $id$, and similarly $H(id, hum)$ for events reporting humidity. Figure~\ref{fig:stream} depicts such a stream: each column is an event and the \emph{value} row is the temperature or humidity if the event is of type $T$ or $H$, respectively.\par

\begin{figure}
	\centering{
		{\small
			\begin{tabular}{|c|c|c|c|c|c|c|c|c|c|c}\hline
				type  &$H$&$T$&$H$&$H$&$T$&$T$&$T$&$H$&$H$ & \ldots \\ \hline
				$id$  & 2 & 0 & 0 & 1 & 1 & 0 & 1 & 1 & 0 & \multirow{2}{*}{\ldots} \\
				value & 25 & 45& 20& 25& 40& 42& 25& 70& 18\\ \hline
				index & 0 & 1 & 2 & 3 & 4 & 5 & 6 & 7 & 8 & \ldots \\ \hline
			\end{tabular}}
			\caption{A stream $S$ of events measuring temperature and humidity. ``value'' contains degrees and humidity for $T$- and $H$- events, respectively.}\label{fig:stream}
			\vspace{-.4cm}
		}
	\end{figure}

The patterns to be deetcted are generally specified by domain experts. For the sake of illustration, assume that the position of sensor~$0$ is particularly prone to fires, and it has been detected that a temperature measurement above 40 degrees Celsius followed by a humidity measurement of less than 25\% represents a fire with high probability. Let us intuitively explain how a domain expert can express this as a pattern (also called a \emph{formula}) in our framework:
\begin{multline*}
\varphi_1=(T \as x \sq H \as y)\FILTER (x.tmp > 40\ \land\\
y.hum <= 25\ \land x.id=0\ \land y.id=0)
\end{multline*}
This formula is asking for two events, one of type temperature ($T$) and one of type humidity ($H$). The events of type temperature and humidity are given names $x$ and $y$, respectively, and the two events are filtered to select only those pairs $(x,y)$ representing a high temperature followed by a low humidity measured by sensor 0. 

What should be the result of evaluating $\varphi_1$ over the stream in Figure~\ref{fig:stream}?
A first important remark is that event streams are noisy in practice, and one does not expect the events matching a formula to be \emph{contiguous} in the stream. Then, a CEP engine needs to be able to dismiss irrelevant events. The semantics of the \emph{sequencing} operator~(;) will thus allow for arbitrary events to occur in between the events of interest. A second remark is that in CEP the set of events matching a pattern, called a \emph{complex event}, is particularly relevant to the end user. Every time that a formula matches a portion of the stream, the final user should retrieve the events that compose that portion of the stream. This means that the evaluation of a formula over a stream should output a set of \emph{complex events}. In our framework, each complex event will be the set of indexes (stream positions) of the events that witness the matching of a formula.
Specifically, let $S[i]$ be the event at position $i$ of the stream $S$. What we expect for the output of formula $\varphi_1$ is a set of pairs $\{i,j\}$ such that $S[i]$ is of type $T$, $S[j]$ is of type $H$, $i<j$, and they satisfy the conditions expressed after the FILTER. By inspecting Figure~\ref{fig:stream}, we can see that the pairs satisfying these conditions are $\{1,2\}$, $\{1,8\}$, and $\{5,8\}$.\par

Formula $\varphi_1$ illustrates in a simple way the two most elemental features of CEP, namely \emph{sequencing} and \emph{filtering}~\cite{cugola2012,Arasu:2003,SASEcomplexity,Abadi:2003,buchmann2009}. But although it detects a set of possible fires, it restricts the \emph{order} in which the two events must occur, namely the temperature must be measured before the humidity. Naturally, this could prevent the detection of a fire in which the humidity was measured first. This motivates the introduction of \emph{disjunction}, another common feature in CEP engines~\cite{cugola2012,Arasu:2003}. To illustrate, we extend $\varphi_1$ by allowing events to appear in arbitrary order.
\begin{multline*}
	\varphi_2=[(T \as x \sq H \as y)\cor (H \as y \sq T \as x)]\FILTER\\
	(x.tmp > 40\ \land y.hum <= 25 \land x.id=0\ \land y.id=0)
\end{multline*}
The $\tcor$ operator allows for any of the two patterns to be matched, and the filter is applied as in $\varphi_1$. The result of evaluating $\varphi_2$ over the stream $S$ of Figure~\ref{fig:stream} is the same as evaluating $\varphi_1$ over $S$ plus the complex event $\{2,5\}$.\par

The previous formulas show how CEP systems raise alerts when a certain complex event occurs. However, from a wider scope the objective of CEP is to retrieve information of interest from streams. For example, assume that we want to see how does temperature change in the location of sensor~1 when there is an increase of humidity. A problem here is that we do not know a priori the amount of temperature measurements; we need to capture an unbounded amount of events. 
The \emph{iteration} operator~\cite{cugola2012,Arasu:2003} (also known as Kleene closure~\cite{gyllstrom2008supporting}) is introduced in most CEP frameworks for solving this problem. This operator introduces many difficulties in the semantics of CEP languages. For example, since events are not required to occur contiguously, the nesting of $+$ is particularly tricky and most frameworks simply disallow this (see \cite{SASE,Arasu:2006,cayuga}). Coming back to our example, the formula for measuring temperatures whenever an increase of humidity is detected by sensor $1$ is:
\begin{multline*}
	\varphi_3=[H\as x \sq (T \as y\FILTER y.id=1)+ \sq H \as z]\\
	\qquad \FILTER(x.hum < 30\ \land z.hum > 60 \land x.id=z.id=1)
\end{multline*}
Intuitively, variables $x$ and $z$ witness the increase of humidity from less than 30\% to more than 60\%, and $y$ captures temperature measures between $x$ and $z$. 
Note that the filter for $y$ is included inside the $+$ operator. Some frameworks allow to declare variables inside a $+$ and filter them outside that operator (e.g. \cite{SASE}). Although it is possible to define the semantics for that syntax, this form of filtering makes the definition of nesting $+$ difficult. 
Another semantic subtlety of the + operator is the association of $y$ to an event.
Given that we want to match the event $(T \as y\FILTER y.id=1)$ an unbounded number of times: how should the events associated to $y$ occur in the complex events generated as output?
Associating different events to the same variable during evaluation has proven to make the semantics of CEP languages cumbersome. In Section~\ref{sec:prelim}, we introduce a natural semantics that allows nesting $+$ and associate variables (inside $+$ operators) to different events across repetitions.\par

Let us now explain the semantics of $\varphi_3$ over stream $S$ (Figure~\ref{fig:stream}). 
The only two humidity events satisfying the top-most filter are $S[3]$ and $S[7]$ and the temperature events between these two are $S[4]$ and $S[6]$. As expected, the complex event $\{3,4,6,7\}$ is part of the output. However, there are also other complex events in the output. Since, as discussed, there might be irrelevant events between relevant ones, the semantics of $+$ must allow for \emph{skipping} arbitrary events.
This implies that the complex events $\{3,6,7\}$ and $\{3,4,7\}$ are also part of the output. \par

The previous discussion raises an interesting question: are users interested in receiving all complex events? Are some complex events more informative than others? Coming back to the output of $\varphi_3$ ($\{3,6,7\}$, $\{3,4,7\}$ and $\{3,4,6,7\}$), one can easily argue that the largest complex event is more informative than others since all events are contained in it. A more complicated analysis deserves the complex events output by $\varphi_1$. In this scenario, the pairs that have the same second component (e.g., $\{1,8\}$ and $\{5,8\}$) represent a fire occurring at the same place and time, so one could argue that only one of the two is necessary. 
For cases like above, it is common to find CEP systems that restrict the output by using so-called \emph{selection strategies} (see for example ~\cite{SASE,SASEcomplexity,TESLA}).
Selection strategies are a fundamental feature of CEP. Unfortunately, they have only been presented as heuristics applied to particular computational models, and thus their semantics given by an algorithm and hard to understand. A special mention deserves the \emph{next} selection strategy (called skip-till-next-match in~\cite{SASE,SASEcomplexity}) which models the idea of outputting only those complex events that can be generated without skipping relevant events. Although the semantics of \emph{next} has been mentioned in previous papers (e.g \cite{CEDR}), it is usually underspecified~\cite{SASE,SASEcomplexity} or complicates the semantics of other operators~\cite{cayuga}. In Section~\ref{sec:selectors}, we formally define a set of selection strategies including \emph{next}.\par

Before formally presenting our framework, we illustrate one more common feature of CEP, namely \emph{correlation}. Correlation is introduced by filtering events with predicates that involve more than one event. For example, consider that we want to see how does temperature change at some location whenever there is an increase of humidity, like in $\varphi_3$. What we need is a pattern where all the events are produced by the same sensor, but that sensor is not necessarily sensor~1. This is achieved by the following pattern:
\begin{multline*}
\varphi_4=[H\as x;(T \as y\FILTER y.id=x.id)+;H \as z]\\
\FILTER(x.hum < 30\ \land z.hum > 60 \land x.id=z.id)
\end{multline*}
Notice that here the filters contain the binary predicates $x.id=y.id$ and $x.id=z.id$ that force all events to have the same id. Although this might seem simple, the evaluation of formulas that correlate events introduces new challenges. Intuitively, formula $\varphi_4$ is more complicated because the value of $x$ must be remembered and used during evaluation in order to compare it with future incoming events.
If the reader is familiar with automata theory~\cite{HopcroftU79,sakarovitch2009elements}, this behavior is clearly not ``regular'' and it will not be captured by a finite state model.
In this paper, we study and characterize the regular part of CEP-systems. 
Therefore, from Section~\ref{sec:computational-models} to Section~\ref{sec:experiments} we focus on formulas without correlation. As we will see, the formal analysis of this fragment already presents important challenges, which is the reason why we defer the analysis of formulas like $\varphi_4$ for future work. It is important to mention that the semantics of our language (including selection strategies) is general and includes correlation.\par

%% file: preliminaries.tex


Having discussed and illustrated the common operators and features of CEP, we proceed to formally introduce CEL (Complex Event Logic), our pattern language for capturing complex events.
\smallskip

\noindent {\bf Schemas, Tuples and Streams.}
Let $\aset$ be a set of \emph{attribute names} and $\dset$ be a set of values. A database schema $\cR$ is a finite set of relation names, where each relation name $R \in \cR$ is associated to a tuple of attributes in $\aset$ denoted by $\att(R)$. If $R$ is a relation name, then an $R$-tuple is a function $t:\att(R) \rightarrow \dset$. We say that the type of an $R$-tuple $t$ is $R$, and denote this by $\type(t)=R$. For any relation name~$R$, $\tuples(R)$ denotes the set of all possible $R$-tuples, i.e., $\tuples(R)=\{t:\att(R) \rightarrow \dset\}$.
Similarly, for any database schema $\cR$, $\tuples(\cR)=\bigcup_{R \in \cR}\tuples(R)$. 

Given a schema $\cR$, an $\cR$\textit{-stream} $S$ is an infinite sequence $S = t_0 t_1 \ldots$ where $t_i \in \tuples(\cR)$. When $\cR$ is clear from the context, we refer to $S$ simply as a stream. Given a stream $S = t_0 t_1 \ldots$ and a position $i \in \bbN$, the $i$-th element of $S$ is denoted by $S[i]=t_i$, and the sub-stream $t_{i}t_{i+1} \ldots$ of $S$ is denoted by $S_i$. 
Note that we consider in this paper that the time of each event is given by its index, and defer a more elaborated time model (like~\cite{white2007next}) for future work.

Let $\vset$ be a set of variables. 
Given a schema $\cR$, a predicate of arity $n$ is an $n$-ary relation $P$ over
$\tuples(\cR)$, i.e. $P \subseteq \tuples(\cR)^n$. 
An atom is an expression  $P(x_1, \ldots, x_n)$ (or $P(\bar{x})$) where $P$ is an $n$-ary predicate and $\bar{x} = x_1, \ldots, x_n \in \vset$. 
For example, $P(x):=x.hum < 30$ is an atom and $P$ is the predicate of all tuples that have a humidity attribute with value less than $30$. 
In this paper, we consider a fixed set of predicates, denoted by $\pset$. Moreover, we assume that $\pset$ is closed under intersection, union, and complement, and $\pset$ contains the predicate $P_R(x) := \type(x)=R$ for checking if a tuple is an $R$-tuple for everv $R \in \cR$. 
\smallskip

\noindent \textbf{CEL syntax.} Now we proceed to give the syntax of what we call the \emph{core} of CEL (core-CEL for short), a logic inspired by the operations described in the previous section. This language features the most essential CEP features.
The set of formulas in core-CEL, or core formulas for short, is given by the following grammar:
\[\varphi \; := \; R \as x \ \mid \  \varphi \FILTER P(\bar{x}) \ \mid \ \varphi \cor \varphi \ \mid \ \varphi\sq\varphi \ \mid \ \varphi\ks \]
Here $R$ is a relation name, $x$ is a variable in $\vset$ and $P(\bar{x})$ is an atom in $\pset$. 
All formulas in Section~\ref{sec:in_action} are CEL formulas. Furthermore, formulas of the form $\varphi \FILTER (P(\bar{x}) \wedge Q(\bar{y}))$ or $\varphi \FILTER (P(\bar{x}) \vee Q(\bar{y}))$ are used as syntactic sugar for $(\varphi \FILTER P(\bar{x})) \FILTER Q(\bar{y})$ or $(\varphi \FILTER P(\bar{x})) \cor$ $(\varphi \FILTER Q(\bar{y}))$, respectively. As opposed to existing frameworks, we do not restrict the use of operators or variables, allowing for arbitrary nestings (in particular of $+$).
\smallskip

\noindent \textbf{CEL semantics.} We proceed to define the semantics of core formulas, for which we need to introduce some further notation. A \emph{complex event} $C$ is defined as a non-empty and finite set of indices. As mentioned in Section~\ref{sec:in_action}, a complex event contains the positions of the events that witness the matching of a formula over a stream, and moreover, they are the final output of evaluating a formula over a stream. We denote by $|C|$ the size of $C$ and by $\min(C)$ and $\max(C)$ the minimum and maximum element of $C$, respectively. 
Given two complex events $C_1$ and $C_2$, $C_1 \cdot C_2$ denotes the \emph{concatenation} of two complex events, that is, $C_1 \cdot C_2 := C_1 \cup C_2$ whenever $\max(C_1) < \min(C_2)$ and is undefined otherwise. \par

In core-CEL formulas, variables are second class citizens because they are only used to filter and select particular events, i.e. they are not retrieved as part of the output. As examples in Section~\ref{sec:in_action} suggest, we are only concerned with finding the events that compose the complex events, and not which position corresponds to which variable. The reason behind this is that the operator $+$ allows for repetitions, and therefore variables under a (possibly nested) $+$ operator would need to have a special meaning, particularly for filtering.
This discussion motivates the following definitions. 
Given a formula $\varphi$ we denote by $\var(\varphi)$ the set of all variables mentioned in $\varphi$ (including its predicates), and by $\vdef(\varphi)$ all variables defined in $\varphi$ by a clause of the form $R \as x$. Furthermore, $\vdefplus(\varphi)$ denotes all variables in $\vdef(\varphi)$ that are defined outside the scope of all $\ks$ operators. For example, for $\varphi = (T \as x \sq (H \as y)\ks) \FILTER z.id = 1$ we have that $\var(\varphi) = \{x, y, z\}$, $\vdef(\varphi) = \{x,y\}$, and $\vdefplus(\varphi) = \{x\}$. 
Finally, a valuation is a function $\nu:\vset \rightarrow \bbN$. Given a finite set of variables $U \subseteq \vset$ and two valuations $\nu_1$ and $\nu_2$, the valuation $\nu_1[\nu_2 \slash U]$ is defined by $\nu_1[\nu_2 \slash U](x) = \nu_2(x)$ if $x \in U$ and by $\nu_1[\nu_2 \slash U](x) = \nu_1(x)$ otherwise. 

We are ready to define the semantics of a core-CEL formula $\varphi$. Given a complex event $C$ and a stream $S$, we say that $C$ is in the evaluation of $\varphi$ over $S$ under valuation $\nu$ ($C \in \sem{\varphi}(S, \nu)$) if one of the following conditions holds:
\begin{itemize}
	\item $\varphi=R \as x$, $C = \{\nu(x)\}$, and $\type(S[\nu(x)]) = R$.
	\item $\varphi = \rho \FILTER P(x_1, \ldots, x_n)$ and both $C \in \sem{\rho}(S, \nu)$ and $(S[\nu(x_1)],\ldots, S[\nu(x_n)]) \in P$ hold.
	\item $\varphi=\rho_1 \cor \rho_2$ and $C \in \sem{\rho_1}(S, \nu)$ or $C \in \sem{\rho_2}(S, \nu)$.
	\item $\varphi=\rho_1 \sq \rho_2$ and there exist complex events $C_1 \in \sem{\rho_1}(S, \nu)$ and $C_2 \in \sem{\rho_2}(S, \nu)$ such that $C = C_1 \cdot C_2$.
	\item $\varphi=\rho\ks$ and there exists $\nu'$ such that $C \in \sem{\rho}(S, \nu[\nu' \slash U])$ or  $C \in \sem{\rho\sq\rho\ks}(S, \nu[\nu' \slash U])$, where $U = \vdefplus(\rho)$.
\end{itemize}
There are a couple of important remarks here. First, the valuation $\nu$ can be defined over a superset of the variables mentioned in the formula. This is important for sequencing (;) because we require the complex events from both sides to be produced with the same valuation. Second, when we evaluate a subformula of the form $\rho+$, we \emph{carry} the value of variables defined outside the subformula. For example, the subformula $(T\as y \FILTER y.id=x.id)+$ of $\varphi_4$ does not define the variable $x$. However, from the definition of the semantics we see that $x$ will be \emph{already assigned} (because $R\as x$ occurs outside the subformula). This is precisely where other frameworks fail to formalize iteration, as without this construct it is not easy to correlate the variables inside + with the ones outside, as we illustrate with $\varphi_4$.

As previously discussed, in core-CEL variables are just used for comparing attributes with $\sFILTER$, but are not relevant for the final output. In consequence, we say that $C$ belongs to the evaluation of $\varphi$ over $S$ (denoted $C \in \sem{\varphi}(S)$) if there is a valuation $\nu$ such that $C \in \sem{\varphi}(S, \nu)$.
As an example, the complex events presented in Section~\ref{sec:in_action} are indeed the outputs of $\varphi_1$ to $\varphi_3$ over the stream in Figure~\ref{fig:stream}.

%% file: selectors.tex

Matching complex events is a computationally intensive task. As the examples in Section~\ref{sec:in_action} might suggest, the main reason behind this is that the amount of complex events can grow exponentially in the size of the stream, forcing systems to process large numbers of \emph{candidate} outputs. In order to speed up the matching processes, it is common to restrict the set of results \cite{carlson2010resource, SASE, SASEcomplexity}. As we validate in the experimental section, this is required for current CEP systems to work in practice. Unfortunately, most proposals in the literature restrict outputs by introducing heuristics to particular computational models without describing how the semantics are affected. For a more general approach, we introduce \emph{selection strategies} (or \emph{selectors}) as unary operators over core-CEL formulas. 
Formally, we define four selection strategies called strict ($\STRICT$), next ($\NEXT$), last ($\LAST$) and max ($\MAX$). 
$\STRICT$ and $\NEXT$ are motivated by previously introduced operators~\cite{SASE} under the name of \emph{strict-contiguity} and \emph{skip-till-next-match}, respectively. $\LAST$ and $\MAX$ are introduced here as useful selection strategies from a semantic point of view.
We proceed to define each selection strategy below, giving the motivation and formal semantics. 

\smallskip
\noindent \textbf{STRICT.} As the name suggest, $\STRICT$ or strict-contiguity keeps only the complex events that are contiguous in the stream, basically reducing the evaluation problem to that of regular expressions.
To motivate this, recall that formula $\varphi_1$ in Section~\ref{sec:in_action} detects complex events composed by a temperature above 40 degrees Celsius followed by a humidity of less than 25\%.
As already argued, in general one could expect other events between $x$ and $y$. However, it could be the case that this pattern is of interest only if the events occur contiguously in the stream, namely a temperature immediately after a humidity measure. For this purpose, $\STRICT$ reduces the set of outputs selecting only strictly consecutive complex events. Formally, for any CEL formula $\varphi$ we have that $C \in  \sem{\STRICT(\varphi)}(S, \nu)$ holds if $C \in  \sem{\varphi}(S, \nu)$ and for every $i, j \in C$, if $i<k<j$ then $k\in C$ (i.e., $C$ is an interval). In our running example, $\STRICT(\varphi_1)$ would only produce $\{1, 2\}$, although $\{1,8\}$ and $\{5, 8\}$ are also outputs for $\varphi_1$ over~$S$.

\par

\smallskip
\noindent \textbf{NXT.} The second selector, $\NEXT$, is similar to the previously proposed operator  \emph{skip-till-next-match}~\cite{SASE}. The motivation behind this operator comes from a heuristic that consumes a stream skipping those events that cannot participate in the output, but matching patterns in a \emph{greedy} manner that selects only the first event satisfying the next element of the query. In~\cite{SASE} the authors gave the definition informally as  
\begin{quote}
``\emph{a further relaxation is to remove the contiguity requirements: all irrelevant events will be skipped until the next relevant event is read}'' (*).
\end{quote}
In practice, the definition of skip-till-next-match is given by a greedy evaluation algorithm that adds an event to the output whenever a sequential operator is used, or goes as far as possible adding events whenever an iteration operator is used. The fact that the semantics is only defined by an algorithm requires a user to understand the algorithm to write meaningful queries. In other words, this operator speeds up the evaluation by sacrificing the clarity of the semantics



To overcome the above problem, we formalize the intuition behind (*) based on a special order over complex events. As we will see later, this allows to speed up the evaluation process as much as skip-till-next-match while providing clear and intuitive semantics.
Let $C_1$ and $C_2$ be complex events. The symmetric difference between $C_1$ and $C_2$ ($C_1\triangle C_2$) is the set of all elements either in $C_1$ or $C_2$ but not in both. We say that $C_1 \leqnext C_2$ if either $C_1=C_2$ or $\min(C_1 \triangle C_2)\in C_2$. For example, $\{5,8\} \leqnext  \{1,8\}$ since the minimum element in $\{5,8\} \triangle \{1,8\} = \{1,5\}$ is $1$, which is in $\{1,8\}$. Note that this is intuitively similar to skip-till-next-match, as we are selecting the first relevant event. An important property is that the $\leqnext$-relation forms a total order among complex events, implying the existence of a minimum and a maximum over any finite set of complex events.
\begin{lemma}\label{lemma:next-order}
	$\leqnext$ is a total order between complex events. 
\end{lemma}
We can define now the semantics of $\NEXT$: for a CEL formula $\varphi$ we have that $C \in \sem{\NEXT(\varphi)}(S,\nu)$ if $C \in \sem{\varphi}(S, \nu)$ and for every complex event $C' \in \sem{\varphi}(S, \nu)$, if $\max(C) = \max(C')$ then $C \leqnext C'$. In our running example, when evaluating $\varphi_1$ over $S$ we have that $\{1,8\}$ matches $\NEXT(\varphi_1)$ but $\{5,8\}$ does not. Furthermore, $\{3,4,6,7\}$ matches $\NEXT(\varphi_4)$ while $\{3,4,7\}$ and $\{3,6,7\}$ do not. Note that we compare outputs with respect to $\leqnext$ that have the same final position. This way, complex events are discarded only when there is a \emph{preferred} complex event triggered by the same last event.

\smallskip
\noindent \textbf{LAST.} The $\NEXT$ selector is motivated by the computational benefit of skipping irrelevant events in a greedy fashion. However, from a semantic point of view it might not be what a user wants. For example, if we consider again $\varphi_1$ and stream $S$ (Section~\ref{sec:in_action}), we know that every complex event in $\NEXT(\varphi_1)$ will have event $1$. In this sense, the $\NEXT$ strategy selects the \emph{oldest} complex event for the formula. We argue here that a user might actually prefer the opposite, i.e. the most recent explanation for the matching of a formula. This is the idea captured by $\LAST$. Formally, the $\LAST$ selector is defined exactly as $\NEXT$, but changing the order $\leqnext$ by $\leqlast$: if $C_1$ and $C_2$ are two complex events, then $C_1 \leqlast C_2$ if either $C_1=C_2$ or $\max(C_1 \triangle C_2)\in C_2$. For example, $\{1,8\} \leqlast  \{5,8\}$. In our running example, $\LAST(\varphi_1)$ would select the \emph{most recent} temperature and humidity that explain the matching of $\varphi_1$ (i.e. $\{5, 8\}$), which might be a better explanation for a possible fire. Surprisingly, we show in Section~\ref{sec:evaluation} that $\LAST$ enjoys the same good computational properties as $\NEXT$.

\smallskip
\noindent \textbf{MAX.} A more ambitious selection strategy is to keep all the maximal complex events in terms of set inclusion. This corresponds to obtaining those complex events that are \emph{as informative as possible}, which could be naturally more useful for end users. Formally, given a CEL formula $\varphi$ we say that $C \in \sem{\MAX(\varphi)}(S, \nu)$ holds iff $C \in \sem{\varphi}(S, \nu)$ and for all $C' \in \sem{\varphi}(S, \nu)$, if $\max(C) = \max(C')$ then $C \subseteq C'$. Coming back to our example $\varphi_1$, the $\MAX$ selector will output both $\{1,8\}$ and $\{5,8\}$, given that both complex events are maximal in terms of set inclusion. On the contrary, formula $\varphi_3$ produced $\{3,6,7\}$, $\{3,4,7\}$, and $\{3,4,6,7\}$. Then, if we evaluate $\MAX(\varphi_3)$ over the same stream, we will obtain only $\{3,4,6,7\}$ as output, which is the maximal complex event. It is interesting to note that if we evaluate both $\NEXT(\varphi_3)$ and $\LAST(\varphi_3)$ over the stream we will also get $\{3,4,6,7\}$ as the only output, illustrating that $\NEXT$ and $\LAST$ also yield complex events with maximal information.\par

We have formally presented the foundations of a language for recognizing complex events, and how to restrict the outputs of this language in meaningful manners. In the following, we study practical aspects of the CEL syntax that impact how efficiently can formulas be evaluated.


%% file: language-prop.tex

We now turn to study the syntactic form of CEL formulas. We define \emph{well-formed} and \emph{safe} formulas, which are syntactic restrictions that characterize semantic properties of interest. Then, we define a convenient normal form and show that any formula can be rewritten in this form.

\subsection{Syntactic restrictions of formulas}

Although CEL has well-defined semantics, there are some formulas whose semantics can be unintuitive because the use of variables is not restricted. Consider for example
$$\varphi_5 \ = \ (H \as x) \FILTER (y.tmp \leq 30).$$ 
Here, $x$ will be naturally bound to the only element in a complex event, but $y$ will not \emph{add} a new position to the output. By the semantics of CEL, a valuation $\nu$ for $\varphi_5$ must assign a position for $y$ that satisfies the filter, but such position is not restricted to occur in the complex event. Moreover, $y$ is not necessarily bound to any of the events seen up to the last element, and thus a complex event could depend on future events. For example, if we evaluate $\varphi_5$ over our running example $S$ (Figure~\ref{fig:stream}), we have that $\{2\} \in \sem{\varphi_5}(S)$, but this depends on the event at position $6$. This means that to evaluate this formula we potentially need to inspect events that occur after all events composing the output complex event have been seen, an arguably undesired situation.

To avoid this problem, we introduce the notion of \emph{well-formed} formulas. As the previous example illustrates, this requires defining where variables are \emph{bound} by a sub-formula of the form $R\as x$. The set of bound variables of a formula $\varphi$ is denoted by $\vi(\varphi)$ and is recursively defined as follows:
$$
\renewcommand{\arraystretch}{1.2}
\begin{array}{rcl}
\vi(R\as x) & = & \{x\} \\
\vi(\rho \FILTER P(\bar{x})) & = & \vi(\rho)  \\
\vi(\rho_1 \cor \rho_2) & = & \vi(\rho_1) \cap \vi(\rho_2) \\
\vi(\rho+) & = & \emptyset \\
\vi(\rho_1 \sq \rho_2) & = & \vi(\varphi_1) \cup \vi(\varphi_2) \\
\vi(\mathtt{SEL}(\rho)) & = & \vi(\rho)
\end{array}
$$
where $\mathtt{SEL}$ is any selection strategy. Note that for the $\scor$ operator a variable must be defined in both formulas in order to be bound. We say that a CEL formula $\varphi$ is \emph{well-formed} if for every sub-formula of the form $\rho\FILTER P(\bar{x})$ and every $x\in \bar{x}$, there is another sub-formula $\rho_x$ such that $x\in\vi(\rho_x)$ and $\rho$ is a sub-formula of $\rho_x$.  Note that this definition allows for including filters with variables defined in a wider scope. For example, formula $\varphi_4$ in Section~\ref{sec:in_action} is well-formed although it has the not-well-formed formula $(T \as y\FILTER y.id=x.id)+$ as a sub-formula. \par

One can argue that it would be desirable to restrict the users to only write well-formed formulas. Indeed, the well-formed property can be checked efficiently by a syntactic parser and users should understand that all variables in a formula must be correctly defined. Given that well-formed formulas have a well-defined variable structure, in the future we restrict our analysis to well-formed formulas.\par

Another issue for CEL is that the reuse of variables can easily produce unsatisfiable formulas. For example, the formula $\psi = T \as x \sq T \as x$ is not satisfiable (i.e. $\sem{\psi}(S) = \emptyset$ for every $S$) because variable $x$ cannot be assigned to two different positions in the stream. However, we do not want to be too conservative and disallow the reuse of variables in the whole formula (otherwise formulas like $\varphi_2$ in Section~\ref{sec:in_action} would not be permitted). This motivates the notion of \emph{safe} CEL formulas. We say that a CEL formula is \emph{safe} if for every sub-formula of the form $\varphi_1 \sq \varphi_2$ it holds that $\vdefplus(\varphi_1) \cap \vdefplus(\varphi_2) = \emptyset$. For example, all CEL formulas in this paper are safe except for the formula~$\psi$ above.

The safe notion is a mild restriction to help the evaluation of CEL, and can be easily checked during parsing time. However, safe formulas are a subclass of CEL and it could be the case that they do not capture the full language. We show in the next result that this is not the case. Formally, we say that two CEL formulas $\varphi$ and $\psi$ are equivalent, denoted by $\varphi\equiv\psi$, if for every stream $S$ and complex event $C$, it is the case that $C \in \sem{\varphi}(S)$ if, and only if, $C \in \sem{\psi}(S)$.
\begin{theorem}\label{theo:safe}
	Given a core-CEL formula $\varphi$, there is a safe formula $\varphi'$ s.t. $\varphi\equiv\varphi'$ and $|\varphi'|$ is at most exponential in~$|\varphi|$.
\end{theorem}
By this result, we can restrict our analysis to safe formulas without loss of generality. 
Unfortunately, we do not know if the exponential size of  $\varphi'$ is necessary. 
We conjecture that this exponential blow-up is unavoidable, however, we do not know yet the corresponding lower bound.

\subsection{LP-normal form}\label{sec:lp-normal-form}

Now we study how to rewrite CEL formulas in order to simplify the evaluation of unary filters. 
Intuitively, filter operators in a CEL formula can become difficult to handle for a CEP query engine.
To illustrate this, consider again formula $\varphi_1$ in Section~\ref{sec:in_action}. Syntactically, this formula states ``find an event $x$ followed by an event $y$, and then check that they satisfy the filter conditions''. However, we would like an execution engine to only consider those events $x$ with $id=0$ that represent temperature above 40 degrees. Only afterwards the possible matching events $y$ should be considered. In other words, formula $\varphi_1$ can be restated as:
\begin{multline*}
	\varphi_1'=[(T \as x) \FILTER (x.tmp > 40 \land x.id=0)];\\
	[(H \as y)\FILTER (y.hum <= 25\ \land y.id=0)]
\end{multline*}


This example motivates defining the \emph{locally parametrized} normal form (LP normal form). Let $\uset$ be the set of all predicates $P\in \pset$ of arity 1 (i.e. $P \subseteq \tuples(\cR)$). We say that a formula $\varphi$ is in LP-normal form if the following condition holds: for every sub-formula $\varphi'\FILTER P(\bar{x})$ of $\varphi$, if $P \in \uset$, then $\varphi' = R \as x$ for some $R$ and $x$. In other words, all filters containing unary predicates are applied directly to the definitions of their variables. For instance, formula $\varphi_1'$ is in LP-normal form while formulas $\varphi_1$ and $\varphi_2$ are not. Note that non-unary predicates are not restricted, and they can be used anywhere in the formula.

One can easily see that having formulas in LP-normal form would be an advantage for an evaluation engine, because it can \emph{filter out} some events as soon as they arrive (see Section~\ref{sec:experiments} for further discussion). However, formulas that are not in LP-normal form can still be very useful for declaring patterns. To illustrate this, consider the formula:
\begin{multline*}
	\varphi_6 = (T\as x); ((T \as y \FILTER x.temp \geq 40) \cor  \\
	(H \as y \FILTER x.temp < 40))
\end{multline*}
Here, the $\sFILTER$ operator works like a conditional statement: if the $x$-temperature is greater than $40$, then the following event should be a temperature, and a humidity event otherwise. This type of conditional statements can be very useful, but at the same time it can be hard to evaluate.
Fortunately, the next result shows that one can always rewrite a formula into LP-normal form, incurring in the worst case in an exponential blow-up in the size of the formula.  

\begin{theorem}\label{theo:LP-normal-form}
	Let $\varphi$ be a core-CEL formula. Then, there is a core-CEL formula $\psi$ in LP-normal form such that $\varphi\equiv\psi$, and $|\psi|$ is at most exponential in $|\varphi|$.
\end{theorem}

%

The importance of this result and Theorem~\ref{theo:safe} will become clear in the next sections, where we show that safe formulas in LP-normal form have good properties for evaluation. Similar to Theorem~\ref{theo:safe}, we do not know if the exponential blow-up is unavoidable and leave this for future work.

%% file: automata.tex

In this section, we introduce a formal computational model for evaluating CEL formulas called \emph{complex event automata} (CEA for short). Similar to classical database management systems (DBMS), it is useful to have a formal model that stands between the query language and the evaluation algorithms, in order to simplify the analysis and optimization of the whole evaluation process. There are several examples of DBMS that are based on this approach like regular expressions and finite state automata~\cite{HopcroftU79,Aho90}, and relational algebra and SQL~\cite{abiteboul1995foundations,Ramakrishnan2003}. Here, we propose CEA as the intermediate evaluation model for CEL and show later how to compile any (unary) CEL formula into a CEA.

As its name suggests, complex event automata (CEA) are an extension of \emph{Finite State Automata} (FSA). The first difference from FSA comes from handling streams instead of words. A CEA is said to run over a stream of tuples, unlike FSA which run over words of a certain alphabet. The second difference arises directly from the first one by the need of processing tuples, which can have infinitely many different values, in contrast to the finite input alphabet of FSA. To handle this, our model is extended the same way as a Symbolic Finite Automata (SFA)~\cite{veanes2013applications}. SFAs are finite state automata in which the alphabet is described implicitly by a boolean algebra over the symbols. This allows automata to work with a possibly infinite alphabet and, at the same time, use finite state memory for processing the input. CEA are extended analogously, which is reflected in transitions labeled by unary predicates over tuples.
The last difference addresses the need to generate complex events instead of boolean answers. A well known extension for FSA are \emph{Finite State Transducers}~\cite{berstel2013transductions}, which are capable of producing an output whenever an input element is read. Our computational model follows the same approach: CEA are allowed to generate and output complex events when reading a stream.

Recall from Section~\ref{sec:language-prop} that $\uset$ is the subset of unary predicates of $\pset$. Let $\amark, \umark$ be two symbols. A \emph{complex event automaton} (CEA) is a tuple $\cA = (Q, \Delta, I, F)$ where $Q$ is a finite set of states, $\Delta \subseteq Q \times (\uset \times \{\amark, \umark\}) \times Q$ is the transition relation, and $I, F \subseteq Q$ are the set of initial and final states, respectively. Given a stream $S = t_0 t_1 \ldots$, a run $\rho$ of $\cA$ over $S$ is a sequence of transitions:
$\rho: q_0 \ \trans{P_0 / m_0} \ q_1 \  \trans{P_1 / m_1} \ \cdots \ \trans{P_n / m_n} \ q_{n+1}$
such that $q_0 \in I$, $t_i \in P_i$ and $(q_i, P_{i}, m_{i}, q_{i+1}) \in \Delta$ for every $i \leq n$. We say that $\rho$ is \emph{accepting} if $q_{n+1} \in F$ and $m_n = \amark$. 
We denote by $\run_n(\cA, S)$ the set of accepting runs of $\cA$ over $S$ of length $n$. Further, $\rmatch(\rho)$ denotes the set of positions where the run \emph{marks} the stream, namely $\rmatch(\rho) = \{i \in [0,n] \, \mid \, m_i = \amark \}$. Intuitively this means that when a transition is taken, if the transition has the $\amark$ symbol then the \emph{current} position of the stream is included in the output (similar to the execution of a transducer). Note that we require the last position of an accepting run to be marking, as otherwise an output could depend on \emph{future} events (see the discussion about well-formed formulas in Section~\ref{sec:language-prop}). Given a stream $S$ and $n \in \bbN$, we define the set of complex events of $\cA$ over~$S$~ at position $n$~as
$
\sem{\cA}_n(S) = \{\rmatch(\rho) \mid \rho \in \run_n(\cA, S) \}
$ and the set of all complex events as $\sem{\cA}(S) = \bigcup_n \sem{\cA}_n(S)$.
Note that $\sem{\cA}(S)$ can be infinite, but $\sem{\cA}_n(S)$ is finite.

Consider as an example the CEA $\cA$ depicted in Figure~\ref{fig:unbounded_ma}. In this CEA, each transition $P(x)\!\mid\!\amark$ marks one $H$-tuple and each transition $P'(x)\!\mid\!\amark$ marks a sequence of $T$-tuples with temperature bigger than $40$. Note also that the transitions labeled by $\TRUE\!\mid\!\umark$ allow $\cA$ to arbitrarily skip tuples of the stream. Then, for every stream~$S$, $\sem{\cA}(S)$ represents the set of all complex events that begin and end with an $H$-tuple and also contain some of the $T$-tuples with temperature higher than $40$.

\begin{figure}
	\begin{center}
	\begin{tikzpicture}[->,>=stealth, semithick, auto, initial text= {}, initial distance= {3mm}, accepting distance= {4mm}, node distance=2.5cm, semithick]
		\tikzstyle{every state}=[draw=black,text=black,inner sep=0pt, minimum size=8mm]
		\node[initial,state]	(1) 				{$q_1$};
		\node[state]			(2) [right of=1]	{$q_2$};
		\node[accepting,state]	(3) [right of=2]	{$q_3$};
		\path
		(1)
		edge 				node {$P(x) \mid \amark$} (2)
		edge [loop above] 	node {$\TRUE \mid \umark$} (1)
		(2)
		edge [in=100,out=130,loop] node[pos=0.55] {$P'(x) \mid \amark$} (2)
		edge [in=50,out=80,loop] node[pos=0.45] {$\TRUE \mid \umark$} (2)
		edge node {$P(x) \mid \amark$} (3);
	\end{tikzpicture}
	\vspace{-.1cm}
	\caption{A CEA that can generate an unbounded amount of complex events. Here $P(x) := \type(x)=H$ and $P'(x) := \type(x)=T \wedge x.temp > 40$.
	}\label{fig:unbounded_ma}
	\vspace{-0.7cm}
	\end{center}
\end{figure}
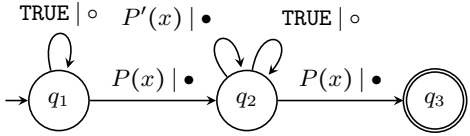

It is important to stress that CEA are designed to be an evaluation model for the unary sub-fragment of CEL (a formal definition is presented in the next paragraph). Several computational models have been proposed for complex event processing~\cite{cayuga,distCED,SASE,nextCEP}, but most of them are informal and non-standard extensions of finite state automata. In our framework, we want to give a step back compared to previous proposals and define a simple but powerful model that captures the \emph{regular core} of CEL. With ``regular'' we mean all CEL formulas that can be evaluated with finite state memory.  Intuitively, formulas like $\varphi_1$, $\varphi_2$ and $\varphi_3$ presented in Section~\ref{sec:in_action} can be evaluated using a bounded amount of memory. In contrast, formula $\varphi_4$ needs unbounded memory to store \emph{candidate} events seen in the past, and thus, it calls for a more sophisticated model (e.g. data automata~\cite{segoufin2006automata}). Of course one would like to have a full-fledged model for CEL, but to this end we must first understand the regular fragment. For these reasons, a computational model for the whole CEP logic is left as future work (see Section~\ref{sec:conclusions}). 

\smallskip
\noindent {\bf Compiling unary CEL into CEA.}
We now show how to compile a well-formed and unary CEL formula $\varphi$ into an equivalent CEA $\cA_\varphi$.
Formally, we say that a CEL formula $\varphi$ is unary if for every subformula of $\varphi$ of the form $\varphi'\FILTER P(\bar{x})$, it holds that $P(\bar{x})$ is a unary predicate (i.e. $P(\bar{x}) \in \uset$).
For example, formulas $\varphi_1$, $\varphi_2$, and $\varphi_3$ in Section~\ref{sec:in_action} are unary, but formula $\varphi_4$ is not (the predicate $y.id = x.id$ is binary).
As motivated in Section~\ref{sec:in_action} and~\ref{sec:lp-normal-form}, and further supported by our experiments (see Section~\ref{sec:experiments}), despite their appear simplicity unary formulas already present non-trivial computational challenges.



\begin{theorem}\label{theo:core-automata}
	For every well-formed formula $\varphi$ in unary core-CEL, there is a CEA $\cA_\varphi$ equivalent to $\varphi$. Furthermore, $\cA_\varphi$ is of size at most linear in $|\varphi|$ if $\varphi$ is safe and in LP-normal form and at most double exponential in $|\varphi|$ otherwise.
\end{theorem} 
The proof of Theorem~\ref{theo:core-automata} is closely related with the safeness condition and the LP-normal form presented in Section~\ref{sec:language-prop}. The construction goes by first converting $\varphi$ into an equivalent CEL formula $\varphi'$ in LP-normal form (Theorem~\ref{theo:LP-normal-form}) and then building an equivalent CEA from $\varphi'$. We show that there is an exponential blow-up for converting $\varphi$ into LP-normal form.
Furthermore, we show that the output of the second step is of linear size if $\varphi'$ is safe, and of exponential size otherwise, suggesting that restricting the language to safe formulas allows for more efficient evaluation.

So far we have described the compilation process without considering selection strategies. To include them, we need to extend our notation and allow selection strategies to be applied directly over CEA.
Given a CEA $\cA$, a selection strategy $\SEL$ in $\{\sSTRICT,\sNEXT,\LAST,\sMAX\}$ and stream $S$, the set of outputs $\sem{\SEL(\cA)}(S)$ is defined analogously to $\sem{\SEL(\varphi)}(S)$ for a formula $\varphi$. Then, we say that a CEA $\cA_1$ is equivalent to $\SEL(\cA_2)$ if $\sem{\cA_1}(S)=\sem{\SEL(\cA_2)}(S)$ for every stream $S$.

\begin{theorem}\label{theo:selectors-compilation}
	Let $\SEL$ be a selection strategy. For any CEA $\cA$, there is a CEA $\cA_{\SEL}$ equivalent to $\SEL(\cA)$. Furthermore, the size of $\cA_{\SEL}$ is, w.r.t. the size of $\cA$, at most linear if $\SEL=\sSTRICT$, and at most exponential otherwise.
\end{theorem}

At first this result might seem unintuitive, specially in the case of $\sNEXT$, $\LAST$ and $\sMAX$. It is not immediate (and rather involved) to show that there exists a CEA for these strategies because they need to \emph{track} an unbounded number of complex events using finite memory. Still, this can be done with an exponential blow-up in the number of states.\par

Theorem~\ref{theo:selectors-compilation} concludes our study of the compilation of unary CEL into CEA. We have shown that CEA is not only able to evaluate CEL formulas, but also that it can be further exploited to evaluate selections strategies. We finish by introducing the notion of I/O-determinism that will be crucial for our evaluation algorithms in the next section.

\smallskip
\noindent {\bf I/O-deterministic CEA.} To evaluate CEA in practice we will focus on the class of the so-called \emph{I/O-deterministic} CEA (for Input/Output deterministic). We say that a CEA $\cA$ is I/O-deterministic if $|I| = 1$ and for any two transitions $(p,P_1,m_1,q_1)$ and $(p,P_2,m_2,q_2)$, either $P_1$ and $P_2$ are mutually exclusive (i.e. $P_1 \cap P_2 = \emptyset$), or $m_1 \neq m_2$. Intuitively, this notion imposes that given a stream $S$ and a complex event $C$, there is at most one run over $S$ that generates $C$ (thus the name referencing the input and the output). In contrast, the classical notion of determinism would require that there is at most one run over the entire stream.

I/O-deterministic CEA are important because they allow for a simple and efficient evaluation algorithm (discussed in Sections~\ref{sec:evaluation}~and~\ref{sec:experiments}). But for this algorithm to be useful, we need to make sure that every CEA can be I/O determinized. 
Formally, we say that two CEA $\cA_1$ and $\cA_2$ are equivalent (denoted $\cA_1\equiv\cA_2$) if for every stream $S$ we have $\sem{\cA_1}(S)=\sem{\cA_2}(S)$. 
Then we say that CEA are closed under I/O determinism if for every CEA $\cA$ there is an I/O-deterministic CEA $\cA'$ such that $\cA\equiv\cA'$.

\begin{proposition} \label{prop:CEA_closure}
	CEA are closed under I/O-determinism.
\end{proposition}

This result and the compilation process allow us to evaluate CEL formulas by means of I/O-deterministic CEA without loss of generality. In the next section we present an algorithm to perform this evaluation efficiently.





%% file: evaluation.tex

In this section we show how to efficiently evaluate a complex event automaton (CEA).
We first formalize the notion of an efficient evaluation in the context of CEP and then provide algorithms to evaluate CEA efficiently.

\subsection{Efficiency in CEP}
Defining a notion of efficiency for CEP is challenging since we would like to compute complex events in one pass and using a restricted amount of resources.
Streaming algorithms~\cite{ikonomovska2013algorithmic,golab2003issues} are a natural starting point as they usually restrict the time allowed to process each tuple and the space needed to process the first $n$ items of a stream (e.g., constant or logarithmic in $n$).
However, an important difference is that in CEP the arrival of a single event might generate an exponential number of complex events as output. Therefore no algorithm producing this output could guarantee any sort of efficiency, because there are particular examples in which only generating the outputs take exponential time in size of the processed sub-stream. To overcome this problem, we propose to divide the evaluation in two parts: (1) consuming new events and updating the internal memory of the system and (2) generating complex events from the internal memory of the system. We require both parts to be as efficient as possible. First, (1) should process each event in a time that does not depend on the number of events seen in the past. Second, (2) should not spend any time \emph{processing} and instead it should be completely devoted to generating the output. To formalize this notion, we assume that there is a special instruction ${\tt yield}_S$ that returns the next element of a stream~$S$. Then, given a function $f:\bbN\rightarrow\bbN$, a \emph{CEP evaluation algorithm} with $f$-update time is an algorithm that evaluates a CEA $\cA$ over a stream $S$ such that:
\begin{enumerate} 
	\item between any two calls to ${\tt yield}_S$, the time spent is bounded by $\cO(f(|\cA|)\cdot|t|)$, where $t$ is the tuple returned by the first of such calls, and
	\item maintains a data structure $D$ in memory, such that after calling ${\tt yield}_S$ $n$ times, the set $\sem{\cA}_n(S)$ can be enumerated from $D$ with constant delay. 
\end{enumerate}
The notion of constant-delay enumeration was defined in the database community~\cite{Segoufin13,Bagan:2007} precisely for defining efficiency whenever the output might be larger than the input. Formally, it requires the existence of a routine $\textproc{Enumerate}$ that receives $D$ as input and outputs all complex events in $\sem{\cA}_n(S)$ without repetitions, while spending a constant amount of time before and after each output. Naturally, the time to generate a complex event $C$ must be linear in $|C|$. We remark that (1) is a natural restriction imposed in the streaming literature~\cite{ikonomovska2013algorithmic}, while (2) is the minimum requirement if an arbitrarily large set of arbitrarily large outputs must be produced~\cite{Segoufin13}.

Note that the update time $\cO(f(|\cA|)\cdot|t|)$ is linear in $|t|$ if we consider that $\cA$ is fixed. Since this is the case in practice (i.e. the automaton is generally small with respect to the stream, and does not change during evaluation), this amounts to constant update time when measured under data complexity (tuples can also be considered of constant size).

\subsection{Evaluation of I/O-deterministic CEA} \label{sec:evaluation-iodet}
We describe a CEP evaluation algorithm with $f(n)=n$ update time for I/O-deterministic CEA.
We define the algorithm's underlying data structure, then show how to update this data structure upon new events, and finally how to enumerate the resulting complex events with constant delay.

\smallskip
\noindent {\bf Data structure.}
The atomic element in our data structure is the \emph{node}.
A node is defined as a pair $(p,l)$, where $p \in \bbN$ represents a position in the stream and $l$ is a list of nodes.
A node is initialized by calling $\mNode(p,l)$, and the methods $\mposition$ and $\mlist$ return $p$ and $l$, respectively.

The data structure maintained by our algorithm is composed by linked-lists of nodes.
For operating a linked-list $l$ we use the methods $\madd$, $\mappend$ and $\mlazycopy$.
Specifically, $\madd(n)$ adds the node $n$ at the beginning of $l$, and $\mappend(l')$ appends a list $l'$ at the end of $l$.
An important property of the data structure is that no element is ever removed from the lists, only adding nodes or appending lists is allowed.
This allows us to represent a list as a pair $l = (s,e)$, where $s$ is its starting node and $e$ its ending node.
Then, $\mlazycopy$ returns a copy of $l$, defined by the pointers $(s,e)$, and the generated \emph{copy} of the list is not affected by future changes on $l$. Furthermore, it is trivial to see that $\mlazycopy$ runs in constant time (i.e. $\cO(1)$). The methods used for navigating the list are $\mbegin$ and $\mnext$. $\mbegin$ gives a pointer to the first node of the list, and $\mnext$ returns the next element of the list and ${\tt false}$ when it reaches the end.

\begin{algorithm}[t]
	\caption{Evaluate $\cA$ over a stream $S$}\label{alg:cea-eval}
	\begin{algorithmic}[1]
		\Require{An I/O deterministic CEA $\cA = (Q,\delta,q_0,F)$}
		\Procedure{Evaluate}{$S$}
		\ForAll{$q \in Q \setminus \{q_0\}$}
		\State $\alist_q \gets \epsilon$
		\EndFor
		\State $\alist_{q_0} \gets [\bot]$
		\While {$t \gets \myield_S$}
			\ForAll{$q \in Q$}  \label{ma-eval-line:initlist1}
				\State $\alist^\aold_q \gets \alist_q\!.\mlazycopy, \; \alist_q \gets \epsilon$ \label{ma-eval-line:initlist2}
			\EndFor
			\ForAll{$q \in Q \textbf{ with } \alist_q^\aold \neq \epsilon$} \label{ma-eval-line:main1}
				\If{$p^\amark \gets \delta(q,t,\amark)$}\label{ma-eval-line:upd1}
					\State $\alist_{p^\amark}\!.\madd(\mNode(t.\mposition,\alist_q^\aold))$\label{ma-eval-line:upd2}
				\EndIf
				\If{$p^\umark \gets \delta(q,t,\umark)$}\label{ma-eval-line:upd3}
					\State $\alist_{p^\umark}\!.\mappend(\alist_q^\aold)$\label{ma-eval-line:upd4}
				\EndIf
			\EndFor
			\State $\textproc{Enumerate}(\{\alist_q\}_{q \in Q},F,t.\mposition)$
		\EndWhile
		\EndProcedure
	\end{algorithmic}
\end{algorithm}

\begin{algorithm}[t]
	\caption{Enumerate all mappings}\label{alg:cea-enum}
	\begin{algorithmic}[1]
		\Procedure{{Enumerate}}{$\{\alist_q\}_{q \in Q}$, $F$, $now$}
			\ForAll{$q \in F$ \textbf{with} $\alist_q \neq \epsilon$}
				\State $\alist_q\!.\mbegin$
				\While{$n \gets \alist\!.\mnext$ \textbf{and} $n.\mposition = now$}
					\State{$\textproc{EnumAll}(n.\mlist, \{n.\mposition\})$}
				\EndWhile
			\EndFor
		\EndProcedure
		
		\smallskip
		
		\Procedure{{EnumAll}}{$\alist, C$}
			\State $\alist\!.\mbegin$
			\While{$n \gets \alist\!.\mnext$}
				\If{$n = \bot$}
					\State ${\tt Output}(C)$
				\Else
					\State $\textsc{EnumAll}(n.\mlist, \,C \cup \{n.\mposition\})$
				\EndIf			
			\EndWhile
		\EndProcedure
	\end{algorithmic}
\end{algorithm}

\paragraph*{\bf Evaluation}
The CEP evaluation algorithm for an I/O-deterministic CEA $\cA=(Q,\delta,q_0,F)$ is given in Algorithms~\ref{alg:cea-eval} and~\ref{alg:cea-enum}.
To ease the notation, we extend $\delta$ as a function $\delta(q,t,m)$ that retrieves the (unique) state $p = \delta(q,P,m)$ for some predicate $P$ such that $t \in P$; if there is no such $P$, it returns ${\tt false}$.
Basically, if a run is in state $q$, then $p$ is the state it moves when reading $t$ and marking $m$.

The procedure $\textproc{Evaluate}$ keeps the evaluation of $\cA$ by simulating all its possible runs, and has a list $\alist_q$ for each state $q$ to keep track on the complex events.
Intuitively, each $\alist_q$ keeps the information of the partial complex events generated by the partial runs currently ending at $q$. Each node $n$ in $\alist_q$ represents (through its $n.\mlist$) a subset of these complex events, all of them having $n.\mposition$ as their last position. These sets are pairwise disjoint (which is an important property for constant-delay enumeration of the output).
Each $\alist_q$ is initialized as the empty list, represented by $\epsilon$, except for $\alist_{q_0}$, which begins with only the sink node $\bot$ in it.
The algorithm then reads $S$ using ${\tt yield}_S$ to get each new event.
For each new event $t$, the procedure updates the data structure as follows.
It starts by creating a copy of each $\alist_q$, and storing it in $\alist_q^\aold$ (lines \ref{ma-eval-line:initlist1}-\ref{ma-eval-line:initlist2}).
Then, for each $q$ with non-empty $\alist_q$ it extends the runs that are currently at $q$ by simulating the possible outgoing transitions satisfied by $t$ (lines \ref{ma-eval-line:main1}-\ref{ma-eval-line:upd4}).
After doing this for all $q$, it calls the $\textproc{Enumerate}$ procedure to enumerate all output complex events generated by $t$.

The core processing of Algorithm~\ref{alg:cea-eval} is in updating the structure by extending the runs currently at $q$ (lines \ref{ma-eval-line:upd1}-\ref{ma-eval-line:upd4}).
Specifically, line~\ref{ma-eval-line:upd2} considers the $\amark$-transition and line~\ref{ma-eval-line:upd4} the $\umark$-transition (recall that $\cA$ is I/O-deterministic).
As we said before, each $\alist_q$ represents the complex events of runs currently at $q$.
To extend these runs with a $\amark$-transition, line~\ref{ma-eval-line:upd2} creates a new node $n^*$ with the current position in $S$ (i.e. $t.\mposition$) as its position, and the old value of $\alist_q$ as its predecessors list.
Then, $n^*$ is added at the top of the new list of $p^\amark = \delta(q,t,\amark)$.
On the other hand, to extend the runs with a $\umark$-transition, it only needs to append the old list of $q$ to the list of $p^\umark=\delta(q,t,\umark)$ (line~\ref{ma-eval-line:upd4}).

By looking at Algorithm~\ref{alg:cea-eval}, one can see that the update of each $\alist_q$ takes time $\cO(|t|)$, and therefore $\cO(|Q|\cdot |t|)$ for the whole update procedure.
This, added to the $\cO(|Q|)$ of the lazy copying of the lists, gives us an overall $\cO(|\cA| \cdot |t|)$ bound on the time between each call to $\myield_S$, satisfying condition~$(1)$ with $f(|\cA|) = |\cA|$.

\smallskip
\noindent{\bf Enumeration.}
One can consider the data structure maintained by $\textproc{Evaluate}$ as a directed acyclic graph: vertices are nodes and there is an outgoing edge from node $n$ to node $n'$ if $n'$ appears in $n.\mlist$.
By following Algorithm~\ref{alg:cea-eval}, one can easily check that the sink node $\bot$ is reachable from every node in this directed acyclic graph, namely, for any $q$ and any node $n$ in $\alist_q$ there exists a path $n=n_1 \ldots, n_k, \bot$.
Furthermore, each of this path represents a complex event $n_k.\mposition, \ldots, n_1.\mposition$ outputted by some run of $\cA$ over $S$ that ends at $q$.

Given the previous discussion, the $\textproc{Enumerate}$ procedure in Algorithm~\ref{alg:cea-enum} is straightforward: it simply traverses the directed acyclic graph in a depth-first manner, computing a complex event for each path. 
To ensure that all outputs are enumerated, it needs to do this for each node $n$ in an accepting state and whose position is equal to the current position (i.e. $now$).
Because new nodes are added on top, it iterates over each accepting list from the beginning, stopping whenever it finds a node with a position different from $now$.

It is important to note that $\textproc{Enumerate}$ does not satisfy condition (2) of a CEP evaluation algorithm, namely, taking a constant delay between two outputs. The problem relies in the depth-first search traversal of the acyclic graph: there can be an unbounded number of backtracking steps, creating a delay that is not constant between outputs. To solve this, one can use a stack with a smart policy to avoid these unbounded backtracking steps. Given space restrictions, we present this modification of Algorithm~\ref{alg:cea-enum} in the appendix.

\subsection{CEA and selection strategies}

Given that any CEA can be I/O-determinized (Proposition~\ref{prop:CEA_closure}), we can use Algorithms~\ref{alg:cea-eval} and~\ref{alg:cea-enum} to evaluate any CEA. Unfortunately, the determinization procedure has an exponential blow-up in the size of the automaton.
\begin{theorem}\label{theo:ndet-evaluation}
	For every CEA $\cA$, there is an CEP evaluation algorithm with $2^{|\cA|}$-update time.
\end{theorem}
We can further extend the CEP evaluation algorithm for I/O-deterministic CEA to any selection strategies by using the results of Theorem~\ref{theo:selectors-compilation}. However, by naively applying Theorem~\ref{theo:selectors-compilation} and then I/O-determinizing the resulting automaton, we will have a double exponential blow-up in the update time. By doing the compilation of the selection strategies and the I/O-determinization together, we can lower the update time. Moreover, and rather surprisingly, we can evaluate $\NEXT$ and $\LAST$ without determinizing the automaton, and therefore with linear update time.
\begin{theorem}\label{theo:selectors-evaluation}
	Let $\SEL$ be a selection strategy. For any CEA $\cA$, there is an CEP evaluation algorithm for $\SEL(\cA)$. Furthermore, the update time is $|\cA|$ if $\SEL \in \{\NEXT, \LAST\}$, $2^{|\cA|}$ if $\SEL = \sSTRICT$ and $4^{|\cA|}$ if $\SEL = \MAX$.
\end{theorem}
Due to space limitations, the constructions and algorithms in Theorem~\ref{theo:selectors-evaluation} are deferred to the appendix.

%% file: experiments.tex

Having all the building blocks, we proceed to show how to evaluate a unary CEL formula in practice. We present an experimental evaluation that validates the simplicity and efficiency of the presented framework.

\begin{figure}
	\small
	\begin{center}
		\begin{tikzpicture}[>=stealth, node distance=0.9cm, every text node part/.style={align=center},
		component/.style={rectangle, rounded corners, draw, minimum width=3.8cm},
		marrow/.style={->, thick, bulletcolor, draw=bulletcolor, >=stealth, line width=1.5pt}, 
		darrow/.style={->, dashed, bulletcolor, draw=bulletcolor, >=stealth, line width=1.5pt}, 
		res/.style={rounded corners, minimum height=5mm, minimum width=2.5cm, dashed, node distance=3.4cm},
		tuples/.style={minimum height=0mm, minimum width=2.5cm}]

		\node [component]  (Parser) {\textbf{Parser} (Th. 1)};
		\node [component, below of=Parser]  (Preprocessor) {\textbf{Query Rewrite} (Th. 2)};
		\node [component, below of=Preprocessor] (Rewrite) {\textbf{Compilation} (Th. 3, 4)};
		\node [component, below of=Rewrite] (Optimizer)  {\textbf{Evaluation} (Th. 5, 6)};
		\node [tuples] (outp) at ($(Optimizer) + (2.5,-0.7)$) {Output (complex events)};
		
		\node [node distance=2.9cm, left of=Parser] (query) {CEL};
		\node [node distance=2.9cm, left of=Optimizer] (stream) {Stream};
		
		\draw [marrow] (query) -- (Parser);
		\draw [marrow] (Parser) edge node (o1) {} (Preprocessor);
		\draw [marrow] (Preprocessor) edge node (o2) {} (Rewrite);
		\draw [marrow] (Rewrite) edge node (o3) {} (Optimizer);
		\draw [marrow] (Optimizer) -- ($(Optimizer) + (0,-0.7)$) -- (outp);
		\draw [marrow] (stream) -- (Optimizer);
		
		\node [res, right of=o1]  (tree) {WF and safe};
		\node [res, right of=o2]  (pretree) {LP-normal form};
		\node [res, right of=o3]  (logical) {CE automaton};
		
		
		\draw [darrow] (o1) -- (tree);
		\draw [darrow] (o2) -- (pretree);
		\draw [darrow] (o3) -- (logical);
		\end{tikzpicture}
		\vspace{-.3cm}
		\caption{Evaluation framework for CEL.}\label{fig:framework}
		\vspace{-.6cm}
	\end{center}
\end{figure}
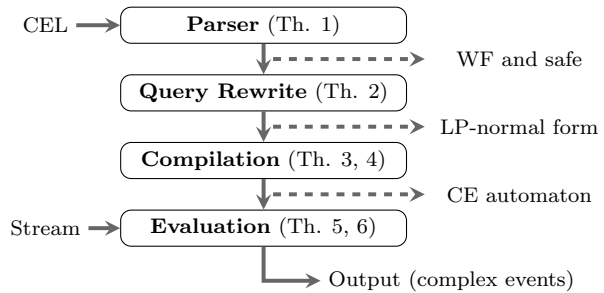

\subsection{Framework}

In Figure~\ref{fig:framework}, we show the evaluation cycle of a CEL formula in our framework and how all the results and theorems fit together. To explain this framework, consider a unary CEL formula $\varphi$ (possibly with selection strategies). 
The process starts in the parser module, where we check if $\varphi$ is well-formed and safe. These conditions are important to ensure that $\varphi$ is satisfiable and make a correct use of variables. Note that a CEP system could translate unsafe formulas (Theorem~\ref{theo:safe}), incurring however in an exponential blow-up.

The next module rewrites a well-formed and safe formula $\varphi$ into LP-normal form by using the rewriting process of Theorem~\ref{theo:LP-normal-form} which, in the worst case, can produce an exponentially larger formula. To avoid this cost, in many cases one can apply \emph{local rewriting rules}~\cite{abiteboul1995foundations,Ramakrishnan2003}.
For example, in Section~\ref{sec:in_action} we converted $\varphi_1$ into $\varphi_1'$ by applying a \emph{filter push} on filters, avoiding the exponential blow-up of Theorem~\ref{theo:LP-normal-form}. Unfortunately, we cannot apply this technique over formulas like $\varphi_6$ in Section~\ref{sec:language-prop}, maintaining the exponential blow-up. Nevertheless, formulas like $\varphi_6$ are rather uncommon in practice and local rewriting rules will usually produce LP-formulas of polynomial size. 

The third module receives a formula in LP-normal form and builds a complex event automaton $\cA_\varphi$ of polynomial size. 
Then, the last module runs $\cA_\varphi$ over the stream by using our CEP evaluation procedure for I/O deterministic CEA (Algorithms~\ref{alg:cea-eval}~and~\ref{alg:cea-enum}). 
If there is no selection strategy, $\cA_\varphi$ must be determinized before running the CEP evaluation algorithm. In the worst case, this determinization is exponential in $\cA_\varphi$, nevertheless, in practice the size of $\cA_\varphi$ is rather small (see the experiments below).
If a selection strategy $\SEL$ is used, we can use the algorithms of Theorem~\ref{theo:selectors-evaluation} for evaluating $\SEL(\cA_\varphi)$, having a similar update time than evaluating $\cA_\varphi$ alone. 
As we show next, evaluating $\MAX(\cA_\varphi)$ or $\LAST(\cA_\varphi)$ has even better performance than evaluating $\cA_\varphi$ directly.

\subsection{Experiments}

To validate our results, we implemented our complete framework (automatized from parsing to evaluation) and compared it against two of the most relevant actors in CEP systems: EsperTech~\cite{EsperTech}, an industrial CEP Stream processing system, and SASE~\cite{SASEwebsite}, an academic prototype. We use the Java-based open version of EsperTech~\cite{EsperTechVersion} and the Java open-source version of SASE\cite{SASEgithub}. Our implementation~\cite{CELGithub} is also written in Java to have a fair comparison.
\smallskip

\noindent \textbf{Setup.} We run our experiments on a server equipped with an 8-core Intel(R) Xeon(R) E5-2609v4 processor running at 1.7GHz, 16GB of RAM and the GNU operating system with Linux kernel 4.4.0-109-generic, distributed under Ubuntu 16.04.02. All experiments are performed using Java 1.8.0\_131 and the Java HotSpot(TM) 64-Bit Server Virtual Machine, build 25.131-b11. The reported measurements are the average over ten runs of the same experiment. The Virtual Machine is restarted with 8GB of freshly allocated memory before each repetition of each experiment. Experiments were stopped after one hour or when the allocated memory is exceeded; this is reported accordingly in the experimental results. Memory usage is measured using the JVM System call and after calling the Garbage Collector. For the sake of consistency, we verified that the three systems produced exactly the same set of output. There were a few cases in which this was not the case because SASE and EsperTech do not have well-defined semantics for selection strategies and use heuristic that affect the generated output. 
\begin{table}[tb]\label{tab:patterns}
	\centering
	\begin{tabular}{|c|l|}
		\hline
		$Q_1$ & $A \as x; B \as y; C \as z$ \\\hline 
		$Q_2$ & $A \as x; B \as y; C \as z; D \as w$ \\\hline 
		$Q_3$ & $((A \as x \cor B \as y) \cor C \as z); D \as w$ \\\hline 
		$Q_4$ & $(A \as x)+; (B \as y)$ \\\hline 
		$Q_5$ & $(A \as x)+; (B \as y)+; C \as z$ \\\hline 
		$Q_6$ & $((A \as x)+; (B \as y))+; C \as z$ \\\hline
	\end{tabular}
	\caption{Queries used in the experiments}
	\vspace{-4mm}
\end{table}

\smallskip
\noindent \textbf{Queries.} As there are no standard benchmarks for CEP patterns, we developed a small set of patterns for understanding how efficiently each system handles the basic operators. The used queries are denoted $Q_1$-$Q_6$ and are depicted in Table~\ref{tab:patterns}.
Despite their simplicity these queries are particularly important and already show the difference in performance between previous CEP systems and our framework.
Queries $Q_1$ and $Q_2$ measure how well a system handles concatenation. 
Queries $Q_3$ and $Q_4$ are intended to measure the efficiency with which a system handles disjunction and a single iteration, respectively. Query $Q_5$ contains two concatenated iterations, and finally $Q_6$ is a more complex query with nested iterations. This last query not only tests the efficiency of the systems under a slightly more complex query, but also the consistency of their semantics. It is important to mention that we do not test unary filters because they do not add complexity to our system. Indeed, we tested similar queries with unary filters and the performance of our framework was not affected while that of EsperTech and SASE were degraded. 
In order to test these queries over EsperTech and SASE we needed to translate them; we present these translations in our appendix.
\begin{figure*}[t!]
\centering
\includegraphics[width=1\textwidth]{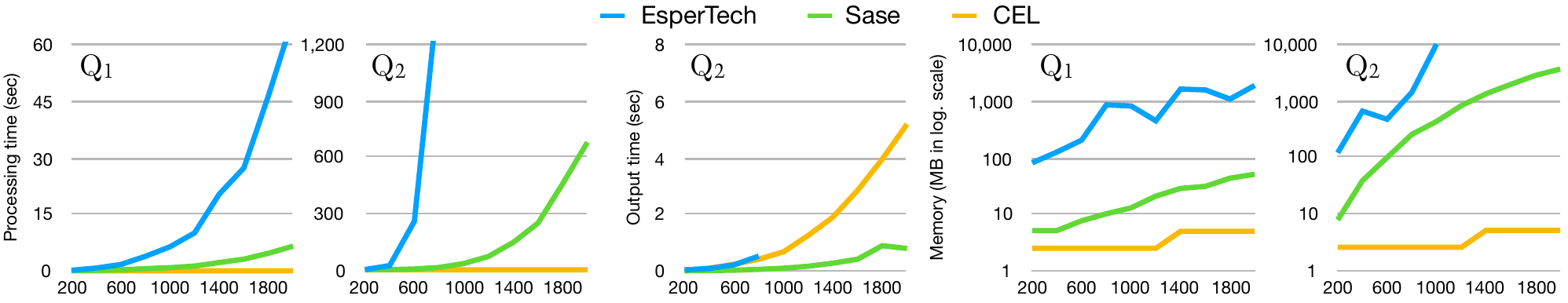}
\caption{Evaluation of $Q_1$ and $Q_2$ over streams of length $200,\ 400,\ \ldots,\ 2000$. We depict the processing times, then the time spent generating the output for $Q_2$ and finally the memory consumption in logarithmic scale.}
\label{fig:exp1}
\end{figure*}

\smallskip
\noindent \textbf{Stress Experiment.} We start by measuring how well each system manages partial complex events. For this, we do a ``stress experiment'': we evaluate queries $Q_1$ and $Q_2$ over a stream where all events are randomly generated with uniform distribution, except for the last event that fires the results (i.e. the last event is $C$ for $Q_1$ and $D$ for $Q_2$). 
This implies that no output is generated until the last event.
We run both queries over streams of increasing length 200, 400, \ldots, 2.000, and measure the processing time, the enumeration time, and the memory consumption.
Note that although 2.000 events (less than 1 MB) is a small amount of data in CEP, the number of outputs is big: around 200,000 for $Q_1$ and 20,000,000 for $Q_2$.
Figure~\ref{fig:exp1} depicts the processing time for both queries, the time required to enumerate the output for $Q_2$ and the memory used by each system for each query (in logarithmic scale).

Let us first analyze the processing times (first and second charts in Figure~\ref{fig:exp1}). An interesting remark is that EsperTech and SASE slow down in a non-linear fashion w.r.t. the stream size. This suggests that they are building partial outputs (e.g. all pairs $(A,B)$ for $Q_1$) in memory, waiting for an event that triggers a complex event. In contrast, our framework provides constant-time per processed event, and therefore degrades linearly at worst. While we process $Q_2$ over a stream of 2,000 events in less than $0,01$ seconds with a throughput of more than 200,000 events per second, our strongest competitor, SASE, takes 678.3 seconds with a throughput of less than 3 events per second. Moreover, EsperTech is not capable of processing $Q_2$ over a stream of 1,000 events in an hour.
Regarding the memory consumption, we can see in the last two charts of Figure~\ref{fig:exp1} the memory consumption (in logarithmic scale) of each system before reading the last event.
The difference is again notorious; for $Q_2$ we only use 5MB of memory, SASE uses 1GB and EsperTech uses more than 10GB. It is important to say that the amount of memory used by both EsperTech and SASE is highly correlated with the amount of partial complex events, suggesting again that they are materializing elements while processing the stream. In contrast, our implementation efficiently updates a highly compressed version of the output that depends linearly in the size of the stream.
Regarding the enumeration of complex events after the last event is seen, in the third chart of Figure~\ref{fig:exp1} we draw the enumeration time. We measure this by writing all results to a freshly created native Java ArrayList for $Q_2$ at different stream lengths. Note that we did not measure this time for $Q_1$ because it was negligible, taking less than 0,2 seconds for all systems.
We can see that our framework takes $5,2$ seconds in enumerating 20,000,000 outputs (produced by 2,000 events) while SASE takes less than one second. Although SASE is more efficient in enumerating all outputs (they are already materialized in an ArrayList), our framework can still enumerate a huge number of outputs in a reasonable amount of time, specially considering that 5 seconds is irrelevant with respect to the 680 seconds that SASE requires to materialize the output while processing the stream.

\smallskip

\noindent \textbf{Using consumption policies.} In the previous experiments we detected a high correlation between the number of (partial) complex events and the amount of time and memory consumed by EsperTech and SASE. Although this does not invalidates the experiments, one could argue that generating large numbers of results is not realistic. In fact, both EsperTech and SASE have ways to speed up their matching algorithms by reducing the number of outputs. A first strategy is to use a so-called \emph{consumption policy}, a way to disregard all partial complex events whenever a complete complex event is found. Note that this does not correspond to the semantics of any of the operators or selection strategies presented in this paper: this is just a heuristic to reduce the number of results. We implemented this strategy and compared experimentally against EsperTech and SASE. For the sake of space we only present the results regarding processing time, but the memory consumption and enumeration time were also measured. Memory consumption was similar to the measurements in Figure~\ref{fig:exp1}, while enumeration was completely negligible given the reduced number of outputs.
The results are depicted in Table~\ref{tab:exp2-proc}; we present them in tabular form because the differences are too big to be appreciated visually. These experiments are performed over streams of $1,000,000$ events in which event types are uniformly distributed. The average number of complex events generated every time that the partial complex events were discarded is also depicted in Table~\ref{tab:exp2-proc}.
Since the query language of SASE does not support queries $Q_3$ and $Q_6$, the corresponding entries are left empty. We can see in these experiments that EsperTech and SASE slow down rapidly with the complexity of the query, while our system is minimally affected. Again, this occurs because our system guarantees constant time per event irrespective of the number of outputs and the complexity of the query.\par

\begin{table}\label{tab:exp2-proc}
	\small
	\centering
	\begin{tabular}{|l|c|c|c|c|c|c|}
		\hline
		& $Q_1$ & $Q_2$ & $Q_3$ & $Q_4$ & $Q_5$ & $Q_6$ \\ \hline
		Outp. avg. & $5$ & $14$ & $4$ & $19$ & $516$ & $2.064$ \\ \hline
		EsperTech & 1.84 & 3.91 & 2.21 & 166.09 & 3600* & 3600* \\ \hline
		SASE & 0.88 & 1.48 & - & 4.08 & 177.02 & - \\ \hline
		CEL & 0.27 & 0.28 & 0.27 & 0.21 & 0.26 & 0.24 \\ \hline
		
	\end{tabular}
	\caption{Processing times (sec) on a stream of 1 million events for $Q_1$-$Q_6$ using a consumption policy.}
	\vspace{-5mm}
\end{table}

\smallskip

\noindent\textbf{Selection strategies.} Although the previous experiment shows that our framework outperforms EsperTech and SASE, the number of outputs might still seem large. A way of reducing the number of outputs even more while producing meaningful results is by means of selection strategies (see Section~\ref{sec:selectors}). Selection strategies produce at most one output per event, and therefore keeping (partial) complex events in memory should not be a bottleneck anymore. To do this experiment we produce two different streams, called $S_1$ and $S_2$: $S_1$ is generated by choosing event types uniformly at random and $S_2$ is generated with distribution $\text{P}(A) = \frac{4}{10}$, $\text{P}(B) = \frac{3}{10}$, $\text{P}(C) = \frac{2}{10}$, $\text{P}(D) = \frac{1}{10}$ and $\text{P}(E) = \frac{2}{10}$ to vary the number of outputs. In this experiment we measure throughput (the number of events that each system can process per second). We leave each system processing a dynamically-generated stream for one minute. In SASE we use the \emph{skip-till-next-match} selector and in EsperTech the default matching strategy (which already generates at most one complex event per event). For our system we test the $\LAST$ and $\MAX$ selection strategies. It is important to mention that in this experiment the produced outputs were not the same. While our system follows the semantics described in Section~\ref{sec:selectors}, the competing systems behave erratically: they follow a \emph{greedy} procedure, missing some complex events that should be in the output. The results are depicted in Figure~\ref{fig:exp-3-throughput}. Note that $\MAX$ runs faster than $\NEXT$ in this case, which is because the automaton has less states (validating that the blow-up in the number of states does not always materialize).  Although the results in this case are more comparable, our system still outperforms consistently the competition while producing meaningful outputs.
\begin{figure}[tb]
\centering
\includegraphics[width=\linewidth]{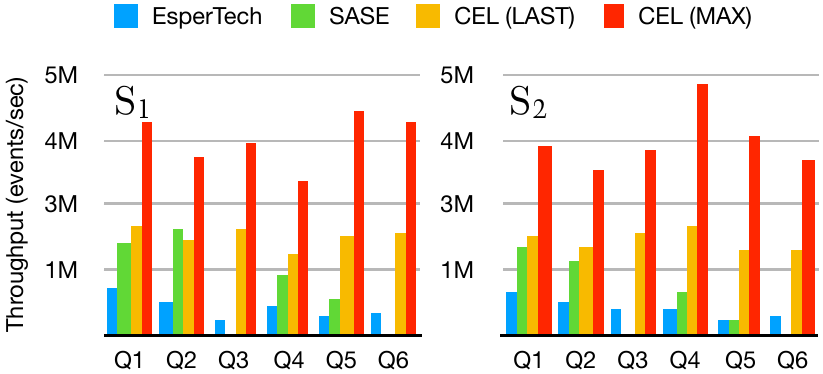}
\caption{Throughput using selection strategies over two different streams. }
\label{fig:exp-3-throughput}
\vspace*{-4mm}
\end{figure}

%% file: conclusions.tex


This paper settles new foundations for CEP systems, stimulating new research directions. In particular, a natural next step is to study the evaluation of non-unary CEL formulas. This requires new insight in rewriting formulas and a more powerful computational models with CEP evaluation algorithms. Another relevant problem is to understand the expressive power of different fragments of CEL and the relationship between the different operators. In this same direction, we envision as future work a generalization of the concept behind selection strategies, together with a thorough study of their expressive power.

Finally, we have focused on the fundamental features of CEP languages, leaving other features outside to keep the language and analysis simple.
These features include correlation, time windows, aggregation, consumption policies, among others (see~\cite{cugola2012} for a more exhaustive list). 
We believe that CEL can be extended with these features to establish a more complete framework for CEP. 


%

%% file: appendix-selectors.tex
\subsection{Proof of Lemma~\ref{lemma:next-order}}

For $\leqnext$ to be a total order between complex events, it has to be \textit{reflexive} (trivial),  \textit{anti-symmetric}, \textit{transitive}, and \textit{total}. The proof for each property is given next.
\medskip

\noindent \emph{Anti-symmetric.}
Consider any two complex events $C_1$ and $C_2$ such that $C_1 \leqnext C_2$ and $C_2 \leqnext C_1$.
$C_2 \leqnext C_1$ means that either $C_1 = C_2$ or $(1) \enspace \min(C_1 \triangle C_2)\in C_1$, and $C_1 \leqnext C_2$ that either $C_2 = C_1$ or $(2) \enspace \min(C_1 \triangle C_2)\in C_2$.
If $(1)$ were true, it would mean that $(2)$ could not be true, so $C_2 = C_1$ would have to be true, becoming a contradiction.
So, the only possible scenario is that $C_1=C_2$.
\medskip

\noindent\emph{Transitivity.}
Consider any three complex events $C_1$, $C_2$ and $C_3$ such that $C_1 \leqnext C_2$ and $C_2 \leqnext C_3$.
Because $C_1 \leqnext C_2$ holds, then either $C_1 = C_2$ or $(1) \enspace \min(C_1 \triangle C_2)\in C_2$.
If $C_1 = C_2$, then $C_1 \leqnext C_3$ because $C_2 \leqnext C_3$.
Now, if $C_1 \neq C_2$, then $(1)$ must hold, which implies that the lowest element that is either in $C_1$ or $C_2$, but not in both, has to be in $C_2$.
Let's call this element $l_1$.
Because $C_2 \leqnext C_3$, then either $C_2 = C_3$ or $(2) \enspace \min(C_2 \triangle C_3)\in C_3$.
Again, if $C_2 = C_3$, then $C_1 \leqnext C_3$ because $C_1 \leqnext C_2$.
Now, if $C_2 \neq C_3$, then $(2)$ must hold, so the lowest element that is either in $C_2$ or $C_3$, but not in both, has to be in $C_3$. Let's call this element $l_2$.
	
Given that $C_1 \neq C_2$ and $C_2 \neq C_3$, define for every $i \in \{1,2,3\}$ and $j \in \{1,2\}$ the set $C_i^{<l_j}$ as the set of elements of $C_i$ which are lower than $l_j$, i.e., $C_i^{<l_j}=\{x \mid x\in C_i \land x < l_j \}$.
It is clear that $C_1^{<l_1}=C_2^{<l_1}$ and $C_2^{<l_2}=C_3^{<l_2}$, because of  $(1)$ and $(2)$, respectively.
Also, because of $(2)$ it holds that $l_2 \notin C_2$, so $l_1 \neq l_2$.
	
Consider first the case where $l_1 < l_2$.
This means that $(3) \enspace C_1^{<l_1} = C_3^{<l_1}$. Moreover, if $l_1$ were not in $C_3$, it would contradict $(2)$, so $(4) \enspace l_1 \in C_3$ must hold.
With $(3)$ and $(4)$, it follows that $l_1$ is the lowest element that is either in $C_1$ or $C_3$ but not in both, and it is in $C_3$.
This proves that $\min(C_1 \triangle C_3)\in C_3$, and thus $C_1 \leqnext C_3$.
	
Now consider the case where $l_2 < l_1$.
Then, $(5) \enspace C_1^{<l_2} = C_3^{<l_2}$ must hold. Because $l_2$ is not in $C_2$, it cannot be in $C_1$, otherwise it would contradict $(1)$, so $(6) \enspace l_2 \notin C_1$ must hold.
Also, because of $(2)$ we know that $(7) \enspace l_2 \in C_3$ must hold.
With $(5)$, $(6)$ and $(7)$, it follows that $l_2$ is the lowest element that is either in $C_1$ or $C_3$ but not in both, and it is in $C_3$.
This proves that $\min(C_1 \triangle C_3)\in C_3$, and thus $C_1 \leqnext C_3$.
\medskip

\noindent \emph{Total.}
Consider any two complex events $C_1$ and $C_2$. If $C_1=C_2$, then $C_1 \leqnext C_2$ holds.
Consider now the case where $C_1 \neq C_2$.
Define the set $C = (C_1\cup C_2) \backslash (C_1\cap C_2)$ which is the set of all elements either in $C_1$ or $C_2$, but not in both.
Because $C_1 \neq C_2$, there must be at least one element in $C$.
In particular, this implies that there is a minimum element $l$ in $C$.
If $l$ is in $C_2$, then $C_1 \leqnext C_2$ holds, and if $l$ is in $C_1$, then $C_2 \leqnext C_1$ holds.
\hfill \qed 

%% file: appendix-language-prop.tex

\subsection{Proof of Theorem~\ref{theo:safe}}

To prove this theorem, we first show that one can push disjunction (by means of $\scor$) to the top-most level of every core-CEL formula. Formally, we say that a CEL formula $\varphi$ is in disjunctive-normal form if $\varphi=(\varphi_1\cor\cdots\cor\varphi_n)$, where for each $i\in\{1,\ldots,n\}$, it is the case that:
\begin{itemize}
	\item Every $\scor$ operator in $\varphi_i$ occurs in the scope of a $+$ operator.
	\item For every subformula of $\varphi_i$ of the form $(\varphi_i')+$, it is the case that $\varphi_i'$ is in disjunctive normal form.
\end{itemize}

Now we show that every formula can be translated into disjunctive normal form.
\begin{lemma}\label{lem:dnf}
	Every formula $\varphi$ in core-CEL can be translated into disjunctive-normal form in time at most exponential $|\varphi|$.
\end{lemma}
\begin{proof}
	We proceed by induction over the structure of $\varphi$.
	\begin{itemize}
		\item If $\varphi=R\as x$, then $\varphi$ is already free of $\scor$.
		
		\item If $\varphi=\varphi_1\cor\varphi_2$, the result readily follows from the induction hypothesis.

		\item If $\varphi=(\varphi')+$, by induction hypothesis $\varphi$ can be translated into disjunctive normal form.
		
		\item If $\varphi=\varphi'\FILTER P(\bar{x})$ with $\bar{x} = (x_1,\ldots,x_k)$, we know by induction hypothesis that $\varphi'$ is equivalent to a formula $(\varphi_1\cor\cdots\cor\varphi_n)$. Therefore, $\varphi$ is equivalent to $(\varphi_1\cor\cdots\cor\varphi_n)\FILTER P(\bar{x})$. We show that this latter formula is equivalent to $(\varphi_1\FILTER P(\bar{x}))\cor\cdots\cor(\varphi_n\FILTER P(\bar{x}))$. Let $S$ be a stream and assume $C\in\sem{(\varphi_1\cor\cdots\cor\varphi_n)\FILTER P(\bar{x})}(S)$. Then, there is some $\nu$ such that $C\in\sem{(\varphi_1\cor\cdots\cor\varphi_n)}(S,\nu)$ and $(S[\nu(x_1)],\ldots,S[\nu(x_k)]) \in P(\bar{x})$. By definition of $\scor$, this implies that there is an $i\in\{1,\ldots,n\}$ such that $C\in\sem{(\varphi_i)}(S,\nu)$. As $(S[\nu(x_1)],\ldots,S[\nu(x_k)])\in P$, we have $C\in\sem{(\varphi_i)\FILTER P(\bar{x})}(S,\nu)$. We can then immediately conclude that $C\in\sem{(\varphi_1\FILTER P(\bar{x}))\cor\cdots\cor(\varphi_n\FILTER P(\bar{x}))}(S,\nu)$, and thus $C\in\sem{(\varphi_1\FILTER P(\bar{x}))\cor\cdots\cor(\varphi_n\FILTER P(\bar{x}))}(S)$. The converse follows from an analogous argument.
		
		\item If $\varphi=(\varphi_1\sq\varphi_2)$, by induction hypothesis we know that $\varphi_1$ is equivalent to a formula $(\varphi^1_1\cor\cdots\cor\varphi^1_n)$ and $\varphi_2$ is equivalent to a formula $(\varphi^2_1\cor\cdots\cor\varphi^2_m)$. Let $\varphi'$ be defined by
		\[
		\varphi'= (\varphi^1_1\sq\varphi^2_1)\cor (\varphi^1_1\sq\varphi^2_2)\cor\cdots\cor(\varphi^1_1\sq\varphi^2_m)\cor(\varphi^1_2\sq\varphi^2_1)\cor\cdots\cor (\varphi^1_2\sq\varphi^2_m)\cor\cdots\cor(\varphi^1_n\sq\varphi^2_1)\cor
		\cdots\cor (\varphi^1_n\sq\varphi^2_m).
		\]
		We show that $\varphi\equiv\varphi'$. Let $S$ be a stream and let $C$ be a complex event. If $C\in\sem{\varphi}(S)$, then there is a valuation $\nu$ and two complex events $C_1$ and $C_2$ such that $C=C_1\cdot C_2$, $C_1\in\sem{\varphi_1}(S,\nu)$ and $C_2\in\sem{\varphi_2}(S,\nu)$. Then, there are two numbers $i$ and $j$ such that $C_1\in\sem{\varphi^1_i}(S,\nu)$ and $C_2\in\sem{\varphi^2_j}(S,\nu)$. As $C=C_1\cdot C_2$, it immediately follows that $C\in\sem{\varphi^1_i\sq\varphi^2_j}(S)$, and thus $C\in\sem{\varphi'}(S)$.\par
		For the converse assume $C\in\sem{\varphi'}(S)$. Then, there is a valuation $\nu$, a complex event $C$ and two numbers $i$ and $j$ such that $C\in\sem{\varphi^1_i\sq\varphi^2_j}(S,\nu)$. Therefore there are two complex events $C_1$ and $C_2$ such that $C=C_1\cdot C_2$, $C_1\in\sem{\varphi^1_i}(S,\nu)$ and $C_2\in\sem{\varphi^2_j}(S,\nu)$. By semantics of $\scor$, we have $C_1\in\sem{\varphi_1}(S,\nu)$ and $C_2\in\sem{\varphi_2}(S,\nu)$. As $C=C_1\cdot C_2$, it readily follows that $C\in\sem{\varphi_1\sq\varphi_2}(S)=\sem{\varphi}(S)$.
	\end{itemize}
\end{proof}

Having this result, we proceed to show that a core-CEL formula in disjunctive normal form can be translated into a safe formula. To this end, we need to show the following two lemmas.

\begin{lemma}\label{lem:var-in-ce}
	Let $\varphi$ be a core-CEL formula in which every $\scor$ occurs inside the scope of a $+$ operator, and let $x\in\vdefplus(\varphi)$. Then, for every complex event $C$, valuation $\nu$ and stream $S$ such that $C\in\sem{\varphi}(S,\nu)$, it is the case that $x\in\dom(\nu)$ and $\nu(x)\in C$.
\end{lemma}
\begin{proof}
	We proceed by induction on the structure of $\varphi$. Let $\nu$ be a valuation, $S$ a stream and $C$ a complex event.
	\begin{itemize}
		\item Assume $\varphi=R\as x$ and that $C\in\sem{\varphi}(S,\nu)$. By definition, we have $C=\{\nu(x)\}$.
		\item Assume $\varphi=\varphi'\FILTER P(\bar{x})$ and that $C\in\sem{\varphi}(S,\nu)$. Let $x\in\vdefplus(\varphi)$. By definition, we have that $C\in\sem{\varphi'}(S,\nu)$. Since $x\in\vdefplus(\varphi')$, by induction hypothesis we have $x\in\dom(\nu)$ and $\nu(x)\in C$.
		\item If $\varphi=(\varphi')+$ the condition trivially holds as $\vdefplus(\varphi)=\emptyset$.
		\item If $\varphi=\varphi_1\sq\varphi_2$, then $x\in\vdefplus(\varphi_1)$ or $x\in\vdefplus(\varphi_2)$. Assume w.l.o.g. that $x\in\vdefplus(\varphi_1)$. If $C\in\sem{\varphi}(S,\nu)$, then $C=C_1\cdot C_2$, where $C_1\in\sem{\varphi_1}(S,\nu)$. As $x\in\vdefplus(\varphi_1)$, by induction hypothesis we have that $x\in\dom(\nu)$ and $\nu(x)\in C_1\subseteq C$, concluding the proof.
	\end{itemize}
\end{proof}

\begin{lemma}\label{lem:dnf-satis}
	Let $\varphi$ be a core-CEL formula in which every $\scor$ occurs inside the scope of a $+$ operator, and let $S$ be a stream. If $\varphi$ has a subformula $\varphi'$ that is not under the scope of a $+$ operator such that $\sem{\varphi'}(S)=\emptyset$, then $\sem{\varphi}(S)=\emptyset$.
\end{lemma}
\begin{proof}
	We proceed by induction on the structure of $\varphi$. Let $S$ a stream and assume $\varphi'$ is a subformula of $\varphi$ such that $\sem{\varphi'}(S)=\emptyset$. We assume that $\varphi'$ is a proper subformula, as otherwise the result immediately follows. For this reason, we can trivially skip the case when $\varphi=R\as x$ or $\varphi=(\varphi_1)+$.
	\begin{itemize}
		\item If $\varphi=\varphi_1\sq\varphi_2$, then $\varphi'$ is a subformula of $\varphi_1$ or of $\varphi_2$. Assume w.l.o.g. that $\varphi'$ is a subformula of $\varphi_1$. By induction hypothesis, as $\sem{\varphi'}(S)=\emptyset$ we have that $\sem{\varphi_1}(S)=\emptyset$, which immediately implies that $\sem{\varphi}(S)=\emptyset$.
		\item If $\varphi=\varphi_1\FILTER P(\bar{x})$, we know that $\varphi'$ is a subformula of $\varphi_1$. By induction hypothesis we have $\sem{\varphi'}(S)=\emptyset$ and by definition of $\FILTER$ we obtain $\sem{\varphi}(S)=\emptyset$.
	\end{itemize}
\end{proof}

Now we are ready to show that any core-CEL formula in disjunctive-normal form can be translated into a safe formula, and moreover, this can be done in linear time.

\begin{lemma}\label{lem:dnf-to-safe}
	Let $\varphi$ be a core-CEL formula in disjunctive-normal form. Then $\varphi$ can be translated in linear time into a safe core-CEL formula $\varphi'$.
\end{lemma}
\begin{proof}
	Assume that $\varphi=\varphi_1\cor\cdots\cor\varphi_n$ is a core-CEL formula in disjunctive-normal form. By induction, we assume that every sub-formula of the form $(\varphi')+$ is already safe. Now we show that every unsafe $\varphi_i$ is unsatisfiable, and therefore it can be safely removed from the disjunction. Proceed by contradiction and assume $\varphi_i$ is unsafe and satisfiable. Then, it must contain a subformula of the form $\psi_1\sq\psi_2$ occurring outside the scope of all $+$ operators, and such that $\vdefplus(\psi_1)\cap\vdefplus(\psi_2)\neq\emptyset$. Let $x\in\vdefplus(\psi_1)\cap\vdefplus(\psi_2)$. By Lemma~\ref{lem:dnf-satis}, we know that $\psi_1\sq\psi_2$ must be satisfiable. Therefore, there is a stream $S$, a valuation $\nu$ and a mapping $C$ such that $C\in\sem{\psi_1\sq\psi_2}(S,\nu)$. This implies the existence of two complex events $C_1$ and $C_2$ such that $C_1\in\sem{\psi_1}(S,\nu)$ and $C_2\in\sem{\psi_2}(S,\nu)$. Since $x\in\vdefplus(\psi_1)$ and $\psi_1$ can only mention $\scor$ inside a $+$ operator, by Lemma~\ref{lem:var-in-ce} we obtain that $\nu(x)\in C_1$. Similarly, as $x\in\vdefplus(\psi_2)$, we have $\nu(x)\in C_2$. But as $C=C_1\cdot C_2$, we have that $C_1\cap C_2=\emptyset$, contradicting the facts that $\nu(x)\in C_1$ and $\nu(x)\in C_2$.\par

	We have obtained that if any disjunct is unsafe, it cannot produce any results. Therefore, as safeness is easily verifiable, the result readily follows by removing the unsafe disjuncts of $\varphi$. Notice that this need to be done in a bottom-up fashion, starting from the subformulas of the form $(\varphi')+$.
\end{proof}

Theorem~\ref{theo:safe} occurs as a corollary of Lemmas~\ref{lem:dnf} and \ref{lem:dnf-to-safe}. Indeed, given a core-CEL formula $\varphi$, one can construct in exponential time an equivalent core-CEL formula $\varphi'$ in disjunctive normal form. Then, from $\varphi'$ one can construct in linear time a safe formula in core-CEL $\psi$ that is equivalent to $\varphi$, which is exactly what we wanted to show.
\hfill \qed

\subsection{Proof of Theorem~\ref{theo:LP-normal-form}}

Without lost of generality, in the proof we consider only unary predicates, since these are the ones that we need to modify in order for the formula to be in LP-normal form.
Indeed, if the formula contains non-unary filters, this can be treated as normal operators, similar than $\scor$ or $\sq$ operators. 
Consider a well-formed core-CEL formula $\varphi$ with unary predicates.
We first provide a construction for a core-CEL formula in LP normal form and then prove that it is equivalent to $\varphi$.
The construction consists of two steps: $(1)$~pop predicates up, and $(2)$~push predicates down.
\begin{enumerate}
	\item The first step is focused on rewriting the formula in a way that for every subformula of the form $\varphi' \FILTER P(x)$ it holds that $x \in \vdefplus(\varphi')$.
	Recall that a well-formed formula could still have a subformula $\varphi' \FILTER P(x)$ such that $x \in \vdefplus(\varphi')$.
	The construction we provide to achieve this is the following.
	For every subformula of the form $\varphi'\FILTER P(x)$ and every predicate, let $\varphi_x$ be the lowest subformula of $\varphi$ where $x$ is defined and that has $\varphi'$ as a subformula.
	Here we use the fact that $\varphi$ is well-formed, which ensures that $\varphi_x$ must exist.
	Then, we rewrite the subformula $\varphi_x$ inside $\varphi$ as $\varphi_x^t \FILTER P(x) \cor \varphi_x^f \FILTER \lnot P(x)$, where $\varphi_x^t$ and $\varphi_x^f$ are the same as $\varphi_x$ but replacing the inside $P(x)$ with $\TRUE$ and $\FALSE$, respectively.
	\item Now that we moved each predicate up to a level where all its variables are defined, the next step is to move each one down to its variable's definition.
	This is done straightforward: for every subformula of the form $\varphi' \FILTER P(x)$, the $P(x)$ filter is removed from $\varphi'$ and instead applied over every subformula of $\varphi'$ with the form $R \as x$, rewriting it as $R \as x \FILTER P(x)$.
	After all predicates was moved to the lowest possible, each assignment $R \as x$ now has a sequence of filters applied to it, e.g. $R \as x \FILTER P_1(x) \ldots \FILTER P_k(x)$, and moreover, all filters appear in this form.
	Because the predicate set $\pset$ is closed under intersection, we know there is some $P \in \pset$ that equals $P_1 \cap \ldots \cap P_k$.
	Then, we replace each sequence of filters $R \as x \FILTER P_1(x) \ldots \FILTER P_k(x)$ with $R \as x \FILTER P(x)$, thus resulting in a formula in LP-normal form.
\end{enumerate}

Now we prove that the construction above satisfies the lemma, i.e., $\sem{\varphi_{lp}}(S) = \sem{\varphi}(S)$ for every stream $S$, where $\varphi_{lp}$ is the resulting formula after the construction.
To prove that the first step does not change the semantics, we show that it stays the same after each iteration.
Consider a subformula $\varphi' \FILTER P(x)$ of $\varphi$ such that $x \notin \vdefplus(\varphi')$.
In particular, the only part of $\varphi$ that is modified by the algorithm is $\varphi_x$, so it suffices to prove that $C \in \sem{\varphi_x}(S,\nu)$ holds iff $C \in \sem{\varphi_x^t \FILTER P(x) \cor \varphi_x^f \FILTER \lnot P(x)}(S,\nu)$.
\begin{itemize}
	\item For the only-if direction, let $S$, $C$, $\nu$ be any stream, complex event and valuation, respectively, such that $C \in \sem{\varphi_x}(S,\nu)$.
	If $S[\nu(x)] \in P$, then it is enough to prove that $C \in \sem{\varphi_x^t}(S,\nu)$.
	In a similar way, the only part in which $\varphi_x^t$ differs with $\varphi_x$ is that in the former the atom $P(x)$ was set to $\TRUE$.
	Therefore, it is enough to prove that, for any $C'$ and $\varphi'$, if $S[\nu'(x)] \in P$ holds, then $C' \in \sem{\varphi' \FILTER P(x)}(S,\nu')$ iff $C' \in \sem{\varphi' \FILTER \TRUE}(S,\nu')$, which is trivially true.
	Notice that we can assure $S[\nu'(x)] \in P$ holds because $S[\nu(x)] \in P$ holds and, when evaluating this part of the formula, the mapping for $x$ must stay the same, otherwise $x$ must have been inside a $\ks$-operator, which cannot be the case because $x \in \vi(\varphi_x)$.
	Moreover, $\nu'$ has to be equal to $\nu$.
	The proof for the case $S[\nu(x_1)] \in \lnot P$ is similar considering $\varphi_x^f$ instead of $\varphi_x^t$, thus $C \in \sem{\varphi_x^t \FILTER P(x) \cor \varphi_x^f \FILTER \lnot P(x)}(S,\nu)$.
	\item For the if direction, let $S$, $C$, $\nu$ be some arbitrary stream, complex event and valuation, respectively, such that $C \in \sem{\varphi_x^t \FILTER P(x) \cor \varphi_x^f \FILTER \lnot P(x)}(S, \nu)$.
	Then, by definition the complex event $C$ is either in $\sem{\varphi_x^t\FILTER P(x)}(S, \nu)$ or in $\sem{\varphi_x^f\FILTER \lnot P(x)}(S, \nu)$.
	Without loss of generality, consider the former case, which implies that $S[\nu(x)] \in P$.
	Then, because $C' \in \sem{\varphi' \FILTER P(x)}(S, \nu')$ iff $C' \in \sem{\varphi' \FILTER \TRUE}(S, \nu')$, it holds that $C \in \sem{\varphi_x}(S, \nu)$.
	It is the same for $S[\nu(x)] \in \lnot P$, thus $C \in \sem{\varphi_x}(S, \nu)$ iff $C \in \sem{\varphi_x^t \FILTER P(x) \cor \varphi_x^f \FILTER \lnot P(x)}(S, \nu)$.
\end{itemize}
Therefore, if we name $\varphi_1$ as the result of applying the first part, we get that $\varphi_1 \equiv \varphi$.
%

Now, we prove that moving the predicates to their definitions does not affect the semantics either, for which we show that it stays the same after each iteration.
Consider a subformula of $\varphi_1$ of the form $\varphi' \FILTER P(x)$.
The same way as before, we focus on the modified part, i.e., we need to prove that $C \in \sem{\varphi' \FILTER P(x)}(S, \nu)$ iff $C \in \sem{\varphi'_P}(S, \nu)$, where $\varphi'_P$ is the result of adding the filter $P(x)$ for each definition of $x$ inside $\varphi'$, i.e., replace $R \as x$ with $R \as x \FILTER P(x)$ where $R$ is any relation.
\begin{itemize}
	\item First, we show the only-if direction.
	Let $S$, $C$, $\nu$ be any stream, complex event and valuation, respectively, such that $C \in \sem{\varphi' \FILTER P(x)}(S, \nu)$, which implies that $S[\nu(x)] \in P$.
	We know that, when evaluating every subformula $R \as x$ of $\varphi'$, the valuation $\nu$ must stay the same, because $x \in \vi(\varphi')$, and thus its definition cannot be inside a $\ks$-operator (notice that if it appears inside a $\ks$, it represents a value different to $x$, thus the $\ks$ subformula can be rewritten using a new variable $x'$).
	Similarly to the reasoning above, it holds that for any $C'$ and $\varphi'$, if $S[\nu'(x)] \in P$, then $C' \in \sem{R \as x \FILTER P(x)}(S, \nu')$ iff $C' \in \sem{R \as x}(S, \nu')$.
	Then, because every subformula $R \as x$ behaves the same, $C \in \sem{\varphi'_P}(S, \nu)$ holds.
	\item We now show the if direction.
	Let $S$, $C$, $\nu$ be any stream, complex event and valuation, respectively, such that $C \in \sem{\varphi'_P}(S, \nu)$.
	We prove that $S[\nu(x)] \in P$ must hold, thus proving that $C \in \sem{\varphi' \FILTER P(x)}(S, \nu)$ holds using the same argument as above.
	By contradiction, assume that $S[\nu(x)] \notin P$.
	Because we showed that when evaluating every $R \as x \FILTER P(x)$ in $\varphi'_P$, the valuation $\nu$ must be the same, the only possible way for $C \in \sem{\varphi'_P}(S, \nu)$ to hold is if all $R \as x$ appear at one side of an $\cor$-operator.
	However, this would contradict the fact that $x \in \vi(\varphi')$, thus $S[\nu(x)] \in P$, and also $C \in \sem{\varphi' \FILTER P(x)}(S, \nu)$.
\end{itemize}
Then, $\varphi' \FILTER P(x)$ and $\varphi'_P$ are equivalent, therefore, if we name $\varphi_{lp}$ the result of applying step~2, we get that $\varphi_{lp} \equiv \varphi_1 \equiv \varphi$.

Finally, it is easy to check that the size of $\varphi_{lp}$ will be at most exponential in the size of $\varphi$.
Each iteration of step~1 could duplicate the size of the formula in the worst case, thus $|\varphi_1| = \cO(2^{|\varphi|})$.
Then, step~2 does not really increase the size of the formula due to the final replacing of predicates $P_1,\ldots, P_k$ with $P$.
The size of $\varphi_{lp}$ w.r.t. $\varphi$ is then $\cO(2^{|\varphi|})$.
However, in our framework (Section~\ref{sec:prelim}) we assumed that $\varphi$ did not use the syntactic sugar $\land$ and $\lor$ inside its filters.
If so, we argue that this would not turn into an extra exponential growth (turning the result to a double-exponential).
To explain why, consider that $\varphi$ uses the $\land$ and $\lor$ syntactic sugar.
Then, if we apply step~1 to each predicate, the resulting formula $\varphi_1$ would still be equivalent to $\varphi$ and of size at most exponential w.r.t. $|\varphi|$, avoiding the double-exponential blow-up mentioned above.
Finally, we have that $|\varphi_{lp}| = \cO(2^{|\varphi|})$, even if $\varphi$ uses $\land$ and $\lor$.
\hfill \qed

%% file: appendix-automata.tex
\medskip

\subsection{Proof of Proposition~\ref{prop:CEA_closure}}

For the following proof consider any two CEA $\cA_1 = (Q_1,\Delta_1,I_1,F_1)$, $\cA_2 = (Q_2,\Delta_2,I_2,F_2)$ and assume, without loss of generality, that they have disjoint sets of states, i.e., $Q_1 \cap Q_2 = \emptyset$.
We begin by proving closure under union, which is exactly the same as the proof for FSA closure under union.
We define the CEA $\cA_1\cup\cA_2=(Q,\Delta,I,F)$ as follows. The set of states is $Q=Q_1\cup Q_2$, the transition relation is $\Delta=\Delta_1 \cup \Delta_2$; the set of initial states is $I=I_1 \cup I_2$ and the set of final states is $F=F_1 \cup F_2$.

Next we prove closure under intersection.
We define the CEA $\cA_1\cap\cA_2=(Q,\Delta,I,F)$ as follows.
The set of states is the Cartesian product $Q= Q_1 \times Q_2$; the transition relation is $\Delta=\{((p_1,p_2),(P_1 \land P_2,m),(q_1,q_2)) \mid \text{$(p_i,(P_i,m),q_i) \in \Delta_i$ for $i\in \{1,2\}$}\}$, that is, the incoming tuple must me inside both predicates $P_1$ and $P_2$ in order to simulate both transitions with the same mark $m$ from $p_1$ to $q_1$ and from $p_2$ to $q_2$ of $\cA_1$ and $\cA_2$, respectively; the set of initial states is $I=I_1 \times I_2$ and the set of final states is $F = F_1 \times F_2$.

Now we prove closure under I/O-determinization.
Define the CEA $\cA_d = (Q_d, \delta_d, I_d, F_d)$ component by component.
First, the set of states is $Q_d = 2^Q$, that is, each state in $Q_d$ represents a different subset of $Q$.
Second, the transition relation is: 
$$
\delta_d = \{(T,(P,m),U) \mid \text{$P \in \types{\cP}$, and $q \in U$ iff there is a $p \in T$ and $P' \in \uset$ such that $(p, (P',m), q) \in \Delta$ and $P \subseteq P'$}\}.
$$
Here, $\cP$ is the set of all predicates in the transitions of $\Delta$ and we use the notion of $\types{\cP}$ defined in the proof of Theorem~\ref{theo:selectors-compilation} (see Section~\ref{proof:next} for the definition).
Finally, the sets of initial and final states are $I_d = \{I\}$ and $F_d = \{T \mid T \in Q_d \land T \cap F \neq \emptyset\}$.
The key notion here is the one of $\types{\cP}$, which partitions the set of all tuples in a way that if a tuple $t$ satisfies a predicate $P_t \in \types{\cP}$, then $P_t$ is a subset of the predicates of all transition that a run of $\cA$ could take when reading $t$.
This allows us to then apply a determinization algorithm similar to the one for FSA.
Notice that $P_1 \cap P_2 = \emptyset$ for every two different predicates $P_1, P_2 \in \types{\cP}$, so the resulting CEA $\cA_d$ is I/O-deterministic.

Finally, we prove closure under complementation.
Basically, the complementation of a CEA is no more than determinizing it and complementing the set of final states.
Formally, we define the CEA $\cA_1^c=(Q,\delta,I,F)$ as follows.
Consider the I/O-deterministic CEA $\det(\cA_1) = (Q_d,\delta_d,I_d,F_d)$.
Then, the set of states, the transition relation and the set of initial states are the same as of $\det(\cA_1)$, i.e., $Q=Q_d$, $\delta = \delta_d$ and $I = I_d$, and the set of final states is $F = Q \setminus F_d$.
\hfill \qed

\medskip

\subsection{Proof of Theorem~\ref{theo:core-automata}}

So simplify the proof, we will add to the model of CEA the ability to have $\epsilon$-transitions.
Formally, now a transition relation has the structure $\Delta \subseteq Q \times ((\uset \times \{\amark, \umark\}) \cup \{\epsilon\}) \times Q$.
This basically means the automaton can have transitions of the form $(p,\epsilon,q)$ that can be part of a run and, if so, the automaton passes from state $p$ to $q$ without reading nor marking any new tuple.
This does not give any additional power to CEA, since any $\epsilon$-transition $(p,\epsilon,q)$ can be removed by adding, for each incoming transition of $p$, an equivalent incoming one to $q$, and for each outgoing transition of $q$ an equivalent outgoing one from $p$.

The results of Theorem~\ref{theo:LP-normal-form} and Theorem~\ref{theo:safe} show that we can rewrite every core-CEL formula as a safe formula in LP-normal form.
We consider that, if $\varphi$ is not in LP-normal form, then it is first turned into one that is, adding an exponential growth from the beginning.
Furthermore, if it is not safe the it is turned into a safe one, adding another exponential growth.
We now give a construction that, for every safe core-CEL formula $\varphi$ in LP-normal form, defines a CEA $\cA$ such that for every complex event $C$, $C \in \sem{\cA}(S)$ iff $C \in \sem{\varphi}(S)$.
This construction is done recursively in a bottom-up fashion such that, for every subformula, an equivalent CEA is built from the CEA of its subformulas.
Moreover, we assume that the CEA for each subformula has one initial state and one final state, since each recursive construction defines a CEA with those properties.
Let $\psi$ be a subformula of $\varphi$.
Then, the CEA $\cA$ is defined as~follows:
\begin{itemize}
	\item If $\psi = R \as x \FILTER P(x)$ then $\cA = (Q, \Delta, \{q^i\}, \{q^f\})$ with the set of states $Q=\{q^i,q^f\}$ and the transitions $\Delta = \{(q^i,(\TRUE,\umark),q^i), (q^i,(P',\amark),q^f)\}$, where $P'(x) = (\type(x)=R) \land P(x)$.
	Graphically, the automaton is:
	\begin{center}
		\begin{tikzpicture}[->,>=stealth, 
		semithick, 
		auto,
		initial text= {},
		initial distance= {3mm},
		accepting distance= {4mm},node distance=2.5cm, semithick]
		\tikzstyle{every state}=[draw=black,text=black,inner sep=0pt, minimum size=8mm]
		
		\node[initial,state]	(1) 				{$q^i$};
		\node[accepting,state]	(2) [right of=1]	{$q^f$};
		
		\path
		(1) edge 				node {$P'(x) \mid \amark$} (2)
		edge [loop above] 		node {$\TRUE \mid \umark$} (1);
		\end{tikzpicture}
	\end{center}
	If $\psi$ has no $\FILTER$ the automaton is the same but with $P'(x)= (\type(x)=R)$.
	
	\item If $\psi = \psi_1 \cor \psi_2$, and $\cA_1 = (Q_1, \Delta_1, \{q^i_1\}, \{q^f_1\})$ and $\cA_2 = (Q_2, \Delta_2, \{q^i_2\}, \{q^f_2\})$ are the CEA for $\psi_1$ and $\psi_2$, respectively, then $\cA = (Q, \Delta, \{q^i\}, \{q^f\})$ where $Q$ is the union of the states of $\cA_1$ and $\cA_2$ plus the new initial and final states $q^i, q^f$, and $\Delta$ is the union of $\Delta_1$ and $\Delta_2$ plus the empty transitions from $q^i$ to the initial states of $\cA_1$ and $\cA_2$, and from the final states of $\cA_1$ and $\cA_2$ to $q^f$.
	Formally, $Q = Q_1 \cup Q_2 \cup \{q^i, q^f\}$ and $\Delta = \Delta_1 \cup \Delta_2 \cup \{(q^i, \epsilon, q^i_1), (q^i, \epsilon, q^i_2), (q^f_1, \epsilon, q^f), (q^f_2, \epsilon, q^f) \}$.
	
	\item If $\psi = \psi_1 \sq \psi_2$, consider that $\cA_1 = (Q_1, \Delta_1, \{q^i_1\}, \{q^f_1\})$ and $\cA_2 = (Q_2, \Delta_2, \{q^i_2\}, \{q^f_2\})$ are the CEA for $\psi_1$ and $\psi_2$, respectively.
	Then, we define $\cA = (Q, \Delta, \{q^i_1\}, \{q^f_2\})$, where the set of states is $Q = Q_1 \cup Q_2$ and the transition relation is $\Delta = \Delta_1 \cup \Delta_2 \cup \{(q^f_1,\epsilon,q^i_2)\}$.
	
	\item If $\psi = \psi_1 \ks$, consider that $\cA_1 = (Q_1, \Delta_1, \{q^i_1\}, \{q^f_1\})$ is the automaton for $\psi_1$.
	Then, we define $\cA = (Q_1, \Delta, \{q^i_1\}, \{q^f_1\})$ where $\Delta = \Delta_1 \cup \{(q^f_1, \epsilon, q^i_1) \}$.
	Basically, is the same automaton for $\psi_1$ with an $\epsilon$-transition from the final to the initial state.
\end{itemize}

Now, we need to prove that the previous construction satisfies Theorem~\ref{theo:core-automata}.
We will prove this by induction over the subformulas of $\varphi$, i.e., assume as induction hypothesis that the theorem holds for any subformula $\psi$ and its respective CEA~$\cA$.

First, consider the base case $\psi = R \as x \FILTER P(x)$.
If $C \in \sem{\cA}(S)$ then there is a run $\rho$ that gets to the accepting state such that $\rmatch(\rho) = C$.
Moreover, $\rho$ must pass through the transition $(q^i,(\type(x) = R \land P(x),q^f)$ while reading a tuple $t_j$ at some position $j$.
Then, consider a valuation $\nu$ such that $\nu(x) = j$.
Clearly, $C = \{\nu(x)\}$, $\type(t_j) = R$, and $S[\nu(x)] \in P$, thus $C \in \sem{\psi}(S,\nu)$.
For the other direction, consider that $C \in \sem{\psi}(S, \nu)$ for some valuation $\nu$.
Then $C$ must contain only one position $j = \nu(x)$ such that $\type(S[j]) = R$ and $S[j] \in P$ hold.
Then $\rho= (q^i, (\TRUE, \umark), q^i)^{j} \cdot (q^i, (P'(x), \amark), q^f)$ is an accepting run of $\cA$ over $S$, where $(q^i, (\TRUE, \umark), q^i)^{j}$ means that it takes the initial loop transition $j$ times.
Because $\rmatch(\rho) = \{j\} = C$, then $C \in \sem{\cA}(S)$.

Now, consider the case $\psi = \psi_1 \cor \psi_2$.
If $C \in \sem{\cA}(S)$, then there is an accepting run $\rho$ that also represents either an accepting run of $\cA_1$ or $\cA_2$ (removing the $\epsilon$ transitions at the beginning and end).
Assume w.l.o.g. that it is the former case.
Then, by induction hypothesis, there is a valuation $\nu$ such that $C \in \sem{\psi_1}(S,\nu)$.
By definition this means that $C \in \sem{\psi}(S,\nu)$.
For the other direction, consider that $C \in \sem{\psi}(S, \nu)$ for some valuation $\nu$.
Then, either $C \in \sem{\psi_1}(S, \nu)$ or $C \in \sem{\psi_2}(S, \nu)$ holds.
Without loss of generality, consider the former case.
By induction hypothesis, it means that $C \in \sem{\cA_1}(S)$, so there is an accepting run $\rho'$ of $\cA_1$ over $S$ such that $\rmatch(\rho')=C$.
Because $\Delta$ contains $\Delta_1$ then the run $\rho=(q^i, \epsilon, q^i_1) \cdot \rho' \cdot (q^f_1, \epsilon, q^f)$ is an accepting run of $\cA$ over $S$.

Next, consider the case $\psi = \psi_1 \sq \psi_2$.
If $C \in \sem{\cA}(S)$, then there is an accepting run $\rho$ of the form $\rho: \rho_1 \cdot (q_1^f,\epsilon,q_2^i) \cdot \rho_2$ and, because of the construction, $C_1 = \rmatch(\rho_1) \in \sem{\cA_1}(S)$ and $C_2 = \rmatch(\rho_2) \in \sem{\cA_2}(S_j)$, with $j = \max(C_1) + 1$.
Then by induction hypothesis there are valuations $\nu_1$ and $\nu_2$ such that $C_1 \in \sem{\psi_1}(S,\nu_1)$, $C_2 \in \sem{\psi_2}(S,\nu_2)$.
Moreover, because $\varphi$ is safe, we know that $\vdefplus(\psi_1) \cap \vdefplus(\psi_2) = \emptyset$.
Therefore, we can define $\nu$ such that $\nu(x) = \nu_1(x)$ if $x \in \vdefplus(\psi_1)$ and $\nu(x) = \nu_2(x)$ if $x \in \vdefplus(\psi_2)$.
Clearly, because $\nu$ represents both $\nu_1$ and $\nu_2$, it holds that $C \in \sem{\psi}(S,\nu)$.
For the other direction, consider a complex event $C$ such that $C \in \sem{\psi}(S,\nu)$ for some valuation $\nu$.
Then there exist complex events $C_1$ and $C_2$ such that $C_1 \in \sem{\psi_1}(S,\nu)$, $C_2 \in \sem{\psi_2}(S,\nu)$ and $C = C_1 \cdot C_2$.
By induction hypothesis, there exist an accepting run $\rho_1$ of $\cA_1$ over $S$ such that $\rmatch(\rho_1) = C_1$.
Similarly, there exist an accepting run $\rho_2$ of $\cA_2$ over $S_j$ with $j = \max(C_1)+1$ such that $\rmatch(\rho_2) = C_2$.
Then, the run of $\cA$ that simulates $\rho_1$ ends at a state $q_1^f$, thus it can continue by simulating $\rho_2$ and reaching a final state.
Therefore, such run $\rho$ is an accepting run of $\cA$.
Notice that $\rmatch(\rho) = C$, thus $C \in \sem{\cA}(S)$.

Finally, consider the case $\psi = \psi_1 \ks$.
If $C \in \sem{\cA}(S)$, it means that there is an accepting run $\rho$ of $\cA$ over $S$.
We define $k$ to be the number of times that $\rho$ passes through the final state $q^f$, and prove by induction over $k$ that $C \in \sem{\psi}(S,\nu)$.
If $k=1$, it means that $\rho$ is also an accepting run of $\cA_1$, thus $C \in \sem{\cA_1}(S)$ and, by (the first) induction hypothesis, there exists some valuation $\nu$ such that $C \in \sem{\psi_1}(S,\nu)$, which implies $C \in \sem{\psi}(S,\nu)$.
Now, consider the case $k > 1$.
It means that $\rho$ has the form $\rho = \rho_1 \cdot (q^f, \epsilon, q^i) \cdot \rho_2$ where $\rho_2$ passes through $q^f$ $k-1$ times.
Then, $C_1 = \rmatch(\rho_1)$ is an accepting run of $\cA_1$, hence $C_1 \in \sem{\psi_1}(S,\nu)$ for some $\nu$.
Furthermore, $\rho_2$ is an accepting run of $\cA$, thus if $C_2 = \rmatch(\rho_2)$ then by induction hypothesis $C_2 \in \sem{\psi}(S_j,\nu)$ for some $\nu$, where $j = \max(C_1)+1$.
If $C = C_1 \cdot C_2$ then $C \in \sem{\psi_1\sq\psi_1\ks}(S, \nu)$, thus $C \in \sem{\psi}(S, \nu)$.
Note that we do not care about $\nu$ because the $\ks$-operator of $\psi$ overwrites it.
For the other direction, consider a complex event $C$ such that $C \in \sem{\psi}(S,\nu)$ for some valuation $\nu$.
Then there exists $\nu'$ such that either $C \in \sem{\psi_1}(S,\nu[\nu' \rightarrow U])$ or $C \in \sem{\psi_1\sq\psi_1\ks}(S,\nu[\nu' \rightarrow U])$ where $U = \vdefplus(\psi_1)$.
We now prove, by induction over the number of iterations, that $C \in \sem{\cA}(S)$.
If there is just one iteration, then $C \in \sem{\psi_1}(S,\nu[\nu' \rightarrow U])$ and, by induction hypothesis, $C \in \sem{\cA_1}(S)$, so there is an accepting run $\rho$ of $\cA_1$ over $S$ such that $\rmatch(\rho) = C$.
Because $\Delta_1 \subseteq \Delta$, then $\rho$ is also an accepting run of $\cA$, thus $C \in \sem{\cA}(S)$.
If there are $k$ iterations with $k > 1$, it means that $C \in \sem{\psi_1\sq\psi_1\ks}(S,\nu[\nu' \rightarrow U])$.
Therefore, there exist complex events $C_1$ and $C_2$ such that $C = C_1 \cdot C_2$, $C_1 \in \sem{\psi_1}(S,\sigma[\sigma' \rightarrow U])$ and $C_2 \in \sem{\psi_1 \ks}(S_j,\nu[\nu' \rightarrow U])$, where $j = \max(C_1)+1$.
Then, by induction hypothesis, there exist accepting runs $\rho_1$ of $\cA_1$ over $S$ and $\rho_2$ of $\cA$ over $S_j$ such that $\rmatch(\rho_1) = C_1$ and $\rmatch(\rho_2) = C_2$ and, because $\Delta_1 \subseteq \Delta$, $\rho_1$ is also an accepting run of $\cA$.
Then, the run $\rho = \rho_1 \cdot (q^f, \epsilon, q^i) \cdot \rho_2$ is an accepting run of $\cA$ over $S$.
Furthermore, $\rmatch(\rho) = C_1 \cdot C_2 = C$ thus $C \in \sem{\cA}(S)$.

Finally, it is clear that the size of $\cA$ is linear with respect to the size of $\varphi$ if $\varphi$ is already safe and in LP-normal form.
As stated at the beginning, if $\varphi$ is not safe and/or in LP-normal form, it first has to be turned into an equivalent $\psi$ that is, and such that $|\psi| = \cO(\exp^2(|\varphi|))$ in the worst-case scenario, where $\exp(x) = 2^x$.
Then, $|\cA| = \cO(|\varphi|)$ if $\varphi$ is safe and in LP-normal form, and $|\cA| = \cO(|\psi|) = \cO(\exp^2(|\varphi|))$ otherwise.
\hfill \qed
\newpage

\subsection{Proof of Theorem~\ref{theo:selectors-compilation}}\label{proof:selectors}

\subsubsection{$\sSTRICT$ operator}

Consider a CEA $\cA = (Q,\Delta,I,F)$.
We will first define a CEA $\cA_\sSTRICT =(Q_\sSTRICT,\Delta_\sSTRICT,I_\sSTRICT,F_\sSTRICT)$ and then prove that it is equivalent to $\sSTRICT(\cA)$.
The set of states is defined as $Q_\sSTRICT= \{q^m \mid q \in Q \text{ and } m \in \{\amark,\umark\}\}$, the transition relation is $\Delta_\sSTRICT = \{(p^m,(P,m),q^m) \mid (p,(P,m),q) \in \Delta)\} \cup \{(p^\umark,(P,\amark),q^\amark) \mid (p,(P,\amark),q) \in \Delta\}$, the initial states are $I_\sSTRICT = \{q^\umark \mid q \in I\}$ and the final states are $F_\sSTRICT = \{q^\amark \mid q \in F\}$.
Basically, there are two copies of $\cA$, the first one which only have the $\umark$ transitions, and the second one which only have the $\amark$ ones, and at any $\amark$ transition it can move from the first on to the second.
On an execution, $\cA_\sSTRICT$ starts in the first copy of $\cA$, moving only through transitions that do not mark the positions, until it decides to mark one.
At that point it moves to the second copy of $\cA$, and from there on it moves only using transitions with $\amark$ until it reaches an accepting state.

Now, we prove that the construction is correct, that is, $\sem{\cA_\sSTRICT}(S) = \sem{\sSTRICT(\cA)}(S)$ for every $S$.
Let $S$ be any stream.
First, consider a complex event $C \in \sem{\sSTRICT(\cA)}(S)$.
This means that $C \in \sem{\cA}(S)$ and that $C$ has the form $C = \{m_0,m_1, \ldots,m_k\}$ with $m_i = m_{i-1}+1$.
Therefore, there is an accepting run of $\cA$ of the form:
$$
\rho: q_0 \ \trans{P_1 / \umark} \ q_1 \  \trans{P_2 / \umark} \ \cdots \ \trans{P_{m_1-1} / \umark} q_{m_1-1} \ \trans{P_{m_1} / \amark} \ q_{m_1} \  \trans{P_{m_2} / \amark} \ \cdots \ \trans{P_{m_k} / \amark} \ q_{m_k}
$$
Such that $\rmatch(\rho)=C$.
Consider now the run over $\cA_\sSTRICT$ of the form:
$$
\rho': q_0^\umark \ \trans{P_1 / \umark} \ q_1^\umark \  \trans{P_2 / \umark} \ \cdots \ \trans{P_{m_1-1} / \umark} q_{m_1-1}^\umark \ \trans{P_{m_1} / \amark} \ q_{m_1}^\amark \  \trans{P_{m_2} / \amark} \ \cdots \ \trans{P_{m_k} / \amark} \ q_{m_k}^\amark
$$
It is clear that all transitions of $\rho'$ are in $\Delta_\sSTRICT$, because the ones with $\umark$ are in the first copy of $\cA$, the first one with $\amark$ passes from the first copy to the second, and the following ones with $\amark$ are in the second copy.
Therefore $\rho'$ is indeed run of $\cA_\sSTRICT$ over $S$, and because $q_{m_k} \in F$, then $q_{m_k}^\amark \in F$ and $\rho'$ is an accepting run.
Moreover, $\rmatch(\rho')=C$, thus $C \in \sem{A_\sSTRICT}(S)$.

Now, consider a complex event $C \in \sem{A_\sSTRICT}(S)$, of the form $C = \{m_0,m_1, \ldots,m_k\}$.
It means that there is an accepting run of $\cA_\sSTRICT$ of the form:
$$
\rho: q_0^\umark \ \trans{P_1 / \umark} \ \cdots \ \trans{P_{m_1-1} / \umark} q_{m_1-1}^\umark \ \trans{P_{m_1} / \amark} \ q_{m_1}^\amark \  \trans{P_{m_2} / \amark} \ \cdots \ \trans{P_{m_k} / \amark} \ q_{m_k}^\amark
$$
Such that $\rmatch(\rho)=C$.
Notice that $\rho$ must have this form because of the structure of $\cA_\sSTRICT$, which force $\rho$ to have $\umark$ transitions at the beginning and $\amark$ ones at the end.
Consider then the run of $\cA$ of the form:
$$
\rho': q_0 \ \trans{P_1 / \umark} \ \cdots \ \trans{P_{m_1-1} / \umark} q_{m_1-1} \ \trans{P_{m_1} / \amark} \ q_{m_1} \  \trans{P_{m_2} / \amark} \ \cdots \ \trans{P_{m_k} / \amark} \ q_{m_k}
$$
Similar to the converse case, it is clear that all transitions in $\rho'$ are in $\Delta$.
Therefore $\rho'$ is an accepting run of $\cA$ over $S$, and because $\rmatch(\rho')=C$, it holds that $C \in \sem{\sSTRICT(\cA)}(S)$.

Finally, notice that $\cA_\sSTRICT$ consists in duplicating $\cA$, thus the size of $\cA_\sSTRICT$ is two times the size of $\cA$.
\hfill \qed

\subsubsection{$\sNEXT$ operator}\label{proof:next}

Let $\cR$ be a schema and $\cA = (Q, \Delta, I, F)$ be a CEA over $\cR$.
In order to define the new CEA $\cA_\sNEXT = (Q_\sNEXT, \Delta_\sNEXT, I_\sNEXT, F_\sNEXT)$ we first need to introduce some notation.
We begin by imposing an arbitrary linear order $<$ between the states of $Q$, i.e., for every two different states $p,q \in Q$, either $p < q$ or $q < p$.
Let $T_1 \ldots T_k$ be a sequence of sets of states such that $T_i \subseteq Q$. 
We say that a sequence $T_1 \ldots T_k$ is a \emph{total preorder} over $Q$ if $T_i \cap T_j = \emptyset$ for every $i \neq j$. Notice that the sequence is not necessarily a partition, i.e., it does not need to include all states of $Q$.
A total preorder naturally defines a preorder between states where ``$p$ is less than $q$'' whenever $p \in T_i$, $q \in T_j$, and $i < j$. 
To simplify notation, we define the concatenation between set of states such that $T \cdot T' = T T'$ whenever $T$ and $T'$ are non-empty and $T \cdot T' = T \cup T'$ otherwise.
The concatenation between sets will help to remove empty sets during the final construction. 
Now, given any sequence $T_1 \ldots T_k$ (not necessarily a total preorder), one can convert $T_1 \ldots T_k$ into a total preorder by applying the operation \textit{Total Pre-Ordering} ($\TPO$) defined as follows:
\[
\TPO(T_1 \ldots T_k) = U_1 \cdot \ldots \cdot U_k \ \ \text{ where }  \ 
U_i \;\; = \;\; T_i - \bigcup_{j=1}^{i-1} T_j.
\]
Let $\cP =\{P_1,P_2,\ldots,P_n\}$ be the set of all predicates in the transitions of $\Delta$.
Define the equivalence relation $=_\cP$ between tuples such that, for every pair of tuples $t_1$ and $t_2$, $t_1=_\cP t_2$ holds if, and only if, both satisfy the same predicates, i.e., $t_1 \in P_i$ holds iff $t_2 \in P_i$ holds, for every $i$.
Moreover, for every tuple $t$ let $[t]_\cP$ represent the equivalence class of $t$ defined by $=_\cP$, that is, $[t]_\cP = \{t' \mid t =_\cP t'\}$.
Notice that, even though there are infinitely many tuples, there is a finite amount of equivalence classes which is bounded by all possible combinations of predicates in $\cP$, i.e., $2^{|\cP|}$.
Now, for every $t$, define the predicate:
\[
P_t = (\bigwedge\limits_{t \in P_i} P_i) \wedge (\bigwedge\limits_{t \notin P_i} \lnot P_i)
\]
and define the new set of predicates $\types{\cP} = \{P_t \mid t \in \tuples(\cR)\}$.
Notice that for every tuple $t$ there is exactly one predicate in $\types{\cP}$ that is satisfied by $t$, and that predicate is precisely $P_t$.
Finally, we extend the transition relation $\Delta$ as a function such that: 
$$\Delta(T, P, m) = \{q \in Q \mid \text{exist $p \in T$ and $P' \in \cP$ such that $P \subseteq P'$ and $(p, (P', m), q) \in \Delta$}\}
$$
for every $T \subseteq Q$, $P \in \types{\cP}$, and $m \in \{\amark, \umark\}$.

In the sequel, we define the CEA $\cA_\sNEXT = (Q_\sNEXT, \Delta_\sNEXT, I_\sNEXT, F_\sNEXT)$ component by component. 
First, the set of states $Q_\sNEXT$ is defined as
\[
Q_\sNEXT \; = \; \{(T_1 \ldots T_k, p) \ \mid\ \text{$T_1 \ldots T_k$ is a total preorder over $Q$} \text{ and } \text{$p \in T_i$ for some $i \leq k$}  \}
\]
Intuitively, the state $p$ is the current state of the `simulation' of $\cA$ and the sets $T_1\ldots T_k$ contain the states in which the automaton could be, considering the prefix of the word read until the current moment. Furthermore, the sets are ordered consistently with respect to $\leqnext$, e.g., if a run $\rho_1$ reach the state $(\{1,2\}\{3\},1)$ and other run $\rho_2$ reach the state $(\{1,2\}\{3\},3)$, then $\rmatch(\rho_2)\lnext\rmatch(\rho_1)$. This property is proven later in Lemma \ref{lemnext}.

Secondly, the transition relation is defined as follows. 
Consider $P \in \types{\cP}$, $m \in \{\amark, \umark\}$ and $(\cT, p), (\cU, q) \in Q_\sNEXT$ where $\cT=T_1 \ldots T_k$ and $p \in T_i$ for some $i \leq k$. 
Then we have that $((\cT, p), P, m, (\cU, q)) \in \Delta_\sNEXT$ if, and only if,
\begin{enumerate}
	\item $(p, P',m,q) \in \Delta$ for some $P'$ such that $P \subseteq P'$,
	\item  $q \notin \Delta(T_j, P, m')$ for every $m' \in \{\amark, \umark\}$ and $j < i$, 
	\item $\cU = \TPO(U_1^\amark \cdot U_1^\umark \cdot \ldots \cdot U_k^\amark \cdot U_k^\umark)$ where $U_j^\amark = \Delta(T_j,P, \amark)$ and $U_j^\umark = \Delta(T_j, P, \umark)$ for $1 \leq j \leq k$,
	\item $q \notin \Delta(T_i, P, \amark)$ when $m = \umark$, and
	\item $(p', P', m, q) \notin \Delta$ for every $p' \in T_i$ such that $p'<p$ and every $P'$ such that $P \subseteq P'$.
\end{enumerate}
Intuitively, the first condition ensures that the `simulation' respects the transitions of $\Delta$, the second checks that the next state could not have been reached from a `higher' run, the third ensures that the sequence is updated correctly and the fourth restricts that if the next state can be reached either marking the letter or not, it always choose to mark it.
The last condition is not strictly necessary, and removing it will not change the semantics of the automaton, but is needed to ensure that there are no two runs $\rho_1$ and $\rho_2$ that end in the same state such that $\rmatch(\rho_1) = \rmatch(\rho_2)$.

Finally, the initial set $I_\sNEXT$ is defined as all states of the form $(I, q)$ where $q \in I$ and the final set $F_\sNEXT$ as all states of the form $(T_1 \ldots T_k, p)$ such that $p \in F$ and  there exists $i \leq k$ such that $p \in T_i$ and $T_j \cap F = \emptyset$ for all $j < i$.

Let $S = t_1 t_2 \ldots$ be any stream.
To prove that the construction is correct, we will need the following lemma.

\begin{lemma}\label{lemnext}
	Consider a CEA $\cA=(Q,\Delta,I,F)$, a stream $S$, two states $(\cT,p),(\cT,q)\in Q_\sNEXT$ with the same sequence $\cT=T_1\ldots T_k$ such that $p\in T_i$, $q \in T_j$ for some $i$ and $j$, and two runs $\rho_1,\rho_2$ of $\cA_\sNEXT$ over $S$ that have the same length and reach the states $(\cT,p)$ and $(\cT,q)$, respectively. Then, $i<j$ if, and only if:
	\[
	\rmatch(\rho_2) \lnext \rmatch(\rho_1)
	\]
\end{lemma}

\begin{proof}
	We will prove it by induction over the length of the runs.
	Let $q_0,q_0'\in I$ be any two initial states of $\cA$, not necessarily different.
	First, assume that both runs consist of reading a single tuple $t$.
	Then, the runs are of the form:
	\[
	\rho_1: (I,q_0) \ \trans{P_t / m_1} \ (\cT,p)
	\qquad\text{and}\qquad
	\rho_2: (I,q_0') \ \trans{P_t / m_2} \ (\cT,q)
	\]
	where $\cT=T_1 T_2=\TPO(\Delta(I,P_t,\amark) \Delta(I,P_t,\umark))$ and neither $T_1$ nor $T_2$ can be empty because $p$ and $q$ are in different sets.
	For the if direction, the only option is that $\rmatch(\rho_1)=\{1\}$ and $\rmatch(\rho_2)=\{\}$, which implies that $m_1=\amark$ and $m_2=\umark$.
	Then $i < j$ because $p \in T_1$ and $q \in T_2$.
	For the only-if direction, because $i < j$ then $p\in T_1$ and $q \in T_2$, so necessarily $m_1=\amark$ and $m_2=\umark$.
	Because of this, $\rmatch(\rho_1)=\{1\}$ and $\rmatch(\rho_2)=\{\}$, therefore $\rmatch(\rho_2) \lnext \rmatch(\rho_1)$.
	Now, let $S = t_1 t_2 \ldots t_n \ldots$ and consider that the runs are of the form:
	\begin{align*}
		\rho_1: (I,q_0) \ \trans{P_{t_1} / m_1} \ (\cT_1,q_1) \  \trans{P_{t_2} / m_2} \ \cdots \ \trans{P_{t_{n-1}} / m_{n-1}} (\cT_{n-1},q_{n-1}) \trans{P_{t_n} / m_n} (\cT,p)\\
		\rho_2: (I,q_0') \ \trans{P_{t_1} / m_1'} \ (\cT_1,q_1') \  \trans{P_{t_2} / m_2'} \ \cdots \ \trans{P_{t_{n-1}} / m_{n-1}'} (\cT_{n-1},q_{n-1}') \trans{P_{t_n} / m_n'} (\cT,q)
	\end{align*}
	Notice that both runs have the same sequences $\cT_1,\ldots,\cT_{n-1}$ because each sequence $\cT_i$ is defined only by the previous sequence $\cT_{i-1}$ and the tuple $t_i$ which implicitly defines the predicate $P_{t_i}$.
	Furthermore, all the runs over the same word must have the same sequences.
	Define the runs $\rho_1'$ and $\rho_2'$, respectively, as the runs $\rho_1$ and $\rho_2$ without the last transition.
	Consider that $\cT_{n-1}$ has the form $\cT_{n-1}=U_1 U_2\ldots U_k$, and that $q_{n-1}\in U_r$ and $q_{n-1}'\in U_s$ for some $r$ and $s$.
	Notice that, because of the construction, if it is the case that $r<s$ ($r>s$), then $i<j$ ($i>j$ resp.) must hold.
	For the if direction, consider that $\rmatch(\rho_2) \lnext \rmatch(\rho_1)$.
	If $\rmatch(\rho_1') = \rmatch(\rho_2')$, by induction hypothesis it means that $r=s$.
	Moreover, the only option is that $m_n = \amark$ and $m_n' = \umark$, therefore, by the construction it holds that $i < j$.
	If $\rmatch(\rho_2') \lnext \rmatch(\rho_1')$, by induction hypothesis it means that $r<s$ and because of the construction, $i<j$.
	Notice that $\rmatch(\rho_1') \lnext \rmatch(\rho_2')$ cannot occur because the lower element of $\rmatch(\rho_2')$ not in $\rmatch(\rho_1')$ would still be the lower element of $\rmatch(\rho_2)$ not in $\rmatch(\rho_1)$, thus contradicting $\rmatch(\rho_2) \lnext \rmatch(\rho_1)$.
	For the only-if direction, consider that $i<j$.
	It is easy to see that, if $r > s$, then $i$ cannot be lower than $j$, thus we do not consider this case.
	Now, consider the case that $r=s$.
	Because $i < j$, it must occur that $m_n=\amark$ and $m_n'=\umark$, so $\rmatch(\rho_1) = \rmatch(\rho_1') \cup \{n\}$ and $\rmatch(\rho_2)=\rmatch(\rho_2')$.
	By induction hypothesis, $\rmatch(\rho_1') = \rmatch(\rho_2')$, therefore $\rmatch(\rho_2) \lnext \rmatch(\rho_1)$.
	Consider now the case that $r < s$. By induction hypothesis, $\rmatch(\rho_2') \lnext \rmatch(\rho_1')$ and, because the last transition can only add $n$ to both complex events, it follows that $\rmatch(\rho_2) \lnext \rmatch(\rho_1)$.		 
\end{proof}

Now, we need to prove that if $C\in \sem{\sNEXT(\cA)}(S)$, then $C\in \sem{\cA_\sNEXT}(S)$ and vice versa. First, consider a complex event $C \in \sem{\cA_\sNEXT}(S)$.
To prove that $C\in \sem{\sNEXT(\cA)}(S)$, we need to show that $C\in\sem{\cA}(S)$ and that for all complex events $C'$ such that $C\leq_\sNEXT C'$ and $\max(C) = \max(C')$, $C'\notin\sem{\cA}(S)$.
Assume that the run associated to $C$ is:
$$
\rho: (\cU_0,q_0) \ \trans{P_{t_1} / m_1} \ (\cU_1,q_1) \  \trans{P_{t_2} / m_2} \ \cdots \ \trans{P_{t_n} / m_n} (\cU_n,q_n)
$$
Because of the construction of $\Delta$ (in particular, the first condition), for every $i$ it holds that $(q_{i-1},P_i,m_i,q_i)\in\Delta$ for some $P_i$ such that $P_{t_i} \subseteq P_i$.
Because $t_i \in P_{t_i}$, then $t_i \in P_i$, thus the run:
$$
\rho': q_0 \ \trans{P_1 / m_1} \ q_1 \  \trans{P_2 / m_2} \ \cdots \ \trans{P_n / m_n} q_n
$$
is an accepting run of $\cA$ over $S$, and thus $C\in \sem{\cA}(S)$. Now, recall from construction of $F_\sNEXT$ that there exists $i \leq k$ such that $q_n \in T_i$ and $T_j \cap F = \emptyset$ for all $j < i$, where $T_1\ldots T_k=\cU_n$. Then, because of Lemma \ref{lemnext}, $C'\lnext C$ for every other $C'\in \sem{\cA}(S)$ such that $\max(C) = \max(C')$, otherwise the run of $C'$ would end in a state inside a $T_j$ such that $j < i$ which cannot happen.
Therefore, $C\in \sem{\sNEXT(\cA)}(S)$.

Now, consider a complex event $C \in \sem{\sNEXT(\cA)}(S)$.
Assume that the run associated to $C$ is:
$$
\rho: q_0 \ \trans{P_1 / m_1} \ q_1 \  \trans{P_2 / m_2} \ \cdots \ \trans{P_n / m_n} q_n
$$
To prove that $C \in \sem{\cA_\sNEXT}(S)$ we will prove that there exists an accepting run on $\cA_\sNEXT$.
Based on $\rho$, consider now the run:
$$
\rho': (\cU_0,p_0) \ \trans{P_{t_1} / m_1} \ (\cU_1,p_1) \  \trans{P_{t_2} / m_2} \ \cdots \ \trans{P_{t_n} / m_n} (\cU_n,p_n)
$$
Where the complex events $m_1, \ldots, m_n$ are the same, each condition $P_{t_i}$ is defined by $t_i$ and each $\cU_i$ is the result of applying the function $\TPO$ based on $\cU_{i-1}$ and $P_{t_i}$
Moreover, each $p_i$ is defined as follows.
As notation, consider that $\cU_i=T^i_1\ldots T^i_{k_i}$ and that every $q_i$ is in the $r_i$-th set of $\cU_i$, i.e., $q_i\in T^i_{r_i}$.
Then, $p_i$ is the lower state in $T^i_{r_i}$ such that $(p_i, (P_{t_{i+1}}, m_{i+1}),p_{i+1}) \in \Delta$, and $p_n = q_n$.
Notice that $\rho'$ is completely defined by $\rho$ and $S$.
We will prove that $\rho'$ is an accepting run by checking that all transitions meet the conditions of the transition relation $\Delta_\sNEXT$.
Now, it is clear that the first condition is satisfied by all transitions, i.e., for every $i$ it holds that $(p_{i-1},(P',m_i),p_i)\in \Delta$ for some $P'$ such that $P_{t_i} \subseteq P'$ (just consider $P' = P_i$).
For the second condition, by contradiction suppose that it is not satisfied by $\rho'$.
It means that for some $i$, $p_i \in \Delta(T^{i-1}_j, P_i, m')$ for some $m' \in \{\amark, \umark\}$ and $j < r_i$.
In particular, consider that the state $p'\in T^{i-1}_j$ is the one for which $(p',a_i,m',p_i) \in \Delta$. Recall that every state inside a sequence is reachable considering the prefix of the word read until that moment.
This means that there exist the accepting runs:
\begin{gather*}
	\sigma: q_0' \ \trans{P_1' / m_1'} \ q_1' \  \trans{P_2' / m_2'} \ \cdots \ \trans{P_{i-1}' / m_{i-1}'} q' \  \trans{P_i' / m_i'} q_i \  \trans{P_{i+1} / m_{i+1}} \ \cdots \ \trans{P_n / m_n} q_n\\
	\sigma': (\cU_0,p_0') \ \trans{P_{t_1} / m_1'} \ (\cU_1,p_1') \  \trans{P_{t_1} / m_2'} \ \cdots \ \trans{P_{t_{i-1}} / m_{i-1}'} (\cU_{i-1},p') \  \trans{P_{t_i} / m_i'} (\cU_i,p_i) \  \trans{P_{t_{i+1}} / m_{i+1}} \ \cdots \ \trans{P_{t_n} / m_n} (\cU_n,p_n)
\end{gather*}
Where $p_i'$ are defined in a similar way to $p_i$.
Define for every run $\gamma$ and every $i$ the run $\gamma_i$ as $\gamma$ until the $i$-th transition.
For example, $\rho_i$ is equal to the run $\rho$ until the state $q_i$.
Then, by Lemma \ref{lemnext}, $\rmatch(\rho'_{i-1}) < \rmatch(\sigma'_{i-1})$, but $\rmatch(\rho') = \rmatch(\sigma')$. This is a contradiction, since $\rmatch(\rho')$ and $\rmatch(\sigma')$ differ from $\rmatch(\rho'_{i-1})$ and $\rmatch(\sigma'_{i-1})$ in that the latters can contain additional positions from $i$ to $n$, but the minimum position remains in $\rmatch(\sigma'_{i-1})$, and therefore in $\rmatch(\sigma')$.
The fourth condition is proven by contradiction too.
Suppose that it is not satisfied by $\rho'$, which means that for some $i$, $p_i \in \Delta(T^{i-1}_{r_{i-1}}, P_{t_i}, \amark)$ when $m_i = \umark$. Then, the run:
$$
\sigma: p_0 \ \trans{P_{t_1} / m_1} \ p_1 \  \trans{P_{t_2} / m_2} \ \cdots \ \trans{P_{t_{i-1}} / m_{i-1}} p_{i-1} \  \trans{P_{t_i} / \amark} p_i \  \trans{P_{t_{i+1}} / m_{i+1}} \ \cdots \ \trans{P_{t_n} / m_n} p_n
$$
is an accepting run such that $\rmatch(\rho) < \rmatch(\sigma)$, which is a contradiction, since $C \in \sem{\sNEXT(\cA)}(S)$.
The third and last conditions are trivially proven because of the construction of the run.
Therefore, $\rho'$ is a valid run of $\cA_\sNEXT$ over $S$.
Moreover, because $p_n = q_n \in F$ then $\rho'$ is an accepting run, therefore $\rmatch(\rho) = C \in \sem{\cA_\sNEXT}(S)$.

Now, we analyze the properties of the automaton $\cA_\sNEXT$.
First, we show that $|\cA_\sNEXT|$ is at most exponential over $|\cA|$.
Notice that each state in $Q_\sNEXT$ represents a sequence of subsets of $Q$, thus each state has at most $|Q|$ subsets.
Moreover, for each one of the subsets there are at most $2^{|Q|}$ possible combinations.
Therefore, there are no more than $2^{|Q|}|Q|$ possible states in $Q_\sNEXT$, thus $|\cA_\sNEXT| \in \cO(2^{|\cA|})$.
\hfill \qed

\subsubsection{$\LAST$ operator}

The $\LAST$ case is done in the same way as the $\sNEXT$ one, with some minor changes.
We now define the CEA $\cA_\LAST = (Q_\LAST, \Delta_\LAST, I_\LAST, F_\LAST)$ component by component. 
First, the set of states $Q_\LAST$ is defined exactly like $Q_\sNEXT$:
\[
Q_\LAST \; = \; \{(T_1 \ldots T_k, p) \ \mid\ \text{$T_1 \ldots T_k$ is a total preorder over $Q$} \text{ and } \text{$p \in T_i$ for some $i \leq k$}  \}
\]
The intuition in this case is that the sets will be ordered consistently with respect to $\leqlast$, e.g., if a run $\rho_1$ reach the state $(\{1,2\}\{3\},1)$ and other run $\rho_2$ reach the state $(\{1,2\}\{3\},3)$, then $\rmatch(\rho_2)\llast\rmatch(\rho_1)$.
This property can be proven in the same way as Lemma \ref{lemnext}.

Secondly, the transition relation is defined as follows. 
Consider $P \in \types{\cP}$, $m \in \{\amark, \umark\}$ and $(\cT, p), (\cU, q) \in Q_\LAST$ where $\cT=T_1 \ldots T_k$ and $p \in T_i$ for some $i \leq k$. 
Then we have that $((\cT, p), P, m, (\cU, q)) \in \Delta_\LAST$ if, and only if,
\begin{enumerate}
	\item $(p, P',m,q) \in \Delta$ for some $P'$ such that $P \subseteq P'$,
	\item  $q \notin \Delta(T_j, P, m')$ for every $m' \in \{\amark, \umark\}$ and $j < i$, 
	\item $\cU = \TPO(U_1^\amark \cdot \ldots \cdot U_k^\amark \cdot U_1^\umark \cdot \ldots \cdot U_k^\umark)$ where $U_j^\amark = \Delta(T_j,P, \amark)$ and $U_j^\umark = \Delta(T_j, P, \umark)$ for $1 \leq j \leq k$, and
	\item $q \notin \Delta(T_i, P, \amark)$ when $m = \umark$,
	\item $(p', P', m, q) \notin \Delta$ for every $p' \in T_i$ such that $p'<p$ and every $P'$ such that $P \subseteq P'$.
\end{enumerate}
The only condition that changes with respect to the $\sNEXT$ case is the third one.
In this case, it ensures that the sequence is updated correctly according to the last-order.
In particular, it changes $U_1^\amark \cdot U_1^\umark \cdot \ldots \cdot U_k^\amark \cdot U_k^\umark$ in the $\sNEXT$ case with $U_1^\amark \cdot \cdot \ldots \cdot U_k^\amark \cdot U_1^\umark \cdot \ldots \cdot U_k^\umark$.
Intuitively, in the first one, if the run reads event at position $i$ and adds it to its complex event $C$ (turning into $C\cup \{i\}$), then it ``wins" over the case in which that same run did not mark it.
On the other hand, in the second case, if the run reads event at position $i$ and adds it to its complex event $C$ (turning into $C\cup \{i\}$), then it ``wins" over every other run that did not mark it.

Finally, the initial and final sets are defined like for the $\sNEXT$: $I_\LAST$ is defined as all states of the form $(I, q)$ where $q \in I$ and $F_\LAST$ as all states of the form $(T_1 \ldots T_k, p)$ such that $p \in F$ and  there exists $i \leq k$ such that $p \in T_i$ and $T_j \cap F = \emptyset$ for all $j < i$.

The proof of correctness is a direct replicate of the one for the $\sNEXT$ case, changing only the notation from $\sNEXT$ to $\LAST$.

\subsubsection{$\sMAX$ operator}

Let $\cA = (Q, \Delta, I, F)$ be a CEA.
Similarly to the construction of CEA for the $\sNEXT$, we define the set $\types{\cP}$ such that for every tuple $t$ there is exactly one predicate $P_t$ in $\types{\cP}$ that is satisfied by $t$, and extend the transition relation $\Delta$ as a function $\Delta(T, P, m)$ for every $T \subseteq Q$, $P \in \types{\cP}$, and $m \in \{\amark, \umark\}$.
Further, we overload the notation of $\Delta$ as a function such that $\Delta(T, P) = \Delta(T,P,\amark) \cup \Delta(T,P,\umark)$.

We now define the CEA $\cA_\sMAX = (Q_\sMAX, \Delta_\sMAX, I_\sMAX, F_\sMAX)$ component by component.
First, the set of states is $Q_\MAX = \{(S,T) \mid S,T \subseteq Q,\, S \neq \emptyset \text{ and } S \cap T = \emptyset\}$.
At each $(S,T) \in Q_\MAX$, $S$ will keep track of the states of $\cA$ that are reached by runs that define the same complex event $C$ (and are not in $T$), and $T$ will keep track of the states that are reached by runs that define a complex event $C'$ such that $C \subset C'$.
The transition relation $\Delta_\MAX$ as $\Delta = \Delta_\MAX^\amark \cup \Delta_\MAX^\umark$, with
\begin{gather*}
\Delta_\MAX^\amark \ = \ \{((S_1,T_1),(P,\amark),(S_2,T_2)) \mid T_2 = \Delta(T_1,P,\amark) \ \text{ and } \ S_2 = \Delta(S_1, P,\amark) \setminus T_2\}\\
\Delta_\MAX^\umark \ = \ \{((S_1,T_1),(P,\umark),(S_2,T_2)) \mid T_2 = \Delta(T_1,P) \cup \Delta(S_1,P,\amark) \ \text{ and } S_2 \ = \Delta(S_1,P,\umark) \setminus T_2\}
\end{gather*}
The former updates $T_1$ to $T_2$ using $\amark$-transitions from $T_1$, and $S_1$ to $S_2$ the same way but removing the ones from $T_2$.
The latter updates $T_1$ to $T_2$ using all transitions from $T_1$ plus the $\amark$-transitions from $S_1$, while it updates $S_1$ to $S_2$ using $\umark$-transitions from $S_1$.
Finally, $I_\sMAX = \{(I, \emptyset)\}$, and $F_\sMAX = \{(S, T) \in Q_\sMAX \mid S \cap F \neq \emptyset \text{ and } T \cap F = \emptyset \}$.

Next, we prove the above, i.e., $C \in \sem{\sMAX(\cA)}(S)$ iff $C \in \sem{\cA_\sMAX}(S)$.
To simplify the proof, we assume that $\cA$ is I/O-deterministic, therefore each state of $Q_\MAX$ now has the form $(q,T)$.
The proof can easily be extended for non I/O-deterministic $\cA$.
First, we prove the if direction.
Consider a complex event $C$ such that $C \in \sem{\cA_\sMAX}(S)$.
To prove that $C \in \sem{\sMAX(\cA)}(S)$, we first prove that $C \in \sem{\cA}(S)$ by giving an accepting run of $\cA$ associated to $C$.
Assume that the run of $\cA_\sMAX$ over $S$ associated to $C$ is:
$$
\rho: (q_0,T_0) \ \trans{P_{t_1} / m_1} \ (q_1,T_1) \  \trans{P_{t_2} / m_2} \ \cdots \ \trans{P_{t_n} / m_n} \ (q_n,T_n)
$$
Where $T_0 = \emptyset$, $T_n \cap F = \emptyset$ and $((q_{i-1}, T_{i-1}), (P_{t_i}, m_i), (q_i,T_i)) \in \Delta_\sMAX$.
Furthermore, $q_0 \in I$ and $q_n \in F$.
Also, from the construction of $\Delta_\sMAX$, we deduce that for every $i$ there is a predicate $P_i$ such that $(q_{i-1},(P_i, m_i), q_i) \in \Delta$.
This means that the run:
$$
q_0 \ \trans{P_1 / m_1} \ q_1 \  \trans{P_2 / m_2} \ \cdots \ \trans{P_n / m_n} \ q_n
$$
Is an accepting run of $\cA$ associated to $C$.
Now, we prove by contradiction that for every $C'$ such that $C \subset C'$, $C' \notin \sem{\cA}(S)$.
In order to do this, we define the next lemma, in which we use the notion of \textit{partial run}, which is the same as a run but not necessarily beginning at an initial state.
\begin{lemma}\label{lemmax1}
	Consider an I/O-deterministic CEA $\cA=(Q,\Delta,I,F)$, a stream $S = t_1, t_2, \ldots$ and two partial runs of $\cA_\sMAX$ and $\cA$ over $S$, respectivelly:
	\begin{gather*}
	\sigma: (q_0,T_0) \ \trans{P_{t_1} / m_1} \ (q_1,T_1) \  \trans{P_{t_2} / m_2} \ \cdots \ \trans{P_{t_n} / m_n} \ (q_n,T_n)\\
	\sigma': p_0 \ \trans{P_1 / m_1'} \ p_1 \  \trans{P_2 / m_2'} \ \cdots \ \trans{P_n / m_n'} \ p_n
	\end{gather*}
	Then, if $p_0\in T_0$ and $m_i'=\amark$ at every $i$ for which $m_i=\amark$, it holds that $p_n \in T_n$.
\end{lemma}
\begin{proof}
	This is proved by induction over the length $n$.
	First, if $n = 0$, then $p_n=p_0$ and $T_n = T_0$, so $p_n\in T_n$.
	Now, assume that the lemma holds for $n-1$, i.e., $p_{n-1}\in T_{n-1}$.
	Consider the case that $m_n = \amark$.
	Then $m_n' = \amark$ too, thus $(p_{n-1},(P_n,\amark),p_n)\in \Delta$.
	Furthermore, $T_n = \Delta(T_{n-1},P_{t_n},\amark)$ and therefore $p_n \in T_n$, because $p_{n-1} \in T_{n-1}$.
	Now, consider the case $m_n = \umark$.
	Either $(p_{n-1},(P_n,\amark),p_n) \in \Delta$ or $(p_{n-1},(P_n,\umark),p_n) \in \Delta$, so $p_n \in \Delta(T_{n-1},P_{t_n})$.
	Moreover, $\Delta(T_{n-1},P_{t_n}) \subseteq T_n$ because of the construction of $\Delta_\sMAX$, therefore $p_n \in T_n$.
\end{proof}

Now, by contradiction consider a complex event $C'$ such that $C \subset C'$ and $C'\in \sem{\cA}(S)$. Then, there must exist an accepting run of $\cA$ over $S$ associated to $C'$ of the form:
$$
\rho': p_0 \ \trans{P_1' / m_1'} \ p_1 \  \trans{P_2' / m_2'} \ \cdots \ \trans{P_n' / m_n'} \ p_n
$$
such that $m_i' = \amark$ at every $i$ for which $m_i = \amark$, and there is at least one $i$ for which $m_i = \umark$ and $m_i' = \amark$.
Consider $i$ to be the lower position for which this happens.
Because $\cA$ is I/O-deterministic, $\rho'$ can be rewritten as:
$$
\rho': q_0 \ \trans{P_1 / m_1} \ \cdots \ \trans{P_{i-1} / m_{i-1}} \ q_{i-1} \ \trans{P_i' / \amark} \ p_i \ \trans{P_{i+1}' / m_{i+1}'} \cdots \ \trans{P_n' / m_n'} \ p_n
$$
Similarly, to ease visualization we rewrite $\rho$ as:
$$
\rho: (q_0,T_0) \ \trans{P_{t_1} / m_1} \ \cdots \ \trans{P_{t_{i-1}} / m_{i-1}} \ (q_{i-1},T_{i-1}) \ \trans{P_{t_i} / \umark} \ (q_i,T_i) \ \trans{P_{t_{i+1}} / m_{i+1}} \cdots \ \trans{P_{t_n} / m_n} \ (q_n,T_n)
$$
In particular, the transition $((q_{i-1},T_{i-1}), (P_{t_i}, \umark), (q_i, T_i))$ is in $\Delta_\sMAX$, which means that $\Delta(\{q_{i-1}\},P_{t_i},\amark) \subseteq T_i$.
Moreover, $(q_{i-1},(P_i',\amark),p_i) \in \Delta$ and, because $t_i \in P_i'$, then $P_{t_i} \subseteq P_i'$ thus $p_i \in T_i$.
Now, by Lemma \ref{lemmax1} it follows that $p_n \in T_n$.
But because $\rho$ is an accepting run, we get that $T_n\cap F = \emptyset$ and so $p_n \notin F$, which is a contradiction to the statement that $\rho'$ is an accepting run.
Therefore, for every $C'$ such that $C \subset C'$, $C'\notin \sem{\cA}(S)$, hence $C \in \sem{\sMAX(\cA)}(S)$.

Next, we will prove the only-if direction.
For this, we will need the following lemma:
\begin{lemma} \label{lemmax2}
	Consider an I/O-deterministic CEA $\cA=(Q,\Delta,I,F)$, a stream $S=t_1, t_2, \ldots$, a run of $\cA_\sMAX$ over $S$:
	$$
	\sigma: (q_0,T_0) \ \trans{P_{t_1} / m_1} \ (q_1,T_1) \  \trans{P_{t_2} / m_2} \ \cdots \ \trans{P_{n_2} / m_n} \ (q_n,T_n)
	$$
	And a state $p \in Q$.
	If $p \in T_n$, then there is a run of $\cA$ over $S$:
	$$
	\sigma': p_0 \ \trans{P_1 / m_1'} \ p_1 \  \trans{P_2 / m_2'} \ \cdots \ \trans{P_{n-1} / m_{n-1}'} \ p_{n-1} \ \trans{P_n / m_n'} \ p
	$$
	Such that $\rmatch(\sigma) \subset \rmatch(\sigma')$.
\end{lemma}
\begin{proof}
	It will be proved by induction over the length $n$.
	The base case is $n = 0$, which is trivially true because $T_0=\emptyset$.
	Assume now that the Lemma holds for $n-1$.
	Define the run $\sigma_{n-1}$ as the run $\sigma$ without the last transition.
	For any state $q \in T_{n-1}$, let $\sigma'_q$ be the run that ends in $q$ such that $\rmatch(\sigma_{n-1}) \subset \rmatch(\sigma'_q)$.
	Consider the case $m_n = \umark$.
	Then, either $p \in \Delta(T_{n-1},P_{t_n})$ or $p\in \Delta(\{q_{n-1}\},P_{t_n},\amark)$.
	In the former scenario, there must be a $q \in T_{n-1}$ and $P \in \uset$ such that $(q,(P_n,m),p)\in\Delta$ and $P_{t_n} \subseteq P$, with $m \in \{\amark,\umark\}$.
	Define $\sigma'$ as the run $\sigma'_q$ followed by the transition $(q,(P,m),p)$.
	Then $\sigma'$ satisfies $\rmatch(\sigma) \subset \rmatch(\sigma')$.
	In the latter scenario, there must be an $P \in \uset$ such that $(q_{n-1},(P,\amark),p) \in \Delta$ and $P_{t_n} \subseteq P$.
	Define $\sigma'$ as $\sigma_{n-1}$ followed by the transition $(q_{n-1},(P,\amark),p)$.
	Then $\sigma'$ satisfies $\rmatch(\sigma) \subset \rmatch(\sigma')$.
	Now, consider the case $m_n = \amark$.
	Here, $p$ has to be in $\Delta(T_{n-1},P_{t_n},\amark)$, so there must be a $q \in T_{n-1}$ and $P \in \uset$ such that $(q,(P,\amark),p)\in\Delta$ and $P_{t_n} \subseteq P$.
	Define $\sigma'$ as the run $\sigma'_q$ followed by the transition $(q,(P,\amark),p)$.
	Then $\sigma'$ satisfies $\rmatch(\sigma) \subset \rmatch(\sigma')$.
	Finally, the Lemma holds for every $n$.
\end{proof}
Consider a complex event $C$ such that $C\in \sem{\sMAX(\cA)}(S)$.
This means that there is an accepting run of $\cA$ over $S$ associated to $C$.
Define that run as:
$$
\rho: q_0 \ \trans{P_1 / m_1} \ q_1 \  \trans{P_2 / m_2} \ \cdots \ \trans{P_n / m_n} \ q_n
$$
Where $q_0\in I$, $q_n \in F$ and $(q_{i-1},(P_i,m_i),q_i) \in \Delta$.
To prove that $C \in \sem{\cA_\sMAX}(S)$ we give an accepting run of $\cA_\sMAX$ over $S$ associated to $C$.
Consider the run:
$$
\rho': (q_0,T_0) \ \trans{P_{t_1} / m_1} \ (q_1,T_1) \  \trans{P_{t_2} / m_2} \ \cdots \ \trans{P_{t_n} / m_n} \ (q_n,T_n)
$$
Where $T_0 = \emptyset$, $T_i = \Delta(T_{i-1}, P_{t_i})\cup \Delta(\{q_{i-1}\},P_{t_i},\amark)$ if $m_i=\umark$, and $T_i = \Delta(T_{i-1},P_{t_i},\amark)$ if $m_i=\amark$.
To be a valid run, every transition $(T_{i-1},P_{t_i},m_i,T_i)$ must be in $\Delta_\sMAX$, which we prove now by induction over $i$.
The base case is $i = 0$, which is trivially true because no transition is required to exist.
Next, assume that transitions up to $i-1$ exist.
We know that there is an $P_i$ such that $P_{t_i} \subseteq P_i$ and $(q_{i-1}, (P_i, m_i),q_i)\in \Delta$, so that condition is satisfied.
We only need to prove that $q_i \notin T_i$.
By contradiction, assume that $q_i \in T_i$.
Consider the case that $m_i = \umark$.
It means that either $q_i \in \Delta(\{q_{i-1}\},P_{t_i},\amark)$ or $q_i\in \Delta(T_{i-1},P_{t_i})$.
In the first scenario, consider a new run $\sigma$ to be exactly the same as $\rho$, but changing $m_i$ with $\amark$.
Then $\sigma$ is also an accepting run, and $\rmatch(\rho) \subset \rmatch(\sigma)$, which is a contradiction to the definition of the $\MAX$ semantic.
In the second scenario, there must be some $p \in T_{i-1}$ and $P \in \uset$ such that $(p,(P,m),q_i) \in \Delta$, where $m \in \{\amark, \umark\}$.
Because of Lemma \ref{lemmax2}, it means that there is a run $\sigma'$ over $S$:
$$
\sigma': p_0 \ \trans{P_1' / m_1'} \ p_1 \  \trans{P_2' / m_2'} \ \cdots \ \trans{P_{i-2}' / m_{i-2}'} \ p_{i-2} \ \trans{P_{i-1}' / m_{i-1}'} \ p
$$
Such that $\rmatch(\rho_{i-1}) \subset \rmatch(\sigma')$, where $\rho_{i-1}$ is the run $\rho$ until transition $i-1$.
Moreover, because $(p,(P,m),q_i)\in\Delta$ we can define the run:
$$
\sigma: p_0 \ \trans{P_1' / m_1'} \ p_1 \  \trans{P_2' / m_2'} \ \cdots \ \trans{P_{i-1}' / m_{i-1}'} \ p \ \trans{P / m} \ q_i \ \trans{P_{i+1} / m_{i+1}} \ \cdots \ \trans{P_n / m_n} \ q_n
$$
Such that $\rmatch(\rho) \subset \rmatch(\sigma)$, which is also a contradiction.
Then, $q_i \notin T_i$ for the case $m_i = \umark$.
Now, consider the case $m_i = \amark$.
Assuming that $q_i \in T_i$, it means that $q_i \in \Delta(T_{i-1},P_{t_i},\amark)$.
Then, there must be some $p \in T_{i-1}$ and $P \in \uset$ such that $(p,(P,\amark),q_i)\in \Delta$.
Alike the previous case, because of Lemma \ref{lemmax2}, there is a run:
$$
\sigma: p_0 \ \trans{P_1' / m_1'} \ p_1 \  \trans{P_2' / m_2'} \ \cdots \ \trans{P_{i-1}' / m_{i-1}'} \ p \ \trans{P / \amark} \ q_i \ \trans{P_{i+1} / m_{i+1}} \ \cdots \ \trans{P_n / m_n} \ q_n
$$
Such that $\rmatch(\rho) \subset \rmatch(\sigma)$, which is a contradiction.
Then, $q_i \notin T_i$, therefore $(T_{i-1},(P_{t_i},m_i),T_i) \in \Delta_\sMAX$ for every $i$.
The above proved that $\rho'$ is a run of $\cA_\sMAX$, but to be a accepting run it must hold that $T_n\cap F = \emptyset$.
By contradiction, assume otherwise, i.e., there is some $q \in Q$ such that $q \in T_n \cap F$.
Then, because of Lemma \ref{lemmax2}, there is another accepting run $\sigma$ of $\cA_\sMAX$ over $S$ such that $C \subset \rmatch(\sigma)$, which contradicts the fact that $C$ is maximal.
Thus, $T_n \cap F = \emptyset$ and $\rho'$ is an accepting run, therefore $C \in \sem{\cA_\sMAX}(S)$. Finally, note that $\cA_\sMAX$ is of size exponential in the size of $\cA$, even when $\cA$ is not I/O-deterministic.
\hfill \qed

%% file: appendix-evaluation.tex

\subsection{Enumeration with constant-delay}

\begin{algorithm}[t]
	\caption{Algorithm equivalent to {\sc EnumAll} that runs with constant delay}\label{alg:cea-cdeenum}
	\begin{algorithmic}[1]
		\Require{$n \neq \bot$ and $n.\mlist$ is non-empty.}
		\Procedure{EnumAll*}{$n$}
			\State $n.\mlist.\mbegin$
			\State $s.\mpushblack(n)$
			\While{$s.\mempty = false$}
				\State $n \gets s.\mpop()$
				\If{$n = \bot$}
					\State $n \gets s.\mpopwhites()$
					\State ${\tt Output}(s)$
				\Else
					\State ${\tt Output}(n)$
					\State $n' \gets n.\mlist.\mnext$
					\If{$n.\mlist.\matend = false$}
						\State $s.\mpushwhite(n)$
					\Else
						\State $s.\mpushblack(n)$
					\EndIf
					\State $n'.\mlist.\mbegin$
					\State $s.\mpushblack(n')$
				\EndIf
			\EndWhile
		\EndProcedure
	\end{algorithmic}
\end{algorithm}

We provide the method {\sc EnumAll*} in  Algorithm~\ref{alg:cea-cdeenum} which does the same as {\sc EnumAll} (Algorithm~\ref{alg:cea-enum} in the body of the paper) and runs with constant delay.
Moreover, the algorithm takes constant time between each output event (i.e. position), and constant time between complex events.

We start by explaining the notation in Algorithm~\ref{alg:cea-cdeenum}.
For doing a wise backtracking during the enumeration, we use an extended stack of nodes (denoted by $s$ in the algorithm) that we call a \emph{black-white stack}. 
This stack works as a traditional stack with the difference that stack elements are colored with black and white. 
For coloring the nodes, we provide the methods $\mpushblack(n)$ and $\mpushwhite(n)$ that assign the colored black and white, respectively, when the node $n$ is push into the stack.
This stack also has the traditional method $\mpop()$ and $\mempty$ for popping the top node and checking if the stack is empty, respectively.
The colors are used when the method $\mpopwhites()$ is called. 
When this method is called, the stack pops all the white nodes that are at the top of the stack. 
For example, if $s = \circledblack{1}\circled{2} \circled{3} \circledblack{4} \circled{5} \circled{6} \circled{7}$ is a black-white stack with node $7$ the top of the stack, then when $s.\mpopwhites()$ is called the resulting stack will be $\circledblack{1}\circled{2} \circled{3} \circledblack{4}$, namely, all white nodes at the top of the stack are popped. 
Note that by keeping pointers to the previous black node, each method of a black-white stack can be run in constant time. 

For printing the output, we assume a method $\menum(n)$ which prints the position of the node $n.\mposition$ in the user output tape.
Furthermore, if $s = n_1 \ldots n_{i-1} n_i$ is the current content of a black-white stack with $n_i$ the top of the stack, we assume a method $\menum(s)$ that prints $\# \, n_1.\mposition \, \ldots \, n_{i-1}.\mposition$ in the output (if $s$ is empty or has one node, it does not print any symbol). 
Note that the output will be printed as a sequence $\bar{C}_1 \# \bar{C}_2 \# \ldots \# \bar{C}_k$ where each $\bar{C}_i$ is a sequence of positions (i.e. a complex event) printed in reverse order, namely, if $C = \{i_1, \ldots, i_k\}$ with $i_1 < \ldots < i_k$, then $\bar{C} = i_k \ldots i_1$.
For lists, we assume an extra method $\matend$, which returns true if the iterator of the list is at the end of the list. 
Finally, we assume that the node $\bot$ always has an empty list (i.e. we can apply $\bot.\mlist.\mbegin$ but $\bot.\mlist.\matend$ is always true). 

The intuition behind {\sc EnumAll*} is the following.
For the sake of simplification, we will see the data structure that stores the complex events as an acyclic directed graph, where each node $n$ has edges to the nodes of $n.\mlist$.
Both {\sc EnumAll} and {\sc EnumAll*} are based on the same intuition: to navigate through the graph in a depth-first-search manner and compute a complex event for each path from the root to a leaf (i.e. $\bot$).
The main difference is that, while {\sc EnumAll} does this with recursion and moves one node at a time, {\sc EnumAll*} can move up an arbitrary number of nodes when it acknowledges that there are no more paths (i.e. complex events) at that section of the graph.
This is achieved by the use of the black-white stack and, specifically, in line 7 where the $\mpopwhites$ method is called to backtrack an arbitrary number of nodes in constant time. 
This is particularly useful in cases when, for example, the graph consists of only two disjoint paths that meet at the root.
In this scenario, after enumerating the complex event $C_1$ of the first path, {\sc EnumAll} would have to go back to the root through $|C_1|$ nodes before enumerating the complex event $C_2$ for the second path, thus taking time $\cO(|C_1|)$ between $C_1$ and $C_2$.
On the other side, {\sc EnumAll*} uses the black-white stack $s$ to store the exact point at which it has to go back (in the example, the root node), therefore it takes constant time between each complex event output.
Moreover, to print the partial complex event, {\sc EnumAll*} uses the method ${\tt Output}(s)$ to recap the output from the current position and continue printing from there (line 8).
This way, {\sc EnumAll*} ensures that the time it takes in enumerating between positions or complex events is bounded by a constant.

\subsection{Proof of Theorem~\ref{theo:ndet-evaluation}}

\begin{algorithm}
	\caption{Evaluate non-deterministic $\cA$ over a stream $S$} \label{alg:ndet-eval}
	\begin{algorithmic}[1]
		\Require{A non-deterministic CEA $\cA = (Q,\Delta,I,F)$}
		\Procedure{NDetEvaluate}{$S$}
			\ForAll{$T \in 2^Q \setminus \{I\}$}
				\State $\alist_T \gets \epsilon$
			\EndFor
			\State $\alist_I \gets [\bot]$
			\State $\aactive \gets \{I\}$
			\While {$t \gets \myield_S$}
				\State $\aactive^\aold \gets \aactive.{\tt copy},\; \aactive \gets \emptyset$
				\ForAll{$T \in \aactive^\aold$}
					\State $\alist^\aold_T \gets \alist_T.\mlazycopy, \; \alist_T \gets \epsilon$
				\EndFor
				\ForAll{$T \in \aactive^\aold$}
					\State $U^\amark \gets \Delta(T,t,\amark)$
					\If{$U^\amark \neq \emptyset$}
						\State $\alist_{U^\amark}.\madd(\mNode(t.\mposition,\alist_T^\aold))$
						\State $\aactive \gets \aactive \cup \{U^\amark\}$
					\EndIf
					\State $U^\umark \gets \Delta(T,t,\umark)$
					\If{$U^\umark \neq \emptyset$}
						\State $\alist_{U^\umark}.\mappend(\alist_T^\aold)$
						\State $\aactive \gets \aactive \cup \{U^\umark\}$
					\EndIf
				\EndFor
				\State $\textproc{Enumerate}(\{\alist_T\}_{T \in 2^Q}, \{T \mid T \cap F \neq \emptyset\}, t.\mposition)$
			\EndWhile
		\EndProcedure
	\end{algorithmic}
\end{algorithm}

Here we provide Algorithm~\ref{alg:ndet-eval}, an evaluation algorithm for evaluating an arbitrary CEA $\cA$.
The procedure \textproc{NDetEvaluate} is strongly based on \textproc{Evaluate} from Algorithm~\ref{alg:cea-enum}, modified to do a determinization of $\cA$ ``on the fly''.
It handles subsets of $Q$ as its new states by keeping a list $\alist_T$ for each subset of states $T$ instead of the lists $\alist_q$ for each state $q$.
Moreover, it extends the transition $\Delta$ as a function $\Delta(T,t,m)$ that returns the set of all states reachable from some state in $T \subseteq Q$ after reading event $t$ and marking with $m \in \{\amark,\umark\}$; a similar extension to the one defined for Algorithm~\ref{alg:cea-eval}.
With this modifications, the update of each list $\alist_T$ is done the same way as \textproc{Evaluate}.
To extend it considering $\amark$-transitions when reading $t$, it creates a new node $n^*$ with the current position $t.\mposition$ and linked to the old $\alist_T$; then adds $n^*$ at the top of $U^\amark = \Delta(T,t,\amark)$.
On the other hand, to extend it with $\umark$-transitions, it appends the old list of $T$ to the list of $U^\umark=\Delta(T,t,\umark)$.

Further, a mild optimization is added in Algorithm~\ref{alg:ndet-eval}.
It utilizes a set $\aactive$, which contains the sets $T$ that have non-empty $\alist_T$, avoiding the need to iterate over all subsets of $Q$ when it is not necessary.
However, the exponential update time is still maintained for the worst-case scenario.
It is worth noting that there is an alternative algorithm for evaluating $\cA$ that consists in first determinizing $\cA$ and then running Algorithm~\ref{alg:cea-eval} on the resulting I/O-deterministic CEA $\cA^{\det}$.
This evaluation algorithm updates in time linear to the size $|\cA^{\det}| = \cO(2^{|\cA|})$, resulting in the same update time as Algorithm~\ref{alg:ndet-eval}.

\subsection{Proof of Theorem~\ref{theo:selectors-evaluation}}

Here we provide for each selection strategy $\SEL \in \{\NEXT,\LAST,\STRICT,\MAX\}$ an evaluation algorithm for evaluating $\SEL(\cA)$ for an arbitrary CEA $\cA$. Each algorithm comes from combining the automata constructions of Theorem~\ref{theo:selectors-compilation} and the evaluation algorithm for I/O deterministic CEA. Moreover, each algorithm uses the $\textproc{Enumerate}$ procedure to enumerate all matchings, similar than the algorithm for I/O-deterministic CEA.

\subsubsection{$\sNEXT$ evaluation}
\begin{algorithm}
	\caption{Evaluate $\cA$ over a stream $S$ with $\NEXT$ semantics}\label{alg:next-eval}
	\begin{algorithmic}[1]
		\Require{CEA $\cA = (Q,\Delta,I,F)$}
		\Procedure{NextEvaluate}{$S$}
			\ForAll{$q \in Q \setminus I$}
				\State $\alist_q \gets \epsilon$
			\EndFor
			\ForAll{$q \in I$}
			\State $\alist_q \gets [\bot]$
			\EndFor
			\State $O \gets [I]$
			\While {$t \gets \myield_S$} \label{next-eval-line:yield}
				\ForAll{$q \in Q$}
					\State $\alist^\aold_q \gets \alist_q\!.\mlazycopy, \; \alist_q \gets \epsilon$
				\EndFor
				\State $O^\aold \gets O, \; O \gets []$
				\ForAll{$A \in O^\aold$}
					\State $\textproc{UpdateMarking}(A,t,\amark)$
					\State $\textproc{UpdateMarking}(A,t,\umark)$
				\EndFor
				\State $\textproc{Enumerate}(\{\alist_q\}_{q \in Q},F,t.\mposition)$
			\EndWhile
		\EndProcedure
		
		\Procedure{UpdateMarking}{$A,t,m$}
			\State $B \gets \emptyset$
			\ForAll{$q \in A \textbf{ and } p \in\Delta(q,t,m) \setminus O.{\tt set}$}
				\State $B \gets B \cup \{p\}$
				\If{$m = \amark$}
					\State $\alist_{p}\! \gets [\mNode(t.\mposition,\alist_q^\aold)]$ 
				\Else
					\State $\alist_{p}\! \gets \alist_q^\aold$
				\EndIf
			\EndFor
			\If{$B \neq \emptyset$}
				\State $O.{\tt enqueue}(B)$
			\EndIf
		\EndProcedure
	\end{algorithmic}
\end{algorithm}

An evaluation algorithm for $\NEXT(\cA)$ is given in Algorithm~\ref{alg:next-eval}.
The procedure \textproc{NextEvaluate} uses the same approach as the construction of the CEA $\cA_{\NEXT}$ of Theorem~\ref{theo:selectors-compilation}, which simulated $\cA$ while keeping an order of priority over the states.
This order was used so that $\cA_\NEXT$ could simulate a run $\rho$ of $\cA$ that reaches $q$ only if there was no other simultaneous run $\rho'$ reaching $q$ and such that $\rmatch(\rho) \leqnext \rmatch(\rho')$.
To mimic this behavior, Algorithm~\ref{alg:next-eval} keeps that order in a queue of set of states, called $O$. We assume that $O$ has two methods: ${\tt enqueue}(A)$ to add a set of states $A$ to the queue and ${\tt set}$ to take the union of all set of states inside the queue.
Furthermore, at each update, $\alist_q$ stores at most one node, defined by the first state in the $O$-order that reaches $q$.
This way, when traversing the structure in the $\textproc{Enumerate}$ procedure, the result is at most one complex event for each $\alist_q$ with $q \in F$, which is exactly the maximum complex event in the $\leqnext$ order that reaches $q$.
This, however, could result in giving the same complex event more than once, when it is defined by different runs that end at different states of $F$.
To avoid this issue, one can make sure that $|F| = 1$ by adding a new final state $q_f$ to $\cA$ and adding a transition $(p,P,m,q_f)$ for each $(p,P,m,q)$ that reaches some $q \in F$.

Regarding the update time of Algorithm~\ref{alg:next-eval}, we examine the while iteration of line~\ref{next-eval-line:yield}.
First of all, note that the $O$-queue keeps disjoint set of states and, therefore, its length is bounded by the number of states in $Q$. 
Furthermore, for each set of states $A \in O$ the function {\sc UpdateMarking} iterates over each state in $A$ and each transition $\Delta(q,t,m)$. 
As we said, the sets $A \in O$ are disjoint which implies that each state and transition is checked at most once in {\sc UpdateMarking}, namely, $|\cA|$. 
By using a smart data structure to check membership in $O$ in logarithmic time, the update time of Algorithm~\ref{alg:next-eval} is at most linear in the size of $\cA$. 

\subsubsection{$\sLAST$ evaluation}

\begin{algorithm}
	\caption{Evaluate $\cA$ over a stream $S$ with $\LAST$ semantics}\label{alg:last-eval}
	\begin{algorithmic}[1]
		\Require{CEA $\cA = (Q,\Delta,I,F)$}
		\Procedure{LastEvaluate}{$S$}
			\ForAll{$q \in Q \setminus I$}
				\State $\alist_q \gets \epsilon$
			\EndFor
			\ForAll{$q \in I$}
				\State $\alist_q \gets [\bot]$
			\EndFor
			\State $O \gets [I]$
			\While {$t \gets \myield_S$}
				\ForAll{$q \in Q$}
					\State $\alist^\aold_q \gets \alist_q\!.\mlazycopy, \; \alist_q \gets \epsilon$
				\EndFor
				\State $O^\aold \gets O, \; O \gets []$
				\ForAll{$A \in O^\aold$}
					\State $\textproc{UpdateMarking}(A,t,\amark)$
				\EndFor
				\ForAll{$A \in O^\aold$}
					\State $\textproc{UpdateMarking}(A,t,\umark)$
				\EndFor
				\State $\textproc{Enumerate}(\{\alist_q\}_{q \in Q},F,t.\mposition)$
			\EndWhile
		\EndProcedure
	\end{algorithmic}
\end{algorithm}

Algorithm~\ref{alg:last-eval} is an evaluation algorithm for $\LAST(\cA)$.
One can se the resemblance of procedure \textproc{LastEvaluate} with \textproc{NextEvaluate}.
In fact, both have the same approach: keeping an order $O$ defining how to update the lists.
The difference is that \textproc{LastEvaluate} follows the order of $\cA_\LAST$ of Theorem~\ref{theo:selectors-compilation}, i.e. simulates a run $\rho$ of $\cA$ that reaches $q$ only if there was no other simultaneous run $\rho'$ reaching $q$ and such that $\rmatch(\rho) \leqlast \rmatch(\rho')$.
To achieve this, it prioritizes the updates that add the last position: it iterates over all $\amark$-transitions before all $\umark$-transitions, unlike \textproc{NextEvaluate} which iterates over each $A$ state checking both $\amark$-transitions and $\umark$-transitions from $A$ at the same time.
The same argument about the complexity of \textproc{NextEvaluate} applies to \textproc{LastEvaluate}, thus its update time is also $\cO(|\cA|)$.

\subsubsection{$\sMAX$ evaluation}

\begin{algorithm}
	\caption{Evaluate $\cA$ over a stream $S$ with $\sMAX$ semantics} \label{alg:max-eval}
	\begin{algorithmic}[1]
		\Require{CEA $\cA = (Q,\Delta,I,F)$}
		\Procedure{MaxEvaluate}{$S$}
			\ForAll{$r \in 2^Q \times 2^Q \setminus \{(I,\emptyset)\}$}
				\State $\alist_r \gets \epsilon$
			\EndFor
			\State $\alist_{(I,\emptyset)} \gets [\bot]$
			\State $\aactive \gets \{(I,\emptyset)\}$
			\While {$t \gets \myield_S$}
				\State $\aactive^\aold \gets \aactive\!.{\tt copy}, \; \aactive \gets \emptyset$ \label{max-eval-line:upd1}
				\ForAll{$r \in \aactive^\aold$}
					\State $\alist^\aold_r \gets \alist_r\!.\mlazycopy, \; \alist_r \gets \epsilon$
				\EndFor
				\ForAll{$r \in \aactive^\aold$}
					\State $\textproc{MoveMarking}(r,t)$ \label{max-eval-line:upd2}
					\State $\textproc{MoveNotMarking}(r,t)$ \label{max-eval-line:upd3}
				\EndFor
				\State $\textproc{Enumerate}(\{\alist_r\}_{r \in 2^Q \times 2^Q},\{(T,U) \mid T \cap F \neq \emptyset \land U \cap F = \emptyset \},t.\mposition)$ \label{max-eval-line:upd4}
			\EndWhile
		\EndProcedure
		\Procedure{MoveMarking}{$(T,U),t$}
			\State $U' \gets \Delta(U,t,\amark)$ \label{max-eval-line:mrk1}
			\State $T' \gets \Delta(T,t,\amark) \setminus U'$ \label{max-eval-line:mrk2}
			\If{$T' \neq \emptyset$}
				\State $\alist_{(T',U')}\!.\madd(\mNode(t.\mposition,\alist_{(T,U)}^\aold))$ \label{max-eval-line:mrk3}
				\State $\aactive \gets \aactive \cup \{(T',U')\}$
			\EndIf
		\EndProcedure
		\Procedure{MoveNotMarking}{$(T,U),t$}
			\State $U' \gets \Delta(U,t,\amark) \cup \Delta(U,t,\umark) \cup \Delta(T,t,\amark)$ \label{max-eval-line:nmrk1}
			\State $T' \gets \Delta(T,t,\umark) \setminus U'$  \label{max-eval-line:nmrk2}
			\If{$T' \neq \emptyset$}
				\State $\alist_{(T',U')}\!.\mappend(\alist_{(T,U)}^\aold)$ \label{max-eval-line:nmrk3}
				\State $\aactive \gets \aactive \cup \{(T',U')\}$
			\EndIf
		\EndProcedure
	\end{algorithmic}
\end{algorithm}

The algorithm for evaluating $\MAX(\cA)$ is Algorithm~\ref{alg:max-eval}, which is arguably the most convoluted one so far.
We extend the transition relation $\Delta$ as a function that receives $(T,t,m)$ and returns the set $T'$ such that $q \in T'$ iff there is some $p \in T$ and predicate $P$ such that $(p,P,m,q) \in \Delta$ and $t \in P$.
This extension is analogous to the one in Section~\ref{sec:evaluation-iodet} for $\delta$.

Procedure \textproc{MaxEvaluate} keeps for each pair $(T,U) \in 2^Q \times 2^Q$ a list $\alist_{(T,U)}$.
Similar than for the algorithm for I/O deterministic CEA, here each list keeps the complex event data for a set of runs.
The procedure initializes all lists as empty except for $\alist_{(I,\emptyset)}$, which begins with $\bot$ in it.
At each update (lines~\ref{max-eval-line:upd1}-\ref{max-eval-line:upd4}), it first creates a lazycopy of each list.
Then, each list is updated by procedures \textproc{MoveMarking} and \textproc{MoveNotMarking} (lines~\ref{max-eval-line:upd2} and~\ref{max-eval-line:upd3}).
\textproc{MoveMarking} updates the list with $\amark$ transitions the same way as Algorithm~\ref{alg:cea-eval}, i.e. adding a new node to the target list with the current $t\!.\mposition$, and linking it with the origin list (line~\ref{max-eval-line:mrk3}).
However, it differs in that the origin and target lists are not defined by a transition, e.g. $(p,P,\amark,q)$.
Instead, the origin $\alist_{(T,U)}$ and target $\alist_{(T',U')}$ are bound by the relations (lines~\ref{max-eval-line:mrk1}-\ref{max-eval-line:mrk2}):
\[
(*) \ \left\{ \ 
	\begin{array}{ll}
		U' = \Delta(U,t,\amark) \\
		T' = \Delta(T,t,\amark) \setminus U'
	\end{array}
\right.
\]
Moreover, \textproc{MoveNoteMarking} updates the list with $\umark$ transitions by appending the origin list to the target list (line~\ref{max-eval-line:nmrk3}), as in Algorithm~\ref{alg:cea-eval}.
In this case, the origin $\alist_{(T,U)}$ and target $\alist_{(T',U')}$ are bound by relations (lines~\ref{max-eval-line:nmrk1}-\ref{max-eval-line:nmrk2}):
\[
(**) \ \left\{ \
	\begin{array}{ll}
		U' = \Delta(U,t,\amark) \cup \Delta(U,t,\umark) \cup \Delta(T,t,\amark) \\
		T' = \Delta(T,t,\umark) \setminus U'
	\end{array}
\right.
\]
Both $(*)$ and $(**)$ are motivated by the standard automata determinization: to compress all the runs that define the same output in a single run $\rho$, keeping track of the set of current states $T$ and updating it using the transition relation $\Delta$.
Here we also need to store the set $U$ of states that are reached by runs that define superset complex events of the current one, i.e. the states that can be reached by some simultaneous run $\rho'$ such that $\rmatch(\rho) \subseteq \rmatch(\rho')$.
Because of $(*)$ and $(**)$, the runs represented by $\alist_{(T,U)}$ are the ones that end at some $q \in T$, and if there is other simultaneous run $\rho'$ such chat $\rmatch(\rho) \subseteq \rmatch(\rho')$ then $\rho'$ must end at some state $p \in U$.
This way, in the call $\textproc{Enumerate}$ at line~\ref{max-eval-line:upd4}, we give as final-states argument the pairs $(T,U)$ such that $T$ has an accepting state and $U$ does not, which means that the runs in $\alist_{(T,U)}$ define complex events that are maximal.

As a basic optimization, a set $\aactive$ is stored which keeps the pairs $(T,U)$ with non-empty $\alist_{(T,U)}$, avoiding the need to iterate over all pairs of $2^Q \times 2^Q$ when it is not necessary.
Still, the complexity in the worst-case scenario remains exponential ($\cO(4^{|\cA|})$).

\subsubsection{$\sSTRICT$ evaluation}

\begin{algorithm}
	\caption{Evaluate $\cA$ over a stream $S$ with $\STRICT$ semantics}\label{alg:strict-eval}
	\begin{algorithmic}[1]
		\Require{An I/O deterministic CEA $\cA = (Q,\delta,q_0,F)$}
		\Procedure{StrictEvaluate}{$S$}
		\ForAll{$q \in Q \setminus \{q_0\}$}
		\State $\alist_q \gets \epsilon$
		\EndFor
		\State $\qinit \gets q_0, \; \alist_{\qinit} \gets [\bot]$
		\While {$t \gets \myield_S$}
			\ForAll{$q \in Q$}
				\State $\alist^\aold_q \gets \alist_q\!.\mlazycopy, \; \alist_q \gets \epsilon$
			\EndFor
			\ForAll{$q \in Q \textbf{ with } \alist_q^\aold \neq \epsilon$}
				\If{$p^\amark \gets \delta(q,t,\amark)$}
					\State $\alist_{p^\amark}\!.\madd(\mNode(t.\mposition,\alist_q^\aold))$
				\EndIf
			\EndFor
			\If{$\qinit \gets \delta(\qinit,t,\umark)$} \label{strict-eval-line:nmrk1}
				\State $\alist_{\qinit}\!.\mappend([\bot])$ \label{strict-eval-line:nmrk2}
			\EndIf
			\State $\textproc{Enumerate}(\{\alist_q\}_{q \in Q},F,t.\mposition)$
		\EndWhile
		\EndProcedure
	\end{algorithmic}
\end{algorithm}

An evaluation algorithm for $\STRICT(\cA)$ is given in Algorithm~\ref{alg:strict-eval}.
First, it requires $\cA$ to be deterministic, so for evaluating an arbitrary CEA it first needs to be determinized, incurring in an additional $2^{|\cA|}$ blow-up.

Procedure \textproc{StrictEvaluate} is very similar to \textproc{Evaluate} of Algorithm~\ref{alg:cea-eval}.
The core difference is that it keeps track of a special $\qinit$ state that represents (if it exists) the run done by following only $\umark$ transitions, i.e. the empty run that still have not marked any position (called \emph{empty run}).
At each update, it follows the same idea as \textproc{Evaluate} to update the lists considering $\amark$ transitions.
On the other hand, it does not update the lists in the same way for $\umark$.
This is because it only computes the runs that define continuous intervals of events, therefore no $\umark$ transition can be taken if some event was marked with a $\amark$ transition.
Therefore, it only considers the $\umark$ transitions in the empty run: it updates $\qinit$ with $\delta(\qinit,t,\umark)$ (line~\ref{strict-eval-line:nmrk1}) and adds the $\bot$ node to the new $\alist_{\qinit}$ (line~\ref{strict-eval-line:nmrk2}).
The call of \textproc{Enumerate} is the same as in \textproc{Evaluate}.

As mentioned above, we first need to determinize $\cA$ before running Algorithm~\ref{alg:strict-eval}, which results in a $2^{|\cA|}$ blow-up on the size of the complex event automaton.
Moreover, since the algorithm runs in linear time over the input CEA (by the same arguments as Algorithm~\ref{alg:cea-eval}), the overall update time is $\cO(2^{|\cA|})$.


%% file: appendix-queries.tex
\medskip
 For our experimental evaluation we proposed the queries $Q_1$-$Q_6$ presented in Table~\ref{tab:patterns}. We had to translate these queries to EsperTech and SASE to perform the experiments. Next we present the corresponding translations for the experiments in which we required all complex events (the ones reported in Figure~\ref{fig:exp1}). For performing the experiment using consumption policy (the ones reported in Table~\ref{tab:exp2-proc}) in EsperTech the query is the same, and for Sase we had to modify the code amd call the function \texttt{Initialize} to reset the memory of their automaton whenever a match was produced. For performing the experiment using selection strategies (the ones reported in Figure~\ref{fig:exp-3-throughput}) in EsperTech is suffices to remove the \texttt{all\_matches} clause, and for Sase it suffices to change \texttt{skip-till-any-match} for \texttt{skip-till-next-match}.

\paragraph*{\bf $Q_1=A \as x; B \as y; C \as z$}
\begin{itemize}
\item EsperTech
\begin{verbatim}
    select a, b, c from Stream
    match_recognize (
    measures A as a, B as b, C as c
    all_matches
    pattern (A s* B s* C)
    define 
        A as A.type = “A”,
        B as B.type = “B”,
        C as C.type = “C”
    )
\end{verbatim}
\item SASE
\begin{verbatim}
    PATTERN SEQ(A a, B b, C c)
    WHERE skip-till-any-match
\end{verbatim}
\end{itemize}

\paragraph*{\bf $Q_2=A \as x; B \as y; C \as z; D \as w$}
\begin{itemize}
\item EsperTech
\begin{verbatim}
    select a, b, c, d from Stream
    match_recognize (
    measures A as a, B as b, C as c, D as d
    all_matches
    pattern (A s* B s* C s* D)
    define 
        A as A.type = “A”,
        B as B.type = “B”,
        C as C.type = “C”,
        D as D.type = “D”
    )
\end{verbatim}
\item SASE
\begin{verbatim}
    PATTERN SEQ(A a, B b, C c, D d)
    WHERE skip-till-any-match
\end{verbatim}
\end{itemize}

\paragraph*{\bf $Q_3=((A \as x \cor B \as y) \cor C \as z); D \as w$}
\begin{itemize}
\item EsperTech
\begin{verbatim}
    select a, b, c, d from Stream
    match_recognize (
    measures A as a, B as b, C as c, D as d
    all_matches
    pattern ( ( A | B | C ) s* D)
    define 
        A as A.type = “A”,
        B as B.type = “B”,
        C as C.type = “C”,
        D as D.type = “D”
    )
\end{verbatim}
\item SASE does not support nested disjunction.
\end{itemize}

\paragraph*{\bf $Q_4=(A \as x)+; (B \as y)$}
\begin{itemize}
\item EsperTech
\begin{verbatim}
    select a, b from Stream
    match_recognize (
    measures A as a, B as b 
    all_matches
    pattern (  A ( A | s )* B)
    define 
        A as A.type = “A”,
        B as B.type = “B”
    )
\end{verbatim}
\item SASE
\begin{verbatim}
    PATTERN SEQ(A+ a[], B b)
    WHERE skip-till-any-match
\end{verbatim}
\end{itemize}

\paragraph*{\bf $Q_5=(A \as x)+; (B \as y)+; C \as z$}
\begin{itemize}
\item EsperTech
\begin{verbatim}
    select a, b, c from Stream
    match_recognize (
    measures A as a, B as b, C as c
    all_matches
    pattern (  A ( A | s )* B ( B | s ) * C)
    define 
        A as A.type = “A”,
        B as B.type = “B”,
        C as C.type = “C”
    )
\end{verbatim}
\item SASE
\begin{verbatim}
    PATTERN SEQ(A+ a[], B+ b[], C c)
    WHERE skip-till-any-match
\end{verbatim}
\end{itemize}

\paragraph*{\bf $Q_6=((A \as x)+; (B \as y))+; C \as z$}
\begin{itemize}
\item EsperTech:
\begin{verbatim}
    select a, b, c from Stream
    match_recognize (
    measures A as a, B as b, C as c
    all_matches
    pattern (  A ( A | s )* B ( ( A ( A | s )* B ) | s ) * C)
    define 
        A as A.type = “A”,
        B as B.type = “B”,
        C as C.type = “C”
    )

\end{verbatim}
\item SASE does not support nested iteration.
\end{itemize}